\documentclass[a4paper,english]{llncs}[2022/01/12]

\usepackage[utf8]{inputenc}
\usepackage[english]{babel}
\usepackage[T1]{fontenc}
\usepackage{amsfonts}
\usepackage{amsmath}
\usepackage{cmll}
\usepackage{stmaryrd}
\usepackage{amssymb}
\usepackage{bussproofs}
\usepackage{graphicx}
\usepackage{rotating}
\usepackage{float}
\usepackage{caption}
\usepackage{xcolor}
\usepackage{hyperref}
\usepackage{subcaption}

\usepackage[capitalise,nameinlink]{cleveref}

\crefname{section}{Sect.}{Sect.}
\Crefname{section}{Section}{Sections}
\crefname{listing}{List.}{List.}
\crefname{listing}{Listing}{Listings}
\Crefname{listing}{Listing}{Listings}
\crefname{lstlisting}{Listing}{Listings}
\Crefname{lstlisting}{Listing}{Listings}

\usepackage{xspace}

\usepackage{tikz}
\usepackage{tikz-cd}
\usetikzlibrary{arrows,calc,matrix,shapes}
\usetikzlibrary{positioning,calc,arrows,decorations.pathmorphing,shapes.multipart}
\usetikzlibrary{automata}
\usetikzlibrary{matrix,shapes,fit,decorations.markings}

\date{}
\hypersetup{                    
    colorlinks=true,                
    breaklinks=true,                
    }

\newcommand*{\logsys}[1]{\textnormal{\textsf{#1}}}
\newcommand*{\A}{\Gamma}
\newcommand*{\B}{\Delta}

\def\C{\mathcal{C}} 

\newcommand*{\Var}{\mathcal{V}}
\newcommand*{\Atom}{\mathcal{A}}

\newcommand*{\Nat}{\text{Nat}}
\newcommand*{\coNat}{\text{coNat}}

\newcommand*{\false}{\mathsf{F}}
\newcommand*{\true}{\mathsf{T}}

\newcommand*{\redseq}{\rightarrow}
\newcommand\rmcutpar{\ensuremath{\mathsf{mcut(\iota, \perp\!\!\!\perp)}}\xspace}
\newcommand*{\rmcutparprime}{{\scriptsize\ensuremath{\mathsf{mcut(\iota', \perp\!\!\!\perp')}}}\xspace}
\newcommand*{\rmcutpardouble}{{\scriptsize\ensuremath{\mathsf{mcut(\iota'', \perp\!\!\!\perp'')}}}\xspace}
\newcommand*{\cutrel}{\perp\!\!\!\perp}

\newcommand*{\LK}{\logsys{\ensuremath{\text{LK}}}}
\newcommand*{\LKmod}{\logsys{\ensuremath{\text{LK}_{\Box}}}}

\newcommand*{\LJ}{\logsys{\ensuremath{\text{LJ}}}}
\newcommand*{\LL}{\logsys{\ensuremath{\text{LL}}}}

\newcommand*{\MALL}{\logsys{\ensuremath{\text{MALL}}}}

\newcommand*{\muMALLinf}{\logsys{\ensuremath{\mu \text{MALL}^{\infty}}}}

\newcommand*{\muLLinf}{\logsys{\ensuremath{\mu \text{LL}^{\infty}}}}

\newcommand*{\muLL}{\logsys{\ensuremath{\mu \text{LL}}}}
\newcommand*{\muLKmod}{\logsys{\ensuremath{\mu \text{LK}_{\Box}}}}
\newcommand*{\muLKinf}{\logsys{\ensuremath{\mu \text{LK}^{\infty}}}}

\newcommand*{\muLK}{\logsys{\ensuremath{\mu \text{LK}}}}
\newcommand*{\muLJ}{\logsys{\ensuremath{\mu \text{LJ}}}}

\newcommand*{\muLKmodinf}{\logsys{\ensuremath{\mu \text{LK}_{\Box}^{\infty}}}}
\newcommand*{\muLLmodinf}{\logsys{\ensuremath{\mu \text{LL}_{\Box}^{\infty}}}}
\newcommand*{\muLLmod}{\logsys{\ensuremath{\mu \text{LL}_{\Box}}}}

\newcommand*{\ax}{\text{ax}}
\newcommand*{\exch}{\text{ex}}
\newcommand*{\mcut}{\text{mcut}}
\newcommand*{\cut}{\text{cut}}
\newcommand*{\lkwk}{\text{w}}
\newcommand*{\wk}{\text{w}}

\newcommand*{\lkcontr}{\text{c}}
\newcommand*{\wnwk}{\wn_{\text{w}}}
\newcommand*{\wnde}{\wn_{\text{d}}}
\newcommand*{\wncontr}{\wn_{\text{c}}}
\newcommand*{\wnprom}{\wn_{\text{p}}}
\newcommand*{\ocwk}{\oc_{\text{w}}}
\newcommand*{\ocde}{\oc_{\text{d}}}
\newcommand*{\ocprom}{\oc_{\text{p}}}
\newcommand*{\diacontr}{\lozenge_{\text{c}}}
\newcommand*{\diawk}{\lozenge_{\text{w}}}
\newcommand*{\diaprom}{\lozenge_{\text{p}}}
\newcommand*{\boxcontr}{\Box_{\text{c}}}
\newcommand*{\boxwk}{\Box_{\text{w}}}
\newcommand*{\boxprom}{\Box_{\text{p}}}
\newcommand*{\ocpromloz}{\ocprom^{\lozenge}}
\newcommand*{\wnpromloz}{\wnprom^{\Box}}

\newcommand*{\cntr}{\text{c}}

\newcommand*{\occontr}[1]{\oc_{\text{c}_{#1}}}

\newcommand*{\trans}[1]{{#1}^{\bullet}}

\newcommand*{\sk}[1]{\textsc{SK}(#1)}

\newcommand{\AIC}[1]{\AxiomC{\ensuremath{#1}}}

\newcommand{\UIC}[1]{\UnaryInfC{\ensuremath{#1}}}
\newcommand{\BIC}[1]{\BinaryInfC{\ensuremath{#1}}}
\newcommand{\TIC}[1]{\TrinaryInfC{\ensuremath{#1}}}
\newcommand{\QIC}[1]{\QuaternaryInfC{\ensuremath{#1}}}

\newcommand{\RL}[1]{\RightLabel{\ensuremath{#1}}}
\newcommand{\DP}{\DisplayProof}

\def\proofref#1{
\marginpar{\vspace{-0.2cm}
\colorbox{lightgray}{\begin{minipage}{2cm}
{\scriptsize{#1}}
\end{minipage}}}}
\def\defref#1{
\marginpar{\vspace{-0.2cm}
\colorbox{lightgray}{\begin{minipage}{2cm}
{\scriptsize{#1}}
\end{minipage}}}}

\newtheorem{prop}{Proposition}
\newtheorem{defi}{Definition}
\newtheorem{lem}{Lemma}
\newtheorem{nota}{Notation}
\newtheorem{exa}{Example}
\newtheorem{coro}{Corollary}
\newtheorem{thm}{Theorem}
\newtheorem{rem}{Remark}

\newcommand*{\boxProofs}{\mathcal{C}^\Box}\newcommand*{\ocboxProofs}{\mathcal{C}^{\oc/\Box}}

{}

\newcommand{\orig}[2]{
\begin{tikzpicture}[remember picture]
\node[inner sep=0pt,outer sep=1pt] (#1) {\ensuremath{#2}};\end{tikzpicture}}

\makeatletter
\tikzset{%
  remember picture with id/.style={%
    remember picture,
    overlay,
    save picture id=#1,
  },
  save picture id/.code={%
    \edef\pgf@temp{#1}%
    \immediate\write\pgfutil@auxout{%
      \noexpand\savepointas{\pgf@temp}{\pgfpictureid}}%
  },
  if picture id/.code args={#1#2#3}{%
    \@ifundefined{save@pt@#1}{%
      \pgfkeysalso{#3}%
    }{
      \pgfkeysalso{#2}%
    }
  }
}

\def\savepointas#1#2{%
  \expandafter\gdef\csname save@pt@#1\endcsname{#2}%
}

\def\tmk@labeldef#1,#2\@nil{%
  \def\tmk@label{#1}%
  \def\tmk@def{#2}%
}

\tikzdeclarecoordinatesystem{pic}{%
  \pgfutil@in@,{#1}%
  \ifpgfutil@in@%
    \tmk@labeldef#1\@nil
  \else
    \tmk@labeldef#1,(0pt,0pt)\@nil
  \fi
  \@ifundefined{save@pt@\tmk@label}{%
    \tikz@scan@one@point\pgfutil@firstofone\tmk@def
  }{%
  \pgfsys@getposition{\csname save@pt@\tmk@label\endcsname}\save@orig@pic%
  \pgfsys@getposition{\pgfpictureid}\save@this@pic%
  \pgf@process{\pgfpointorigin\save@this@pic}%
  \pgf@xa=\pgf@x
  \pgf@ya=\pgf@y
  \pgf@process{\pgfpointorigin\save@orig@pic}%
  \advance\pgf@x by -\pgf@xa
  \advance\pgf@y by -\pgf@ya
  }%
}

\newcommand\tikzmark[2][]{%
\tikz[remember picture with id=#2] #1;}
\makeatother

\makeatletter
\tikzset{%
  remember picture with id/.style={%
    remember picture,
    overlay,
    save picture id=#1,
  },
  save picture id/.code={%
    \edef\pgf@temp{#1}%
    \immediate\write\pgfutil@auxout{%
      \noexpand\savepointas{\pgf@temp}{\pgfpictureid}}%
  },
  if picture id/.code args={#1#2#3}{%
    \@ifundefined{save@pt@#1}{%
      \pgfkeysalso{#3}%
    }{
      \pgfkeysalso{#2}%
    }
  }
}

\def\savepointas#1#2{%
  \expandafter\gdef\csname save@pt@#1\endcsname{#2}%
}

\def\tmk@labeldef#1,#2\@nil{%
  \def\tmk@label{#1}%
  \def\tmk@def{#2}%
}

\tikzdeclarecoordinatesystem{pic}{%
  \pgfutil@in@,{#1}%
  \ifpgfutil@in@%
    \tmk@labeldef#1\@nil
  \else
    \tmk@labeldef#1,(0pt,0pt)\@nil
  \fi
  \@ifundefined{save@pt@\tmk@label}{%
    \tikz@scan@one@point\pgfutil@firstofone\tmk@def
  }{%
  \pgfsys@getposition{\csname save@pt@\tmk@label\endcsname}\save@orig@pic%
  \pgfsys@getposition{\pgfpictureid}\save@this@pic%
  \pgf@process{\pgfpointorigin\save@this@pic}%
  \pgf@xa=\pgf@x
  \pgf@ya=\pgf@y
  \pgf@process{\pgfpointorigin\save@orig@pic}%
  \advance\pgf@x by -\pgf@xa
  \advance\pgf@y by -\pgf@ya
  }%
}

\begin{document}

\title{On the cut-elimination of the modal $\mu$-calculus: Linear Logic to the rescue}
\titlerunning{Modal $\mu$-calculus cut-elimination}

 \author{Esaïe B{\scriptsize AUER} \and Alexis S{\scriptsize AURIN}}

 \authorrunning{E. B{\scriptsize AUER} \& A. S{\scriptsize AURIN}}

 \institute{\email{esaie.bauer@irif.fr}\qquad \email{alexis.saurin@irif.fr}\\Université Paris Cité \& CNRS \& INRIA, Pl. Aurélie Nemours,  75013 Paris, France}

%
%

\maketitle

 \begin{abstract}
This paper presents a proof-theoretic analysis of the modal $\mu$-calculus.
More precisely, we prove a syntactic cut-elimination for the non-wellfounded modal $\mu$-calculus, using methods from linear logic and its exponential modalities. 
To achieve this, we introduce a new system, \muLLmodinf{}, which is a linear version of the modal $\mu$-calculus, intertwining the modalities from the modal $\mu$-calculus with the exponential modalities from linear logic. Our strategy for proving cut-elimination involves (i) proving cut-elimination for \muLLmodinf{} and (ii) translating proofs of the modal mu-calculus into this new system via a ``linear translation'', allowing us to extract the cut-elimination result.
\end{abstract}


\section{Introduction}

\paragraph{Eliminability of cuts and the modal $\mu$-calculus.}
Since Kozen's seminal work on the \emph{modal $\mu$-calculus}~\cite{DBLP:journals/tcs/Kozen83}, this logic extending basic modal logic with least and greatest fixed-points has been extremely fruitful for the study of computational systems,
especially reactive systems.
In addition to its wide expressive power, its deep roots in logic also allow for a number of fruitful approaches, be they model-theoretic, proof-theoretic or automata-theoretic.
Still, \emph{cut-elimination} -- a cornerstone of modern proof-theory -- has only received partial solutions~\cite{Bahareh17,BRUNNLER12,Mints12,MintsStuder12,NIWINSKI96}:
\begin{itemize}
\item Either as cut-admissibility statements which are noneffective 
possibly using infinitary proof-systems such as (i) infinitely branching proof systems, allowing the $\omega$-rule~\cite{JAGER2008270} or (ii) non-wellfounded or circular proof systems~\cite{Bahareh17,NIWINSKI96}, allowing proof-trees with infinitely long branches;
\item Or as syntactic cut-elimination results capturing only a {\it fragment} of the calculus in systems with the $\omega$-rule~\cite{Mints12,MintsStuder12}. 
A challenge in describing a \emph{syntactic cut-elimination} for such systems is that the number of applications of a $\mu$-rule must sometimes be determined before knowing how many are needed to match each premises of a $\nu$-rule. In~\cite{BRUNNLER12}, the authors discuss a specific example where syntactic cut-elimination fails.
While there are syntactic cut-elimination results in systems based on the $\omega$-rule~\cite{BRUNNLER12,Mints12,MintsStuder12}, they capture only strict fragments of the modal $\mu$-calculus.
\end{itemize}
In fact, there is no syntactic cut-elimination theorem for the full modal $\mu$-calculus.
The present work establishes such a syntactic cut-elimination theorem for the modal $\mu$-calculus 
in the setting of non-wellfounded sequent calculus. 

\paragraph{On unity and diversity in computational logic.} 
Logic presents at the same time a deep unity and a wide diversity. Miller~\cite{miller-unity-computational-logic} argues that the 
{\it ``universal character [of logic] has been badly fractured in the past few decades''}, due to the wide range of its applications, the various families of logics that have emerged and the different computational tools that are in use, often with little relationship.
Miller thus proposes 
 the following questions as the first of a list of ``challenges'':

\begin{description}
\item[Challenge 1:] {\it Unify a wide range of logical features into a single framework. How best can we explain the many enhancements that have been designed for logic: for example, classical / intuitionistic / linear, fixed points, first-order / higher-order quantification, modalities, and temporal operators? (...)
}
\end{description}

In the present paper,  we partially address Miller's first challenge, providing a common framework for two of the main logics that emerged in the 1980s, Kozen's modal $\mu$-calculus~\cite{DBLP:journals/tcs/Kozen83} and Girard's linear logic (\LL{})~\cite{girard87}.
Working in the setting of circular and non-wellfounded proof systems for the above logics, we propose a so-called {\it linear decomposition} of the modal $\mu$-calculus in linear logic with fixed-points. This proof-theoretic analysis of the modal $\mu$-calculus allows us to do a finer-grained treatment of syntactic cut-elimination.

\paragraph{Cut-admissibility vs. cut-elimination.} 
The treatment of the cut-inference in \linebreak sequent-based proof-systems follows two main traditions: (i) one can consider cut-free proofs as the primitive proof-objects and establish that the cut inference is {\it admissible} (according to this tradition, the cut-inference essentially lives at the metalevel, ensuring compositionality of the logic) or 
(ii) one can consider that the cut inference lives at the object-level and is a fundamental piece of proofs: one thus establishes that the cut inference is \emph{eliminable}, ensuring the sub-formula property (and its numerous important consequences, ranging from consistency to interpolation properties). 

This second tradition may use similar techniques as the first tradition, 
but it also permits the investigation of a syntactic, or effective, approach to cut-elimination, consisting of a cut-reduction relation on proofs, shown to be (at least) weakly normalizing, with the normal forms being cut-free proofs. 
An advantage of such syntactic cut-elimination results is that, in many settings (most notably $\LJ$ and $\LL$~\cite{girard87}), such cut-reductions induce an interesting relation on proofs and have a computational interpretation that is the starting point of the Curry-Howard correspondence built upon sequent calculus~\cite{DBLP:conf/icfp/CurienH00}.

\paragraph{Linear Logic.}
Linear logic (\LL) is often described as a resource-sensitive logic. It is more accurate, though, to view it as a logic designed for analyzing cut-elimination itself.
Indeed, \LL{} comes from an analysis of structural rules, aiming at controlling them rather than weakening them as in substructural logics. This solves some fundamental drawbacks of cut-elimination in classical logic, such as its non-termination or non-confluence. 
%
%
For instance, \LL{} permits the decomposition of
both intuitionistic and classical logic, in a structured and fine-grained manner allowing the refinement of the cut-elimination of those logics as well as their notion of model (allowing the building of a non-trivial denotational model of classical proofs); the prototypical example of such a linear decomposition consists in decomposing the usual \emph{intuitionistic} arrow ({\it i.e.}, the function type of the $\lambda$-calculus), $A \Rightarrow B$, into a replication operator and a linear implication: $\oc A \multimap B$~\cite{DanosJS97,girard87}. 
Further analyses on these exponential modalities led to the discovery of \emph{light logics}, where the complexity of cut-elimination is tamed in a flexible way, usually by considering alternative, 
weaker 
exponential modalities.

The proof theory of \LL{} was extended to \muLLinf{}, that is \LL{} 
with fixed-points in the finitary and non-wellfounded setting~\cite{DBLP:journals/tocl/Baelde12,aminaphd,DBLP:conf/fossacs/Santocanale02,Santocanale13,TABLEAUX23}
and \muLLinf{} allowed for the same kind of linear decomposition for (the non-wellfounded version of) \muLJ{} and \muLK{}. 
A natural question is whether the extensions of \LL{} with fixed-points can also help us achieve syntactic cut-elimination for the modal $\mu$-calculus.

\paragraph{Contributions.}
The discussion of the above paragraph suggests a first question: what would be a linear decomposition of the modal $\mu$-calculus? 
The first contribution of this paper is to provide such a linear decomposition of the modal $\mu$-calculus which is compatible with circular and non-wellfounded proof theory,  \muLLmodinf{}. This \emph{linear-logical modal $\mu$-calculus} will allow us to complete the analysis of cut-elimination for the modal $\mu$-calculus, proving the first syntactic cut-elimination theorem for the full modal $\mu$-calculus (in the non-wellfounded setting).
We therefore adopt the following roadmap  in the body of the paper.
 
 In \Cref{section:definitionsOfKnownSequentCalculi}, we recall the necessary technical background about \muLLinf{} and \muLKmodinf{} proof theory. 
 In \Cref{section:ModalMu}, we motivate and introduce \muLLmodinf{}
and prove its cut-elimination 
in \Cref{section:mullmodinfCutElim}. We then define the linear decomposition of \muLKmodinf{}   into  \muLLmodinf{} from which we conclude to  \muLKmodinf{} cut-elimination theorem in \Cref{section:muLKmodinfcutElim} in the form of an infinitary weak-normalizing cut-reduction system. 

\section{Circular and (non-)wellfounded proof systems}
\label{section:definitionsOfKnownSequentCalculi}

We recall here some basic definitions of both wellfounded \& non-wellfounded systems. 
Some standard definitions of proof theory such as \emph{derivation rules}, \emph{active formula} or \emph{principal formula} are not recalled and can be found in~\cite{BussHandbook}.

\subsection{The Modal $\mu$-calculus}

\subsubsection{Formulas}
Let $\Var$ and $\Atom$ be two disjoint sets of  \emph{fixed-point variables} and of \emph{atoms} respectively.
We define the pre-formulas of the modal $\mu$-calculus, \muLKmodinf{}, as:
{\small $F, G\hspace{0.5em} ::=\hspace{0.5em} 
a
\mid X
\mid \mu X. F \mid \nu X. F \mid \Box F \mid \lozenge F \mid F^\perp \mid F\rightarrow G \mid F \lor G \mid F \land G \mid \false \mid \true\\ (a\in\Atom, X\in\Var)$.}
Knaster-Tarski's theorem guarantees the existence of extremal fixed-points of monotonic functions on complete lattices; monotonicity is reflected syntactically as a \emph{positivity condition} on variables,  defined as:

\begin{defi}[Positive \& negative occurrence of fixed-point variables]
Let $X\in \Var$ be a fixed-point variable, one defines the fact, for $X$, to occur positively (resp. negatively) in a pre-formula by induction on the structure of pre-formulas: 
\begin{itemize}
\item 
 $X$ occurs positively in $X$.
\item
 $X$ occurs positively (resp. negatively) in ${c} (F_1, \dots, F_n)$, if there is some $1\leq i \leq n$ such that $X$ occurs positively (resp. negatively) in $F_i$ for $c\in\{\Box, \lozenge, \lor, \land\}$.
\item
 $X$ occurs positively (resp. negatively) in $F^\perp$ if it occurs negatively (resp. positively) in $F$.
\item 
 $X$ occurs positively (resp. negatively) in $F\rightarrow G$ if $X$ occurs either positively (resp. negatively) in $G$ or negatively (resp. positively) in $F$.
\item 
 $X$ occurs positively (resp. negatively) in $\delta Y. G$ (with $Y\neq X$) if it occurs positively (resp. negatively) in $G$ (for $\delta\in\{\mu, \nu\}$).
\end{itemize}

\end{defi}

We can now define the formulas of \muLKmodinf{}:
\begin{defi}[Formulas]
\label{formulasdef}
A \muLKmodinf{} \emph{formula} $F$ is a closed pre-formula 
such that for any sub-pre-formula of $F$ of the form $\delta X. G$ (with $\delta\in\{\mu, \nu\}$), $X$ does not occur negatively in $G$.
\end{defi}

By considering the $\{\mu,\nu, X\}$-free formulas of this system, we get
\LKmod. By considering the $\{\Box, \lozenge\}$-free formulas of \muLKmodinf{}, we get the $\mu$-calculus.
Finally, the intersection of these two systems, is
propositional classical logic.

\medskip
\subsubsection{Sequent calculus}
We define, here, the sequents, rules and proofs for \muLKmodinf{}.

\begin{defi}[Sequent]
A \emph{sequent} is an ordered pair of two lists of formulas $\A$, the \emph{antecedent}, and $\B$, the \emph{succedent}, that we write $\A\vdash \B$. 
%
\end{defi}

\begin{rem}[Derivation rules \& ancestor relation]
In the structural proof theory literature, inference rules are defined together with an \emph{ancestor relation} (or sub-occurrence relation) between formulas of the conclusion and formulas of the premises of the rule. While this  relation is often overlooked we provide some details here.
Sequent being lists, we define the ancestor relation, to be a relation from the positions of the formula in the conclusion, to the positions of the formula in the premises.
Those ancestor relations will be dealt graphically as in \Cref{fig:LKmodrules,fig:fixFragment}, by drawing the ancestor relation on sequents when needed and leaving it implicit when unambiguous.
\end{rem}

We define the inference rules for \LKmod{} in \Cref{fig:LKmodrules}.
\input{LKmodrules}
Rules for \LK{} will be the $\{\Box,\lozenge\}$-free rules of \LKmod.
We add rules of \Cref{fig:fixFragment} to \LK, (resp. \LKmod) to get the fixed-point version \muLKinf{} (resp. \muLKmodinf) of this system.
\begin{figure}[t]
\hspace{-1cm}
${\scriptsize
\AIC{\tikzmark{mulcl2}\A, \tikzmark{mulf2}F[X:=\mu X.F]\vdash \tikzmark{mulc2}\B}
\RL{\mu_l}
\UIC{\tikzmark{mulcl1}\A, \tikzmark{mulf1}\mu X.F\vdash \tikzmark{mulc1}\B}
\DP\quad
\AIC{\tikzmark{murcl2}\A\vdash \tikzmark{murf2}F[X:=\mu X.F], \tikzmark{murc2}\B}
\RL{\mu_r}
\UIC{\tikzmark{murcl1}\A\vdash \tikzmark{murf1}\mu X.F, \tikzmark{murc1}\B}
\DP\quad
\AIC{\tikzmark{nulcl2}\A, \tikzmark{nulf2}F[X:=\nu X.F]\vdash \tikzmark{nulc2}\B}
\RL{\nu_l}
\UIC{\tikzmark{nulcl1}\A, \tikzmark{nulf1}\nu X.F\vdash \tikzmark{nulc1}\B}
\DP\quad
\AIC{\tikzmark{nurcl2}\A\vdash \tikzmark{nurf2}F[X:=\nu X.F], \tikzmark{nurc2}\B}
\RL{\nu_r}
\UIC{\tikzmark{nurcl1}\A\vdash \tikzmark{nurf1}\nu X.F, \tikzmark{nurc1}\B}
\DP}
$
\begin{tikzpicture}[overlay,remember picture,-,line cap=round,line width=0.1cm]
    \draw[rounded corners, smooth=2,green, opacity=.3] ([xshift=4mm] pic cs:mulf1) to ([xshift=2mm, yshift=2mm] pic cs:mulf2);
    \draw[rounded corners, smooth=2,red, opacity=.3] (pic cs:mulc1) to ([xshift=2mm, yshift=2mm] pic cs:mulc2);
    \draw[rounded corners, smooth=2,red, opacity=.3] (pic cs:mulcl1) to ([xshift=2mm, yshift=2mm] pic cs:mulcl2);
    \draw[rounded corners, smooth=2,green, opacity=.3] ([xshift=4mm] pic cs:murf1) to ([xshift=2mm, yshift=2mm] pic cs:murf2);
    \draw[rounded corners, smooth=2,red, opacity=.3] (pic cs:murc1) to ([xshift=2mm, yshift=2mm] pic cs:murc2);
    \draw[rounded corners, smooth=2,red, opacity=.3] (pic cs:murcl1) to ([xshift=2mm, yshift=2mm] pic cs:murcl2);
    \draw[rounded corners, smooth=2,green, opacity=.3] ([xshift=4mm] pic cs:nulf1) to ([xshift=2mm, yshift=2mm] pic cs:nulf2);
    \draw[rounded corners, smooth=2,red, opacity=.3] (pic cs:nulc1) to ([xshift=2mm, yshift=2mm] pic cs:nulc2);
    \draw[rounded corners, smooth=2,red, opacity=.3] (pic cs:nulcl1) to ([xshift=2mm, yshift=2mm] pic cs:nulcl2);
    \draw[rounded corners, smooth=2,green, opacity=.3] ([xshift=4mm] pic cs:nurf1) to ([xshift=2mm, yshift=2mm] pic cs:nurf2);
    \draw[rounded corners, smooth=2,red, opacity=.3] (pic cs:nurc1) to ([xshift=2mm, yshift=2mm] pic cs:nurc2);
    \draw[rounded corners, smooth=2,red, opacity=.3] (pic cs:nurcl1) to ([xshift=2mm, yshift=2mm] pic cs:nurcl2);
\end{tikzpicture}
    \caption{Rules for the fixed-point fragment\label{fig:fixFragment}}
\end{figure}
The two \emph{exchange rules} $(\exch_l)$ and $(\exch_r)$ from \Cref{fig:LKmodrules,fig:MALLrules}
allows one to derive the rule \quad $\AIC{\sigma'(\A)\vdash \sigma{(\B)}}
\RL{\exch(\sigma',\sigma)}
\UIC{\A\vdash \B}
\DP
$\quad
for any permutations $\sigma$ and $\sigma'$ of $\{ 1,\dots, \#(\B)\}$ and $\{ 1,\dots, \#(\A)\}$ respectively, where $\sigma(\B)$ designates the action of $\sigma$ on the list $\B$, with the induced ancestor relation.
In the rest of the article, we will intentionally treat the exchange rule implicitly: the reader can consider that each of our rules are preceded and followed by a finite number of rule $(\exch)$.
%

Proofs of non fixed-point systems, \LK, \LKmod{} are the trees inductively generated by the corresponding set of rules of each of these systems.
We can define a first notion of infinite derivations, pre-proofs, that will soon be refined:
\begin{defi}[Pre-proofs]
Given a set of derivation rules, a \emph{pre-proof} is a tree co-inductively generated using the rules of the system.
\end{defi}

\begin{exa}[Regular pre-proof]
    \label{exa:circularProof}
Regular pre-proofs are those pre-proofs having a finite number of distinct sub-proofs. We represent them with back-edges.
Taking $F:=\nu X. \lozenge X$, we have:
\hfill \scalebox{.9}{$
\AIC{\orig{circularExas2}{F}\vdash F}
\RL{\diaprom}
\UIC{\lozenge F\vdash \lozenge F}
\doubleLine
\RL{\nu_l, \nu_r}
\UIC{\orig{circularExat2}{F}\vdash F}
\AIC{F\vdash \orig{circularExas1}{F}}
\RL{\diaprom}
\UIC{\lozenge F\vdash \lozenge F}
\doubleLine
\RL{\nu_l, \nu_r}
\UIC{F\vdash \orig{circularExat1}{F}}
\RL{\cut}
\BIC{F\vdash F}
\DP
$
\begin{tikzpicture}[remember picture,overlay]
     \draw [->,>=latex] ([yshift=1mm] circularExas1.east) .. controls +(15:2cm) and +(-11:2cm) .. (circularExat1.east);
    \end{tikzpicture}
\begin{tikzpicture}[remember picture,overlay]
     \draw [->,>=latex] ([yshift=1mm] circularExas2.west) .. controls +(165:1.5cm) and +(10:-2cm) .. (circularExat2.west);
    \end{tikzpicture}}
\end{exa}

\begin{rem}
    \label{rem:inconsistantPreProofSystem}
    The pre-proofs define an inconsistent system. In fact, any sequent is provable:
 \qquad\qquad    \scalebox{.9}{    $
\AIC{~~}
\noLine
\UIC{\orig{inconsists1}\A \vdash\nu X. X}
    \RL{\nu_r}
    \UIC{\orig{inconsistt1}\A \vdash \nu X. X}
    \AIC{\nu X. X\vdash \orig{inconsists2}\B}
    \RL{\nu_l}
    \UIC{\nu X. X \vdash \orig{inconsistt2}\B}
    \RL{\cut}
    \BIC{\A\vdash \B}
    \DP
    $
    \begin{tikzpicture}[remember picture,overlay]
        \draw [->,>=latex] ([yshift=1mm] inconsists1.west) .. controls +(160:2cm) and +(11:-2cm) .. (inconsistt1.west);
    \end{tikzpicture}
    \begin{tikzpicture}[remember picture,overlay]
        \draw [->,>=latex] ([yshift=1mm] inconsists2.east) .. controls +(15:2cm) and +(-11:2cm) .. (inconsistt2.east);
    \end{tikzpicture}
}
\end{rem}

Proofs are defined  as those pre-proofs satisfying a correctness condition:
\begin{defi}[Validity and proofs]
Let $b=(s_i)_{i\in\omega}$ be a consecutive sequence of sequents defining an infinite branch in a pre-proof $\pi$.
A \emph{thread} of $b$ is a sequence $(F_i\in s_i)_{i>n}$, for some $n\in\omega$, of occurrences such that for each $j$, $F_j$ and $F_{j+1}$ satisfy the ancestor relation.
We say that a thread of $b$ is \emph{valid} if the minimal recurring formula
of this sequence, for the sub-formula ordering, exists and is (i) either a $\nu$-formula appearing infinitely often on the succedent of its sequent or a $\mu$-formula appearing infinitely often in the antecedent of its sequent and (ii) the thread is infinitely often principal (there are an infinite number of principal formulas in it). A branch $b$ is \emph{valid} if there is a valid thread of $b$.

A pre-proof is \emph{valid} and is a \emph{proof} if each of its infinite branches is valid.
\end{defi}

The least ($\mu$) and greatest ($\nu$) fixed-point constructors have the same (local) derivation rules: they are distinguished by the (global) validity condition, which is a parity condition akin to parity games for the $\mu$-calculus.

\begin{figure}[t]
    \centering
    \hspace{-1cm}\begin{subfigure}{.5\textwidth}
    $
    {\scriptsize
    \pi_0:=\AIC{}
    \RL{\true}
    \UIC{\vdash \true}
    \RL{\lor_r^1}
    \UIC{\vdash \true\lor \Nat}
    \RL{\mu_r}
    \UIC{\vdash \Nat}
    \DP
    \quad
    \pi_{n+1}:=
    \AIC{\pi_n}
    \noLine
    \UIC{\vdash \Nat}
    \RL{\lor_r^2}
    \UIC{\vdash \true \lor \Nat}
    \RL{\mu_r}
    \UIC{\vdash\Nat}
    \DP}
    $\\[2ex]
    \centering
    $
    \pi_\infty:=\AIC{\vdash\orig{infiniNats}{\Nat}}
    \RL{\mu_r, \lor_r^2}
    \doubleLine
    \UIC{\vdash\orig{infiniNatt}{\Nat}}
    \DP
    $
    \begin{tikzpicture}[remember picture,overlay]
        \draw [->,>=latex] ([yshift=1mm] infiniNats.east) .. controls +(20:2cm) and +(-20:2cm) .. (infiniNatt.east);
    \end{tikzpicture}
    \caption{Cut-free pre-proofs of $\Nat$}
    \label{fig:NatProofDef}
    \end{subfigure}\quad
    \begin{subfigure}{.4\textwidth}
        $
        {\footnotesize\AIC{}
        \RL{\ax}
        \UIC{\true\vdash \true}
        \RL{\lor_r^1}
        \UIC{\true\vdash \true\lor\Nat}
        \RL{\mu_r}
        \UIC{\true\vdash\Nat}
        \AIC{\Nat\vdash\orig{doublefuns}{\Nat}}
        \RL{\mu_r, \lor_r^2}
        \doubleLine
        \UIC{\Nat\vdash\Nat}
        \RL{\lor_r^2}
        \UIC{\Nat\vdash\true\lor\Nat}
        \RL{\mu_r}
        \UIC{\Nat\vdash\Nat}
        \RL{\lor_l}
        \BIC{\true\lor\Nat\vdash\Nat}
        \RL{\mu_l}
        \UIC{\Nat\vdash\orig{doublefunt}{\Nat}}
        \DP}
        $
        \begin{tikzpicture}[remember picture,overlay]
            \draw [->,>=latex] ([yshift=1mm] doublefuns.east) .. controls +(15:2.1cm) and +(-11:3cm) .. (doublefunt.east);
        \end{tikzpicture}
    \caption{Function {\tt double}}
    \label{fig:doubleNatFun}
    \end{subfigure}
    \caption{Valid and invalid pre-proofs}
    \end{figure}


\begin{exa}[Valid and invalid pre-proofs]
    Here, we give some examples of infinite proofs. The pre-proof of \Cref{exa:circularProof} is valid, that of \Cref{rem:inconsistantPreProofSystem} is invalid.

We use the notation $\Nat := \mu X. \true\lor X$, representing the type of natural numbers. We can represent any natural number $n$ by a finite valid proof $\pi_n$ defined in \Cref{fig:NatProofDef}.
    The infinite pre-proof $\pi_{\infty}$, defined in \Cref{fig:NatProofDef}, of $\vdash\Nat$ is not valid:
    the infinite branch in $\pi_\infty$ is supported by only one thread, which is not valid as the minimal formula is a $\mu$-formula appearing only on the right of the proof.
    This is coherent with the interpretation that $\mu$ is a least fixed-point: the system would reject the \emph{infinite} proof of $\Nat$.
    Note that the same kind of proof with $\coNat := \nu X. \true\lor X$ would have given a valid proof.

    The pre-proof of \Cref{fig:doubleNatFun} (representing the {\tt double} function) is valid:
    The only infinite branch in it is the one going infinitely to the right at the application of the $(\lor_l)$-rule. This branch is supported by the infinite thread in the antecedent of each sequent which has a $\mu$-formula as its minimal formula.

\end{exa}

\subsection{Linear Logic with fixed-points}

The main difference between \LK{} (or \LJ{}) and \LL{} lies in the fact that formulas are not always erasable nor duplicable. Hence, the sequent $A, B\vdash A$ is not always provable, neither is $A\vdash A\otimes A$ (a sequent similar to $A\vdash A\land A$ in \LK). This restriction allows \LL{} to interpret programs with finer resource control than \LK{} (or \LJ{}).
Here we recall the usual definitions of both the wellfounded and non-wellfounded systems of \LL, following the definitions of the previous section. \defref{Complete definitions can be found in \Cref{app:lldefdetails}}

\subsubsection{Formulas}
As for \muLKmodinf, let us set $\Var$ and $\Atom$ as two disjoint sets.
The pre-formulas of the non-well-founded linear logic, \muLLinf{} are:
$F, G ::=  ~a \mid X \mid \mu X. F \mid \nu X. F \mid F^\perp\mid F \multimap G \mid F \parr G 
\mid F \otimes G \mid \bot \mid 1 \mid F \oplus G \mid F \with G \mid 0 \mid \top \mid \wn F \mid \oc F$, with $a\in\Atom, X\in\Var$.
%
Positivity of pre-formulas and \muLLinf{} formulas are defined in the same way as for \muLKmodinf{}.
%
The $\{\oc, \wn\}$-free formulas of \muLLinf{} are the formulas of \muMALLinf{}.
The $\{\mu,\nu, X\}$-free fragment of formulas of \muLLinf{} are the formulas of linear logic $\LL$.
The intersection of these two fragments is \MALL.

\medskip
\subsubsection{Sequent calculus}
The definition of sequent is the same as for \muLKmodinf.
The rules of \MALL{} are given by \Cref{fig:MALLrules}, the rules of \LL{} are the rules of \MALL{} together with the rules of \Cref{fig:ExponentialRules}. We add the rules of \Cref{fig:fixFragment} to \MALL{} (resp. \LL{}) to obtain the rules of \muMALLinf{} (resp. \muLLinf).
\input{MALLrules}
\begin{figure}[t]
\centering
${\scriptsize\AIC{\tikzmark{llwnwkcl2}\A \vdash \tikzmark{llwnwkc2}\B}
\RL{\wnwk}
\UIC{\tikzmark{llwnwkcl1}\A\vdash\wn F, \tikzmark{llwnwkc1}\B}
\DP
\quad
\AIC{\tikzmark{llocwkcl2}\A \vdash \tikzmark{llocwkc2}\B}
\RL{\ocwk}
\UIC{\tikzmark{llocwkcl1}\A, \oc F\vdash \tikzmark{llocwkc1}\B}
\DP\quad
\AIC{\tikzmark{llwncontrcl2}\A\vdash \tikzmark{llwncontrf2}\wn F, \tikzmark{llwncontrf3}\wn F, \tikzmark{llwncontrc2}\B}
\RL{\wncontr}
\UIC{\tikzmark{llwncontrcl1}\A\vdash \tikzmark{llwncontrf1}\wn F, \tikzmark{llwncontrc1}\B}
\DP\quad
\AIC{\tikzmark{lloccontrcl2}\A, \tikzmark{lloccontrf2}\oc F, \tikzmark{lloccontrf3}\oc F\vdash\tikzmark{lloccontrc2}\B}
\RL{\occontr{}}
\UIC{\tikzmark{lloccontrcl1}\A, \tikzmark{lloccontrf1}\oc F\vdash \tikzmark{lloccontrc1}\B}
\DP}$\\[2ex]
${\scriptsize
\AIC{\tikzmark{llwndecl2}\A\vdash \tikzmark{llwndef2}F, \tikzmark{llwndec2}\B}
\RL{\wnde}
\UIC{\tikzmark{llwndecl1}\A\vdash \tikzmark{llwndef1}\wn F, \tikzmark{llwndec1}\B}
\DP
\quad
\AIC{\tikzmark{llocdecl2}\A, \tikzmark{llocdef2}F\vdash \tikzmark{llocdec2}\B}
\RL{\ocde}
\UIC{\tikzmark{llocdecl1}\A, \tikzmark{llocdef1}\oc F\vdash \tikzmark{llocdec1}\B}
\DP
\quad
\AIC{\tikzmark{llocpromcl2}\oc\A\vdash \tikzmark{llocpromf2}F, \tikzmark{llocpromc2}\wn\B}
\RL{\ocprom}
\UIC{\tikzmark{llocpromcl1}\oc\A\vdash \tikzmark{llocpromf1}\oc F, \tikzmark{llocpromc1}\wn\B}
\DP
\quad
\AIC{\tikzmark{llwnpromcl2}\oc\A,\tikzmark{llwnpromf2}F\vdash \tikzmark{llwnpromc2}\wn\B}
\RL{\wnprom}
\UIC{\tikzmark{llwnpromcl1}\oc\A, \tikzmark{llwnpromf1}\wn F\vdash \tikzmark{llwnpromc1}\wn\B}
\DP}
$
\begin{tikzpicture}[overlay,remember picture,-,line cap=round,line width=0.1cm]
   \draw[rounded corners, smooth=2,red, opacity=.3] ([xshift=1mm, yshift=0mm] pic cs:llwnwkc1) to ([xshift=1mm, yshift=2mm] pic cs:llwnwkc2);
   \draw[rounded corners, smooth=2,red, opacity=.3] ([xshift=1mm, yshift=0mm] pic cs:llwnwkcl1) to ([xshift=1mm, yshift=2mm] pic cs:llwnwkcl2);
      \draw[rounded corners, smooth=2,red, opacity=.3] ([xshift=1mm, yshift=0mm] pic cs:llocwkc1) to ([xshift=1mm, yshift=2mm] pic cs:llocwkc2);
      \draw[rounded corners, smooth=2,red, opacity=.3] ([xshift=1mm, yshift=0mm] pic cs:llocwkcl1) to ([xshift=1mm, yshift=2mm] pic cs:llocwkcl2);
   \draw[rounded corners, smooth=2,green, opacity=.3] ([xshift=1mm, yshift=0mm] pic cs:llwncontrf1) to ([xshift=2mm, yshift=2mm] pic cs:llwncontrf2);
   \draw[rounded corners, smooth=2,green, opacity=.3] ([xshift=1mm, yshift=0mm] pic cs:llwncontrf1) to ([xshift=2mm, yshift=2mm] pic cs:llwncontrf3);
   \draw[rounded corners, smooth=2,red, opacity=.3] ([xshift=0mm, yshift=0mm] pic cs:llwncontrc1) to ([xshift=1mm, yshift=2mm] pic cs:llwncontrc2);
   \draw[rounded corners, smooth=2,red, opacity=.3] ([xshift=0mm, yshift=0mm] pic cs:llwncontrcl1) to ([xshift=1mm, yshift=2mm] pic cs:llwncontrcl2);
   \draw[rounded corners, smooth=2,green, opacity=.3] ([xshift=1mm, yshift=0mm] pic cs:lloccontrf1) to ([xshift=2mm, yshift=2mm] pic cs:lloccontrf2);
   \draw[rounded corners, smooth=2,green, opacity=.3] ([xshift=1mm, yshift=0mm] pic cs:lloccontrf1) to ([xshift=2mm, yshift=2mm] pic cs:lloccontrf3);
   \draw[rounded corners, smooth=2,red, opacity=.3] ([xshift=0mm, yshift=0mm] pic cs:lloccontrc1) to ([xshift=1mm, yshift=2mm] pic cs:lloccontrc2);
   \draw[rounded corners, smooth=2,red, opacity=.3] ([xshift=0mm, yshift=0mm] pic cs:lloccontrcl1) to ([xshift=1mm, yshift=2mm] pic cs:lloccontrcl2);
   \draw[rounded corners, smooth=2,green, opacity=.3] ([xshift=2mm, yshift=0mm] pic cs:llwndef1) to ([xshift=1mm, yshift=2mm] pic cs:llwndef2);
   \draw[rounded corners, smooth=2,red, opacity=.3] ([xshift=1mm, yshift=0mm] pic cs:llwndec1) to ([xshift=1mm, yshift=2mm] pic cs:llwndec2);
   \draw[rounded corners, smooth=2,red, opacity=.3] ([xshift=1mm, yshift=0mm] pic cs:llwndecl1) to ([xshift=1mm, yshift=2mm] pic cs:llwndecl2);
   \draw[rounded corners, smooth=2,green, opacity=.3] ([xshift=2mm, yshift=0mm] pic cs:llocdef1) to ([xshift=1mm, yshift=2mm] pic cs:llocdef2);
   \draw[rounded corners, smooth=2,red, opacity=.3] ([xshift=1mm, yshift=0mm] pic cs:llocdec1) to ([xshift=1mm, yshift=2mm] pic cs:llocdec2);
   \draw[rounded corners, smooth=2,red, opacity=.3] ([xshift=1mm, yshift=0mm] pic cs:llocdecl1) to ([xshift=1mm, yshift=2mm] pic cs:llocdecl2);
   \draw[rounded corners, smooth=2,green, opacity=.3] ([xshift=2mm, yshift=0mm] pic cs:llocpromf1) to ([xshift=1mm, yshift=2mm] pic cs:llocpromf2);
   \draw[rounded corners, smooth=2,red, opacity=.3] ([xshift=2mm, yshift=0mm] pic cs:llocpromc1) to ([xshift=2mm, yshift=2mm] pic cs:llocpromc2);
   \draw[rounded corners, smooth=2,red, opacity=.3] ([xshift=2mm, yshift=0mm] pic cs:llocpromcl1) to ([xshift=2mm, yshift=2mm] pic cs:llocpromcl2);
   \draw[rounded corners, smooth=2,green, opacity=.3] ([xshift=2mm, yshift=0mm] pic cs:llwnpromf1) to ([xshift=1mm, yshift=2mm] pic cs:llwnpromf2);
   \draw[rounded corners, smooth=2,red, opacity=.3] ([xshift=2mm, yshift=0mm] pic cs:llwnpromc1) to ([xshift=2mm, yshift=2mm] pic cs:llwnpromc2);
   \draw[rounded corners, smooth=2,red, opacity=.3] ([xshift=2mm, yshift=0mm] pic cs:llwnpromcl1) to ([xshift=2mm, yshift=2mm] pic cs:llwnpromcl2);
\end{tikzpicture}
\caption{Exponential fragment of \LL}\label{fig:ExponentialRules}
\end{figure}
Pre-proofs as well as validity are defined as in the previous section.

In linear logic, duplicability and erasability are controlled by $\wn$ (\emph{why not}) and $\oc$ (\emph{of course}) modalities. The $(\occontr{})$, $(\ocwk)$, and $(\ocde)$ rules allow duplication, erasure, and use of $\oc$-prefixed antecedents, respectively. The $(\ocprom)$ rule enables using an $\oc$-prefixed succedent when all antecedents are $\oc$-prefixed. This controlled approach to contraction and weakening sequentializes certain reductions in cut-elimination, leading to strong normalization for \LL{}~\cite{PagTor10tcs}.
%
However, the good normalization properties of \LL{} can be recovered by using a linear translation from \LK{} to \LL{}, similar to the double negation translations from \LK{} to \LJ{}. Indeed, every formula, every sequent, and every proof of \LK{} can be translated into a proof in \LL{} by adding $\wn$ and $\oc$ modalities. For instance in  \Cref{exa:linearTranslation}
we show an example where each connective $c(A_1, \dots, A_n)$ can be translated as $\oc(c(\wn A_1, \dots, \wn A_n))$  adding a $\wn$ on the formula from the succedent, the additional rules being shown in blue.
\begin{figure}[t]
\hfill \scalebox{.9}{$
\AIC{}
\RL{\ax}
\UIC{A\vdash A}
\RL{(-)^\perp_r}
\UIC{\vdash A^\perp, A}
\AIC{}
\RL{\ax}
\UIC{A\vdash A}
\RL{\rightarrow_l}
\BIC{A^\perp\rightarrow A\vdash A, A}
\RL{\cntr}
\UIC{A^\perp\rightarrow A\vdash A}
\DP\quad\rightsquigarrow\quad
\AIC{}
\RL{\ax}
\UIC{A\vdash A}
\RL{{\color{blue}\wnprom, \wnde}}
\doubleLine
\UIC{\wn A\vdash \wn A}
\RL{(-)^\perp_r}
\UIC{\vdash (\wn A)^\perp, \wn A}
\doubleLine
\RL{{\color{blue}\wnde, \ocprom}}
\doubleLine
\UIC{\vdash \wn\oc(\wn A)^\perp, \wn A}
\AIC{}
\RL{\ax}
\UIC{A\vdash A}
\RL{{\color{blue}\wnprom, \wnde}}
\doubleLine
\UIC{\wn A\vdash \wn A}
\RL{\multimap_l}
\BIC{\wn\oc(\wn A)^\perp\multimap \wn A\vdash \wn A, \wn A}
\RL{{\color{blue}\ocde}}
\UIC{\oc(\wn\oc(\wn A)^\perp\multimap \wn A)\vdash \wn A, \wn A}
\RL{\wncontr}
\UIC{\oc(\wn\oc(\wn A)^\perp\multimap \wn A)\vdash \wn A}
\DP$}
\caption{Example of a linear translation\label{exa:linearTranslation}}\end{figure}
By taking any maximal sequence of cut reduction on such proofs in \LL{} and using the strong normalization property, we find a cut-free proof of the same sequent in \LL{}. This projects to an \LK{} proof of the original sequent, by simply \emph{forgetting} superfluous modalities.
In \Cref{section:muLKmodinfcutElim}, we will use this technique to prove the cut elimination of the modal $\mu$-calculus.

\subsubsection{Cut-elimination for \muLLinf{}} We postpone to \Cref{section:mullmodinfCutElim}  the discussion of cut-elimination reductions (\Cref{multicutdef}, \ref{def:infty-red-seq} and \ref{def:fair-red-seq}) and theorems (\Cref{muLLinfCutElim}): they are directly introduced for \muLLmodinf{} that we shall consider in the next section.


\section{A linear-logical modal $\mu$-calculus: \muLLmodinf{}}\label{section:ModalMu}

We now introduce \muLLmodinf{}, which can be viewed from two perspectives:
\begin{itemize}
\item as an extension of \muLLinf{} with modalities akin to the modal $\mu$-calculus;
\item as a linearization of the modal $\mu$-calculus (where \emph{linear} refers to linear logic -- by default, any assumption must be used exactly once in a proof -- and not to the structures of its models, as in LTL or linear-time $\mu$-calculus~\cite{stirling}).
\end{itemize}

In designing the \muLLmodinf{} sequent calculus, the primary challenge lies in understanding the interaction between the $\Box$ and $\lozenge$ modalities and \LL{} exponentials, while being compatible with the fixed-point inferences. 
A natural requirement for such a sequent calculus is to allow us to extend the linear decomposition of \muLKinf{}~\cite{TABLEAUX23} to \muLKmodinf{}. 
The constraints imposed by such a linear decomposition will therefore guide us in defining the exponential and modal rules of \muLLmodinf{}. The following discussion illustrates these requirements through examples and will lead us to the desired sequent calculus defined in \Cref{def:mullmod} below.

A first approach would be to simply extend \muLLinf{} with the usual inferences for $\Box$ and $\lozenge$, extending the translation $\trans{(-)}$ from~\cite{TABLEAUX23} on modalities as in \Cref{exa:linearTranslation} setting $ \trans{(\lozenge A)} := \oc\lozenge\wn\trans{(A)}$ and $\trans{(\Box A)} := \oc\Box\wn\trans{(A)} $.
This approach is too simplistic, though, and fails.
Consider indeed an instance of the modal rule $\boxprom$ (defined in \Cref{fig:LKmodrules})
with $\A=[]$ and $\B=[B]$:
$
\AIC{\vdash A, B}
\RL{\boxprom}
\UIC{\vdash \Box A, \lozenge B}
\DP.
$
Following the standard sequent translation from \LK{} to \LL{}, we need to derive $\vdash \wn\oc\Box \wn\trans{A}, \wn\oc\lozenge\wn\trans{B}$ from the premise $\vdash \wn\trans{A}, \wn\trans{B}$. Reasoning bottom-up:
to conclude $\vdash\wn\trans{A}, \wn\trans{B}$, we must remove both the $\lozenge$ and the $\Box$ with a modal rule. However, regardless of the sequence of rules used, we always end up with a sequent containing an $\oc$ together with either $\Box$ or $\lozenge$. An example of such a derivation is shown in \Cref{fig:failed-modal-derivation}.
\begin{figure}[t]
\centering
\begin{subfigure}{.3\textwidth}
$
\AIC{\vdash\oc\Box \wn\trans{A}, \lozenge\wn\trans{B}}
\RL{\wnde}
\UIC{\vdash\wn\oc\Box \wn\trans{A}, \lozenge\wn\trans{B}}
\RL{\ocprom}
\UIC{\vdash\wn\oc\Box \wn\trans{A}, \oc\lozenge\wn\trans{B}}
\RL{\wnde}
\UIC{\vdash \wn\oc\Box \wn\trans{A}, \wn\oc\lozenge\wn\trans{B}}
\DP
$
\caption{Failed linear translation of the $\Box$ rule}
\label{fig:failed-modal-derivation}
\end{subfigure}\quad
\begin{subfigure}{.3\textwidth}
$\AIC{\vdash \wn\trans{A}, \wn\trans{B}}
\RL{\boxprom}
\UIC{\vdash\Box \wn\trans{A}, \lozenge\wn\trans{B}}
\RL{\ocpromloz}
\UIC{\vdash\oc\Box \wn\trans{A}, \lozenge\wn\trans{B}}
\RL{\wnde}
\UIC{\vdash\wn\oc\Box \wn\trans{A}, \lozenge\wn\trans{B}}
\RL{\ocpromloz}
\UIC{\vdash\wn\oc\Box \wn\trans{A}, \oc\lozenge\wn\trans{B}}
\RL{\wnde}
\UIC{\vdash \wn\oc\Box \wn\trans{A}, \wn\oc\lozenge\wn\trans{B}}
\DP
$
\caption{Successful translation of the $\Box$ rule}
\label{fig:success-modal-derivation}
\end{subfigure}\quad
\begin{subfigure}{.3\textwidth}
$
\AIC{\Box\A', \oc\A\vdash A, \wn\B, \lozenge\B'}
\RL{\ocpromloz}
\UIC{\Box\A', \oc\A\vdash \oc A, \wn\B, \lozenge\B'}
\DP
$
\caption{New $\oc$-rule}
\label{fig:linear-modal-rule}
\end{subfigure}
\caption{Translating the modal rule in a linear framework}
\end{figure}
In our attempt to translate this rule, we are left with an unprovable sequent: indeed, the top sequent cannot be the conclusion of any rule in the system (except for cut and exchange, of course): $\oc\Box\wn\trans{A}$ cannot be principal since there is a $\lozenge$-formula in the context, and $\lozenge\wn\trans B$ cannot be principal since there is an $\oc$-formula. 

A natural solution is to allow $\oc$-promotion in contexts containing 
$\oc,\Box$-formulas on the left and $\wn,\lozenge$-formulas on the right
 as in \Cref{fig:linear-modal-rule}.\footnote{Note that allowing modal rules in contexts containing $\wn$ on the right and $\oc$ on the left would also solve our problem. However, forgetting the linear information of such a rule would not give us a rule derivable in \muLKmodinf{}.}
The derivation of the translation of our $(\boxprom)$ instance can be completed as
 in \Cref{fig:success-modal-derivation}.

Allowing $\Box/\lozenge$-formulas in the context of a promotion rule has implications for the system's robustness to cut-elimination.
For instance, taking the $(\wnprom/\wnwk)$ principal case and adding modal formulas to the context naturally requires
being able to weaken  $\Box$ and $\lozenge$-formulas in the antecedent (resp. succedent):

\scalebox{.85}{
$
\AIC{\Gamma\vdash \Sigma}
\RL{\wnwk}
\UIC{\Gamma\vdash \wn C, \Sigma}
\AIC{\oc\Delta{\color{blue}, \Box\Lambda}, C\vdash \wn\Phi{\color{blue}, \lozenge\Psi}}
\RL{\wnprom}
\UIC{\oc\Delta, {\color{blue}\Box\Lambda,} \wn C\vdash \wn\Phi{\color{blue}, \lozenge\Psi}}
\RL{\cut}
\BIC{\Gamma, \oc\Delta, {\color{blue}\Box\Lambda}\vdash \Sigma, \wn\Phi{\color{blue}, \lozenge\Psi}}
\DP
\rightsquigarrow
\hspace{0cm}
\AIC{\Gamma\vdash \Sigma}
\RL{\wnwk, \ocwk{\color{blue}, \boxwk, \diawk}}
\doubleLine
\UIC{\Gamma, \oc\Delta{\color{blue}, \Box\Lambda}\vdash \Sigma, \wn\Phi{\color{blue}, \lozenge\Psi}}
\DP
$}

Similarly, the $(\wncontr/\wnprom)$ key-case requires the ability to contract $\Box/\lozenge$-formulas. \defref{More details can be found in \Cref{app:DetailsDiscussionRobustnessLLbox}.}

We now give a formal definition of the linear-logical modal $\mu$-calculus:
\begin{defi}[Linear-logical modal $\mu$-calculus]
\label{def:mullmod}
Pre-formulas of \muLLmodinf{} are defined as:
\hfill $F, G \hspace{0.5em}::= \hspace{0.5em} a\in\Atom \mid X\in\Var \mid \mu X. F \mid \nu X. F \mid F^\perp\mid F \multimap G \mid F \parr G $

~\hfill $ \mid F \otimes G \mid \bot \mid 1 \mid F \oplus G \mid F \with G \mid 0 \mid \top \mid \wn F \mid \oc F \mid \lozenge F \mid \Box F.$

We obtain positivity of an occurrence and the definition of formulas of \muLLmodinf{} in the same way as for \muLLinf{} and \muLKmodinf{}.
Rules for $\muLLmodinf$ are the rules for \muLLinf{} together with 
those of
 \Cref{fig:LLTwoSidedModRules}.
Note that $(\ocprom)$ (resp. $(\wnprom)$) is an instance of $(\ocpromloz)$ (resp. $(\wnpromloz)$).
We define the sequents, pre-proofs, and proofs 
as in \Cref{section:definitionsOfKnownSequentCalculi}.
\end{defi}
\begin{figure*}[t]
    ${\footnotesize
    \AIC{\A\vdash A, \B}
    \RL{\boxprom}
    \UIC{\Box\A\vdash \Box A, \lozenge\B}
    \DP,~
    \AIC{\A\vdash \lozenge A, \lozenge A, \B}
    \RL{\diacontr}
    \UIC{\A\vdash \lozenge A, \B}
    \DP, ~
    \AIC{\A\vdash \B}
    \RL{\diawk}
    \UIC{\A\vdash \lozenge A, \B}
    \DP, ~
    \AIC{\Box \A', \oc\A\vdash A, \wn\B,\lozenge\B'}
    \RL{\ocpromloz}
    \UIC{\Box \A', \oc\A\vdash \oc A, \wn\B, \lozenge\B'}
    \DP}
    $\\[1ex]
    ${\footnotesize
    \AIC{\A, A\vdash \B}
    \RL{\diaprom}
    \UIC{\Box\A, \lozenge A\vdash \lozenge\B}
    \DP,~
    \AIC{\A, \Box A, \Box A\vdash \B}
    \RL{\boxcontr}
    \UIC{\A, \Box A\vdash \B}
    \DP, ~
    \AIC{\A\vdash \B}
    \RL{\boxwk}
    \UIC{\A, \Box A\vdash \B}
    \DP, ~
    \AIC{\Box \A', \oc\A, A\vdash \wn\B,\lozenge\B'}
    \RL{\wnpromloz}
    \UIC{\Box \A', \oc\A, \wn A\vdash \wn\B, \lozenge\B'}
    \DP}
    $
\caption{Rules involving modalities for \muLLmodinf{}}\label{fig:LLTwoSidedModRules}
\end{figure*}


To justify the relevance of \muLLmodinf{} we prove that 
there is a forgetful operation from \muLLmodinf{}  proofs to  \muLKmodinf{}, the \emph{skeleton}.
We define a translation $\sk{-}$ from \muLLmodinf{} formulas, rules and pre-proofs to \muLKmodinf{} :
\begin{defi}[\muLKmodinf-Skeleton]
    \label{def:skTranslation}
    We define the skeleton of formulas inductively: \\
$\begin{array}{rcl rcl rcl}
\sk{1}&=&\true\quad&
\sk{F\otimes G} &=& \sk{F} \land \sk{G}\quad&
\sk {F^\perp} &= &\sk{F}^\perp\\
\sk{\bot} &=& \false\quad&
\sk{F\parr G} &=& \sk{F}\lor\sk{G}\quad&
\sk{\lozenge F} & = & \lozenge \sk{F}\\
\sk{\top} &=& \true \quad&
\sk{F\with G} &=& \sk F\land \sk G \quad&
\sk{\Box F} & = & \Box \sk{F}\\
\sk{0} &=& \false \quad&
\sk{F\oplus G} &=& \sk F \lor \sk G \quad&
\sk{\wn F} &=& \sk F \\
\sk a &=& a \quad&
\sk{F\multimap G} &=&  \sk F \rightarrow \sk G&
\sk{\oc F} &=& \sk F\\
\sk X &=& X \quad&
\sk{\mu X. F} &=& \mu X. \sk F \quad&
\sk{\nu X. F} &=& \nu X. \sk F
 \end{array}$

Sequents of \muLLmodinf{} are translated to sequents of skeletons of these formulas.
Translation of rules are standard. \defref{More details in \Cref{app:skTranslation}.}
Translations of pre-proofs are obtained co-inductively by applying rule translations.
\end{defi}

The following ensures that validity is preserved both ways by $\sk{-}$:

\begin{prop}[Robustness of the skeleton to validity]\label{skeletonValidity}
If $\pi$ is a valid pre-proof, then $\sk{\pi}$ is a \muLKmodinf{} valid pre-proof, and vice versa.
\proofref{See proof in \Cref{sec:app:skeletonvalidity}.}\end{prop}

    \section{Cut-elimination for \muLLmodinf{}}
\label{section:mullmodinfCutElim}
To eliminate cuts in \muLLmodinf{}, we employ a generalization of the cut inference called multicuts, as done in previous works on similar non-wellfounded proof systems~\cite{LICS22,CSL16,aminaphd,TABLEAUX23}. \defref{More details in \Cref{app:multicutdef}.}
\begin{defi}[Multicut rule]\label{multicutdef}
The multicut rule is the following rule:

$$
\AIC{\A_1\vdash \B_1}
\AIC{\dots}
\AIC{\A_n\vdash \B_n}
\RL{\rmcutpar}
\TIC{\A\vdash\B}
\DP
$$

It can have any number of premises.
The ancestor relation $\iota$ maps one formula of the conclusion to exactly one formula of the premises, while the $\cutrel$-relation links cut-formulas together, subject to acyclicity and connectedness conditions.
\end{defi}

\begin{rem}
    The idea behind the multicut is to abstract a finite tree of binary cuts quotiented by the cut-commutation rule.
    We provide an example of a multicut rule and represent graphically $\iota$ in red and $\cutrel$ in blue.

$$
    \AIC{\vdash A\tikzmark{PC1}, B}
    \AIC{B\tikzmark{PC2}\vdash C}
    \AIC{C\tikzmark{PC3}\vdash D}
    \RL{\rmcutpar}
    \TIC{\vdash A\tikzmark{CC},D}
    \DP
    \begin{tikzpicture}[overlay,remember picture,-,line cap=round,line width=0.1cm]
       \draw[rounded corners, smooth=2,red, opacity=.25] ($(pic cs:CC)+(-.2cm,.1cm)$)to ($(pic cs:PC1)+(-.2cm,-.1cm)$)to ($(pic cs:PC1)+(-.2cm,.1cm)$); 
       \draw[rounded corners, smooth=2,red, opacity=.25] ($(pic cs:CC)+(.3cm,.1cm)$)to ($(pic cs:PC3)+(.5cm,-.1cm)$)to ($(pic cs:PC3)+(.5cm,.1cm)$); 
      \draw[rounded corners, smooth=2,cyan, opacity=.25] ($(pic cs:PC1)+(.25cm,.1cm)$)to ($(pic cs:PC1)+(.25cm,-.1cm)$)to ($(pic cs:PC2)+(-.1cm,-.1cm)$) to ($(pic cs:PC2)+(-.1cm,.1cm)$);
        \draw[rounded corners, smooth=2,cyan, opacity=.25] ($(pic cs:PC2)+(.5cm,.1cm)$)to ($(pic cs:PC2)+(.5cm,-.1cm)$)to ($(pic cs:PC3)+(-.1cm,-.1cm)$) to ($(pic cs:PC3)+(-.1cm,.1cm)$);  
       \end{tikzpicture}
       $$

The multicut rule should be seen as a tree of binary cuts, cf  (\cut/\mcut)-case:

$${\small
    \AIC{\C}
    \AIC{\A_1\vdash F, \B_1 ~~ \A_2, F\vdash\B_2}
    \RL{\cut}
    \UIC{\A_1, \A_2\vdash \B_1, \B_2}
    \RL{\tiny \rmcutpar}
    \BIC{\A\vdash \B}
    \DP~
    \rightsquigarrow~
   \AIC{\C ~~
    \A_1\vdash F, \B_1
    ~~\A_2, F\vdash \B_2}
    \RL\rmcutparprime
    \UIC{\A\vdash \B}
    \DP}
    $$
    
    Here, $\iota'(p)=\iota(p)$ for positions of formulas sent to $\C$, and uses the ancestor relation of the $\cut$-rule together with $\iota$ to determine the image of the other positions.
    The relation $\cutrel'$ is obtained from $\cutrel$ by adding $p\cutrel p'$, where $p$ and $p'$ are the positions of the two cut-formulas $F$.
    \end{rem}

\subsection{The (\mcut) reduction steps}

As we use a multicut reduction strategy, we first describe the steps of reduction.
To describe these \mcut-steps of reduction, we will use a notation and a definition:
\begin{nota}[$(\oc)$-contexts] 
    If $\C$ is a list of $\muLLmodinf$-proofs, we denote by $\ocboxProofs$ the list of proofs obtained by applying one of $(\ocpromloz)$, $(\wnpromloz)$, $(\boxprom)$ or $(\diaprom)$ rules to each proof of $\C$.
    If $\C$ is a list of $\muLLmodinf$-proofs, we denote by $\boxProofs$ the list of proofs obtained by applying one of $(\boxprom)$ or $(\diaprom)$ rules to each proof of $\C$.
\end{nota}

\begin{defi}[Restriction of a $\mathsf{mcut}$ context]
\label{multicutRestriction}
Let $
\AIC{\C}
\RL{\rmcutpar}
\UIC{s}
\DP
$ be a multicut occurrence such that $\C = s_1~\dots~ s_n$ with $s_i := F_1, \dots, F_{k_i}\vdash G_1, \dots, G_{r_i}$. We define $\C_{F_j}$ (resp. $\C_{G_j}$) to be the sequents $\cutrel$-connected 
to the formula $F_j$ (resp. $G_j$).
%
This is extended to contexts of formulas with $\C_\Gamma := \cup_{F\in\Gamma} \C_F$.
\defref{More details are provided in \Cref{app:multicutRestriction}.}\end{defi}

\begin{figure*}[t]
\[
\begin{array}{c}
        \AIC{}
        \RL{\top_r}
        \UIC{\A'\vdash\top, \B'}
        \AIC{\C}
        \RL\rmcutpar
        \BIC{\A\vdash \top, \B}
        \DP\quad \rightsquigarrow\quad
        \AIC{}
        \RL{\top_r}
        \UIC{\A\vdash \top, \B}
        \DP\\[15pt]
 \AIC{}
        \RL{\ax}
        \UIC{A\vdash A}
        \RL{\rmcutpar}
        \UIC{A\vdash A}
        \DP\quad\rightsquigarrow\quad
        \AIC{}
        \RL{\ax}
        \UIC{A\vdash A}
        \DP\\[15pt]
        \hspace{-1cm}
\AIC{\A_1, F[X:=\nu X. F]\vdash \B_1}
        \RL{\nu_l}
        \UIC{\A_1, \nu X. F\vdash \B_1}
        \AIC{\A_2\vdash F[X:=\nu X. F], \B_2}
        \RL{\nu_r}
        \UIC{\A_2\vdash \nu X. F, \B_2}
        \AIC{\C}
        \RL{\rmcutpar}
        \TIC{\A\vdash\B}
        \DP
        \\[15pt]
\hspace{1cm}
     \rightsquigarrow\quad   \AIC{\A_1, F[X:=\nu X. F]\vdash \B_1}
        \AIC{\A_2\vdash F[X:=\nu X. F], \B_2}
        \AIC{\C}
        \RL{\rmcutpar}
        \TIC{\A\vdash\B}
        \DP\\[15pt]
        \hspace{-.5cm}
\AIC{\A_1\vdash A_1, \B_1}
        \AIC{\A_2\vdash A_2, \B_2}
        \RL{\otimes_r}
        \BIC{\A_1, \A_2\vdash A_1\otimes A_2, \B_1, \B_2}
        \AIC{\A_3, A_1, A_2\vdash \B_3}
        \RL{\otimes_l}
        \UIC{\A_3, A_1\otimes A_2\vdash \B_3}
        \AIC{\C}
        \RL{\rmcutpar}
        \TIC{\A\vdash \B}
        \DP
        \\[15pt]
 \hspace{1cm}\rightsquigarrow\quad
        \AIC{\A_1\vdash A_1, \B_1}
        \AIC{\A_2\vdash A_2, \B_2}
        \AIC{\A_3, A_1, A_2\vdash \B_3}
        \AIC{\C}
        \RL{\rmcutparprime}
        \QIC{\A\vdash \B}
        \DP\\[15pt]
        \hspace{-3cm}
\AIC{\A'\vdash A_1, \B'}
        \AIC{\A'\vdash A_2, \B'}
        \RL{\with_r}
        \BIC{\A'\vdash A_1\with A_2, \B'}
        \AIC{\C}
        \RL\rmcutpar
        \BIC{\A\vdash A_1\with A_2, \B}
        \DP
\\[15pt]
        \hspace{1cm}\quad
        \rightsquigarrow\quad\AIC{\A'\vdash A_1, \B'}
        \AIC{\C}
        \RL{\rmcutpar}
        \BIC{\A\vdash A_1, \B}
        \AIC{\A'\vdash A_2, \B'}
        \AIC{\C}
        \RL{\rmcutpar}
        \BIC{\A\vdash A_2, \B}
        \RL{\with_r}
        \BIC{\A\vdash A_1\with A_2, \B}
        \DP
        \end{array}
        \]
\caption{Examples of (\mcut)-step of \muMALLinf{}}\label{fig:mcutMALLExa}
\end{figure*}

The \mcut-reduction steps of \muMALLinf{} can be found in~\cite{CSL16,TABLEAUX23}; we provide examples of these steps in \Cref{fig:mcutMALLExa}.
The reduction steps for the exponential and modal fragment of \muLLmodinf{} are presented in \Cref{fig:mullmodinfexpcommcutstep,fig:mullmodinfexpprincipcutstep}, which include the exponential steps of \muLLinf{}.
For commutative steps, we present only the cases where the principal formula of the rule being commuted appears in the succedent of the sequent.
For principal steps, we present only the cases where the cut-formula is a $\wn$- or $\lozenge$-formula.
All other cases can be derived by duality.
\defref{More details in \Cref{app:mcutSteps}}
\begin{figure*}[t] 
   {\footnotesize
   \begin{align*}
{\AIC{\pi}
   \noLine
   \UIC{\oc\A_1, \Box\A_2\vdash A, \wn\B_1, \lozenge\B_2}
   \RL{\ocpromloz}
   \UIC{\oc\A_1, \Box\A_2\vdash \oc A, \wn\B_1, \lozenge\B_2}
   \AIC{\ocboxProofs}
   \RL{\rmcutpar}
   \BIC{\oc\A', \Box\A\vdash \oc A, \wn\B, \lozenge\B'}
   \DP}\quad\qquad\qquad\qquad\qquad   
   \\
   \rightsquigarrow\quad
   {
   \AIC{\pi}
   \noLine
   \UIC{\oc\A_1, \Box\A_2\vdash A, \wn\B_1, \lozenge\B_2}
   \AIC{\ocboxProofs}
   \RL{\rmcutpar}
   \BIC{\oc\A', \Box\A\vdash A, \wn\B, \lozenge\B'}
   \RL{\ocpromloz}
   \UIC{\oc\A', \Box\A\vdash \oc A, \wn\B, \lozenge\B'}
   \DP}
   \\
   {
   \AIC{\pi}
   \noLine
   \UIC{\A\vdash A, \B}
   \RL{\boxprom}
   \UIC{\Box\A\vdash \Box A, \lozenge\B}
   \AIC{\boxProofs}
   \RL{\rmcutpar}
   \BIC{\Box\A'\vdash \Box A, \lozenge\B'}
   \DP}
\quad   \rightsquigarrow\quad
   {
   \AIC{\pi}
   \noLine
   \UIC{\A\vdash A, \B}
   \AIC{\C}
   \RL{\rmcutpar}
   \BIC{\A'\vdash A, \B'}
   \RL{\boxprom}
   \UIC{\Box\A'\vdash \Box A, \lozenge\B'}
   \DP}
   \\{
   \AIC{\pi}
   \noLine
   \UIC{\A\vdash \B}
   \RL{\delta_\wk}
   \UIC{\A\vdash \delta A, \B}
   \AIC{\mathcal{C}}
   \RL{\rmcutpar}
   \BIC{\A'\vdash \delta A, \B'}
   \DP}
\quad   \rightsquigarrow\quad
   {
   \AIC{\pi}
   \noLine
   \UIC{\A\vdash \B}
   \AIC{\mathcal{C}}
   \RL{\mathsf{mcut(\iota', \perp\!\!\!\perp')}}
   \BIC{\A'\vdash \B'}
   \RL{\delta_\wk}
   \UIC{\A'\vdash \delta A, \B'}
   \DP}
   \\{
   \AIC{\pi}
   \noLine
   \UIC{\A\vdash \delta A, \delta A, \B}
   \RL{\delta_\cntr{}}
   \UIC{\A\vdash \delta A, \B}
   \AIC{\mathcal{C}}
   \RL{\rmcutpar}
   \BIC{\A'\vdash \delta A, \B'}
   \DP} 
\quad   \rightsquigarrow\quad
   {
   \AIC{\pi}
   \noLine
   \UIC{\A\vdash \delta A, \delta A, \B}
   \AIC{\mathcal{C}}
   \RL{\mathsf{mcut(\iota', \perp\!\!\!\perp')}}
   \BIC{\A'\vdash \delta A, \delta A, \B'}
   \RL{\delta_\cntr{}}
   \UIC{\A'\vdash \delta A, \B'}
   \DP}\\ 
   {
   \AIC{\pi}
   \noLine
   \UIC{\A\vdash  A, \B}
   \RL{\wnde}
   \UIC{\A\vdash \wn A, \B}
   \AIC{\mathcal{C}}
   \RL{\rmcutpar}
   \BIC{\A'\vdash \wn A, \B'}
   \DP}
   \quad\rightsquigarrow\quad
   {
   \AIC{\pi}
   \noLine
   \UIC{\A\vdash A, \B}
   \AIC{\mathcal{C}}
   \RL{\mathsf{mcut(\iota', \perp\!\!\!\perp')}}
   \BIC{\A'\vdash A, \B'}
   \RL{\wnde}
   \UIC{\A'\vdash \wn A, \B'}
   \DP}
   \end{align*}
   }
   \caption{\muLLmodinf{} exponential \& modal commutative cut-elimination steps (commutation with right rules) -- $\delta\in\{\wn, \lozenge\}$}\label{fig:mullmodinfexpcommcutstep}
   \end{figure*}

   \begin{figure*}[t]
{\small   \begin{align*}
   {\AIC{\tikzmark{contrPrincipRed12}\C_{\A,\B}}
   \AIC{\pi}
   \noLine
   \UIC{\A\vdash \delta A, \delta A, \B}
   \RL{\delta_\cntr}
   \UIC{\tikzmark{contrPrincipRed12pl}\A\vdash \delta A, \B\tikzmark{contrPrincipRed12p}}
   \AIC{\tikzmark{contrPrincipRed22}\ocboxProofs_{\delta A}}
   \RL{\rmcutpar}
   \TIC{\tikzmark{contrPrincipRed21l}\oc\A_1, \tikzmark{contrPrincipRed31l}\Box\A_2, \tikzmark{contrPrincipRed11l} \A_3\vdash \tikzmark{contrPrincipRed11}\B_1, \tikzmark{contrPrincipRed21}\wn\B_2, \tikzmark{contrPrincipRed31}\lozenge\B_3}
   \DP}
   \qquad\qquad\qquad\qquad\qquad\qquad\qquad\qquad
\\
   \quad\rightsquigarrow\quad
   {\AIC{\tikzmark{contrPrincipRedpostred12}\C_{\A,\B}}
   \AIC{\pi}
   \noLine
   \UIC{\A\vdash \delta A, \delta A, \B\tikzmark{contrPrincipRedpostred12p}}
   \AIC{\tikzmark{contrPrincipRedpostred22}\ocboxProofs_{\delta A}}
   \AIC{\ocboxProofs_{\delta A}\tikzmark{contrPrincipRedpostred32}}
   \RL{\mathsf{mcut(\iota', \perp\!\!\!\perp')}}
   \QIC{\oc\A_1, \Box\A_2, \quad\oc\A_1, \Box\A_2,\quad \A_3\vdash\tikzmark{contrPrincipRedpostred11}\B_1, \quad \tikzmark{contrPrincipRedpostred211}\wn\B_2,\tikzmark{contrPrincipRedpostred21} \lozenge\B_3, \quad \tikzmark{contrPrincipRedpostred311}\wn\B_2, \tikzmark{contrPrincipRedpostred31}\lozenge\B_3}
   \doubleLine
   \RL{\wncontr, \occontr{}}
   \UIC{\oc\A_1, \Box\A_2, \Box\A_2, \A_3\vdash \B_1, \wn\B_2, \lozenge\B_3, \lozenge\B_3}
   \doubleLine
   \RL{\diacontr, \boxcontr}
   \UIC{\oc\A_1, \Box\A_2, \A_3\vdash \B_1, \wn\B_2, \lozenge\B_3}
   \DP} 
   \qquad
   \\[2ex]
   \begin{tikzpicture}[overlay,remember picture,-,line cap=round,line width=0.1cm]
   \draw[rounded corners, smooth=2,blue, opacity=.3] ([xshift=1mm, yshift=1mm] pic cs:contrPrincipRed21) to ([xshift=2mm, yshift=1mm] pic cs:contrPrincipRed22);
   \draw[rounded corners, smooth=2,blue, opacity=.3] ([xshift=1mm, yshift=1mm] pic cs:contrPrincipRed21l) to ([xshift=2mm, yshift=1mm] pic cs:contrPrincipRed22);
   \draw[rounded corners, smooth=2,blue, opacity=.3] ([xshift=1mm, yshift=1mm] pic cs:contrPrincipRed31) to ([xshift=2mm, yshift=1mm] pic cs:contrPrincipRed22);
   \draw[rounded corners, smooth=2,blue, opacity=.3] ([xshift=1mm, yshift=1mm] pic cs:contrPrincipRed31l) to ([xshift=2mm, yshift=1mm] pic cs:contrPrincipRed22);
   \draw[rounded corners, smooth=2,red, opacity=.3] ([xshift=1mm, yshift=1mm] pic cs:contrPrincipRed11) to ([xshift=2mm, yshift=1mm] pic cs:contrPrincipRed12);
   \draw[rounded corners, smooth=2,red, opacity=.3] ([xshift=1mm, yshift=1mm] pic cs:contrPrincipRed11l) to ([xshift=2mm, yshift=1mm] pic cs:contrPrincipRed12);
   \draw[rounded corners, smooth=2,red, opacity=.3] ([xshift=1mm, yshift=1mm] pic cs:contrPrincipRed11) to ([xshift=-1mm, yshift=1mm] pic cs:contrPrincipRed12p);
   \draw[rounded corners, smooth=2,red, opacity=.3] ([xshift=1mm, yshift=1mm] pic cs:contrPrincipRed11l) to ([xshift=-1mm, yshift=1mm] pic cs:contrPrincipRed12p);
   \draw[rounded corners, smooth=2,red, opacity=.3] ([xshift=1mm, yshift=1mm] pic cs:contrPrincipRed11) to ([xshift=0mm, yshift=1mm] pic cs:contrPrincipRed12pl);
   \draw[rounded corners, smooth=2,red, opacity=.3] ([xshift=1mm, yshift=1mm] pic cs:contrPrincipRed11l) to ([xshift=0mm, yshift=1mm] pic cs:contrPrincipRed12pl);
\end{tikzpicture}
   \hspace{-3.2cm}{\AIC{\C_{\A,\B}\tikzmark{wkPrincipRedab}}
   \AIC{\A\vdash \B}
   \RL{\delta_\wk}
   \UIC{\tikzmark{wkPrincipReda}\A\vdash \delta A, \B\tikzmark{wkPrincipRedb}}
   \AIC{\ocboxProofs_{\delta A}\tikzmark{wkPrincipRedocP}}
   \RL{\rmcutpar}
   \TIC{\tikzmark{wkPrincipReda1}\oc\A_1, \tikzmark{wkPrincipReda2}\Box\A_2, \tikzmark{wkPrincipReda3}\A_3\vdash \B_1\tikzmark{wkPrincipRedb1}, \wn\B_2\tikzmark{wkPrincipRedb2}, \lozenge\B_3\tikzmark{wkPrincipRedb3}}
   \DP}
   ~\rightsquigarrow~
   {\AIC{\C_{\A, \B}}
   \AIC{\A\vdash \B}
   \RL{\mathsf{mcut(\iota', \perp\!\!\!\perp')}}
   \BIC{\A_3\vdash\B_3}
   \RL{\wnwk, \ocwk}
   \doubleLine
   \UIC{\oc\A_1, \A_3\vdash \B_1, \wn\B_2}
   \doubleLine
   \RL{\diawk, \boxwk}
   \UIC{\oc\A_1, \Box\A_2, \A_3\vdash \B_1, \wn\B_2, \lozenge\B_3}
   \DP}
   \\[2ex]
   \begin{tikzpicture}[overlay,remember picture,-,line cap=round,line width=0.1cm]
      \draw[rounded corners, smooth=2,blue, opacity=.3] ([xshift=1mm, yshift=1mm] pic cs:wkPrincipReda1) to ([xshift=-1mm, yshift=1mm] pic cs:wkPrincipRedocP);
      \draw[rounded corners, smooth=2,blue, opacity=.3] ([xshift=1mm, yshift=1mm] pic cs:wkPrincipReda2) to ([xshift=-1mm, yshift=1mm] pic cs:wkPrincipRedocP);
      \draw[rounded corners, smooth=2,blue, opacity=.3] ([xshift=-1mm, yshift=1mm] pic cs:wkPrincipRedb2) to ([xshift=-1mm, yshift=1mm] pic cs:wkPrincipRedocP);
      \draw[rounded corners, smooth=2,blue, opacity=.3] ([xshift=-1mm, yshift=1mm] pic cs:wkPrincipRedb3) to ([xshift=-1mm, yshift=1mm] pic cs:wkPrincipRedocP);
      \draw[rounded corners, smooth=2,red, opacity=.3] ([xshift=1mm, yshift=1mm] pic cs:wkPrincipReda3) to ([xshift=-2mm, yshift=1mm] pic cs:wkPrincipRedab);
      \draw[rounded corners, smooth=2,red, opacity=.3] ([xshift=-1mm, yshift=1mm] pic cs:wkPrincipRedb1) to ([xshift=-2mm, yshift=1mm] pic cs:wkPrincipRedab);
      \draw[rounded corners, smooth=2,red, opacity=.3] ([xshift=1mm, yshift=1mm] pic cs:wkPrincipReda3) to ([xshift=1mm, yshift=1mm] pic cs:wkPrincipReda);
      \draw[rounded corners, smooth=2,red, opacity=.3] ([xshift=-1mm, yshift=1mm] pic cs:wkPrincipRedb1) to ([xshift=1mm, yshift=1mm] pic cs:wkPrincipReda);
      \draw[rounded corners, smooth=2,red, opacity=.3] ([xshift=1mm, yshift=1mm] pic cs:wkPrincipReda3) to ([xshift=-1mm, yshift=1mm] pic cs:wkPrincipRedb);
      \draw[rounded corners, smooth=2,red, opacity=.3] ([xshift=-1mm, yshift=1mm] pic cs:wkPrincipRedb1) to ([xshift=-1mm, yshift=1mm] pic cs:wkPrincipRedb);
   \end{tikzpicture}
   \hspace{-3.2cm}{\AIC{\A_1\vdash A, \B_1}
   \RL{\wnde}
   \UIC{\A_1\vdash \wn A, \B_1}
   \AIC{\oc\A_2, \Box\A_3, A\vdash \wn\B_2, \lozenge\B_3}
   \RL{\wnpromloz}
   \UIC{\oc\A_2, \Box\A_3, \wn A\vdash \wn\B_2, \lozenge\B_3}
   \AIC{\C}
   \RL{\rmcutpar}
   \TIC{\A\vdash \B}
   \DP}
   \qquad\qquad\qquad\qquad
\\
   \quad\rightsquigarrow\quad
   {\AIC{\A_1\vdash A, \B_1}
   \AIC{\oc\A_2, \Box\A_3, A\vdash \wn\B_2, \lozenge\B_3}
   \AIC{\C}
   \RL{\rmcutpar}
   \TIC{\A\vdash\B}
   \DP}
\qquad\qquad\qquad  \end{align*}
}
   \caption{One side of the \muLLmodinf{} exponential \& modal principal cut-elimination steps -- in all these proofs, $\delta\in\{\wn, \lozenge\}$}\label{fig:mullmodinfexpprincipcutstep}
   \end{figure*}

\begin{defi}[Reduction sequence]
\label{def:infty-red-seq}
A reduction sequence $(\pi_i)_{i\in 1+\lambda}$ ($\lambda\in\omega+1$) is a $\rightsquigarrow$ sequence s.t. $\pi_0$ contains at most 
one $\mathsf{(mcut)}$ rule per branch.
\end{defi}

We aim to prove that each reduction sequence converges to a cut-free proof.
However, this theorem is not true as stated, even for infinite reduction sequences: one could apply infinitely many reductions only on some part of the proof, without reducing cuts in another part of the proof. Therefore, we need to be more precise, which motivates the following definition, directly borrowed from~\cite{LICS22,CSL16} (the notion of residual is the usual one from rewriting):
\begin{defi}[Fair reduction sequences \cite{LICS22,CSL16}]
\label{def:fair-red-seq}
    A \mcut-reduction sequence $(\pi_i)_{i\in\omega}$ is \emph{fair} if, for each $\pi_i$ such that there is a reduction $\mathcal{R}$ to a proof $\pi'$, there exists a $j>i$ such that $\pi_j$ does not contain any residual of $\mathcal{R}$.
\end{defi}

We can now state our cut-elimination theorem:
\begin{thm}[Cut-elimination for \muLLmodinf]
    \label{muLLmodinfCutElim}
    Each fair $\mathsf{(mcut)}$-reduction se-\linebreak quence of \muLLmodinf{} valid proofs converges to a cut-free valid proof.
\end{thm}
To prove it, we will translate formulas, proofs, and (\mcut)-steps of \muLLmodinf{} into \muLLinf{} and use the following cut-elimination result from~\cite{TABLEAUX23}:
\begin{thm}[Cut-elimination for \muLLinf~\cite{TABLEAUX23}]
    \label{muLLinfCutElim}
    Every fair $(\mcut)$-reduction \linebreak sequence of \muLLinf{} valid proofs converges to a cut-free valid proof.
\end{thm}

\begin{rem}
In~\cite{TABLEAUX23}, exponential formulas, proofs, and cut-steps are encoded into \muMALLinf{}.
We could have encoded \muLLmodinf{} modalities into \muMALLinf{},
replaying the proof of~\cite{TABLEAUX23} for cut-elimination.
However, using the \muLLinf{} cut-elimination theorem makes the result more modular and adaptable to future extensions of \muLLinf{} validity conditions or cut-elimination variants.
\end{rem}

\subsection{Translation of \muLLmodinf{} into \muLLinf{}}

We provide a translation of \muLLmodinf{} into \muLLinf{}:
\defref{See more details  in \Cref{app:circTranslationDef}.}
\begin{defi}[Translation of \muLLmodinf{} into \muLLinf]
    \label{defi:circTranslation}
The translation of formulas is defined inductively:
\begin{itemize}
    \item Translations of $\lozenge$ and $\Box$-formulas:
    \hfill $(\lozenge A)^{\circ} := \wn A^{\circ}\quad \text{and} \quad(\Box A)^{\circ} := \oc A^{\circ}. $
    \item Translations of atomic and unit formulas and variables $f$:
    \hfill $ f^\circ := f. $
    \item Translations of fixed-point:
    \hfill $ (\delta X. F)^\circ := \delta X. F^\circ$
    (with $\delta\in\{\mu,\nu\}$).
    \item Translations of other connectives: \hfill
    $ c(A_1, \dots, A_n)^\circ := c(A_1^\circ, \dots, A_n^\circ).$
\end{itemize}

Translations of structural rules for modalities, $(\diacontr)$, $(\diawk)$, $(\boxcontr)$ and $(\boxwk)$ are respectively $(\wncontr)$, $(\wnwk)$, $(\occontr{})$ and $(\ocwk)$.
%
Translations for the promotions $(\ocpromloz)$ and $(\wnpromloz)$ are respectively $(\ocprom)$ and $(\wnprom)$.
Translations of the modal rules are given by:
$${\footnotesize
\AIC{\A\vdash A, \B}
\RL{\boxprom}
\UIC{\Box\A\vdash \Box A, \lozenge\B}
\DP\quad\rightsquigarrow\quad
\AIC{\A^\circ\vdash A^{\circ}, \B^\circ}
\doubleLine
\RL{\ocde^{\#(\A)}, \wnde^{\#(\B)}}
\UIC{\oc \A^\circ\vdash A^{\circ}, \wn\B^{\circ}}
\RL{\ocprom}
\UIC{\oc \A^\circ\vdash \oc A^{\circ}, \wn\B^{\circ}}
\DP}$$
$${\footnotesize
\AIC{\A, A\vdash\B}
\RL{\diaprom}
\UIC{\Box\A, \lozenge A\vdash\lozenge\B}
\DP\quad\rightsquigarrow\quad
\AIC{\A^\circ, A^{\circ}\vdash\B^\circ}
\doubleLine
\RL{\ocde, \wnde}
\UIC{\oc \A^\circ, A^{\circ}\vdash\wn\B^{\circ}}
\RL{\ocprom}
\UIC{\oc \A^\circ, \wn A^{\circ}\vdash\wn\B^{\circ}}
\DP}
$$
Translations of other inference rules $(r)$ are $(r)$ themselves.
Translations of pre-proofs are defined co-inductively using translations of rules.
\end{defi}

The translation preserves the validity of pre-proof in both directions:
\begin{lem}[Validity robustness to $(-)^\circ$ translation]
\label{mumdocircValidityRobustness}
Let $\pi$ be a \muLLmodinf{} pre-proof. The proof $\pi$ is valid if and only if the proof $\pi^\circ$ is valid.
\end{lem}
\begin{proof}
Let $B$ be a branch of $\pi$. We have that $B$ is validated by a thread $(A_i)$ if and only if $B^\circ$ is validated by $(A^\circ_i)$, as the minimal recurring fixed point formula is a $\nu$ on the right (resp. $\mu$ on the left) in $(A_i)$ if and only if it is in $(A^\circ_i)$.
\end{proof}

Finally, we need to ensure that $(\mcut)$-reduction sequences are robust under this translation. In our proof of the final theorem, we also need one-step reduction rules to be simulated by a finite number of reduction steps in the translation.
\begin{lem}
\label{mumodredSeqTranslationFiniteness}
Consider a \muLLmodinf{} reduction step $\pi_0\rightsquigarrow\pi_1$. There exist a finite number of \muLLinf{} proofs $\theta_0, \dots, \theta_n$ such that:
\hfill
$\pi_0^\circ = \theta_0 \redseq \theta_1\redseq\dots\redseq\theta_{n-1} \redseq \theta_n = \pi_1^{\circ}.$
\end{lem}
\begin{proof}[Proof sketch]
\proofref{See~\Cref{app:mumodredSeqTranslationFiniteness} for full proof details.}
    Reductions from the non-exponential part of $\muLLmodinf$ translate easily to one step of reduction in \muLLinf.
    The same is true for the exponential part, except for the commutation of the modal rule.
    The translation of the left proof of this step is of the form (we only consider the case of $\diaprom$; $\boxprom$ is similar):
        $$
\scalebox{1}        {\small\AIC{\pi_i^\circ}
        \noLine
        \UIC{\oc\A_i^\circ  \vdash A_i^\circ, \B_i^\circ}
        \RL{\ocde, \wnde}
        \doubleLine
        \UIC{\oc\A_i^\circ\vdash A_i^\circ, \wn\B_i^\circ}
        \RL{\ocprom}
        \UIC{\oc\A_i^\circ\vdash \oc A_i^\circ, \wn\B_i^\circ}
        \AIC{\pi_{j}^\circ}
        \noLine
        \UIC{\oc\A_{j}^\circ, A_{j}^\circ\vdash\B_{j}^\circ}
        \RL{\ocde, \wnde}
        \doubleLine
        \UIC{\oc\A_{j}^\circ, A_{j}^\circ\vdash\wn\B_{j}^\circ}
        \RL{\wnprom}
        \UIC{\oc\A_{j}^\circ, \wn A_{j}^\circ\vdash \wn\B_{j}^\circ}
        \AIC{\text{with }1\leq i \leq n ~\&}
         \noLine \UIC{n+1 \leq j\leq n+m}
        \RL{\rmcutpar}
        \TIC{\vdash \oc A^\circ, \wn\A^\circ}
        \DP}
        $$
    Here, we notice that for each dereliction on a cut-formula, there exists a corresponding promotion that will be erased by a dereliction/promotion key-case. The first promotion will therefore commute under the cut, and then each dereliction on formulas of the conclusion will commute as well. Each dereliction and each promotion on cut-formulas will be erased, giving us the correct translation
    .
\end{proof}

Now that we know that a step of (\mcut)-reduction in \muLLmodinf{} translates to one or more \muLLinf{} (\mcut)-reduction steps, it is easy to translate each reduction sequence of \muLLmodinf{} into a reduction sequence of \muLLinf{}. However, to use the cut-elimination theorem of \muLLinf{}, we need the reduction sequence to be fair. The purpose of the following lemma is to control the fairness of the translated reduction sequence:
\begin{lem}[Completeness of the (\mcut)-reduction system]
\label{mumodmcutonestepCompleteness}
Let $\pi$ and $\pi'$ be two \muLLmodinf{} proofs.
If there is a \muLLinf{}-redex $\mathcal{R}$ sending $\pi^\circ$ to $\pi'$, then there is also a $\muLLmodinf$-redex $\mathcal{R}'$ sending $\pi$ to a proof $\pi''$, such that in the translation of $\mathcal{R}'$, $\mathcal{R}$ is reduced.
\end{lem}
\begin{proof}
    We only prove the exponential cases, as the non-exponential cases are immediate. We have several cases:
    \begin{itemize}
        \item If the case is the commutative step of a weakening (resp. contraction, resp. dereliction) $(r)$, it necessarily means that $(r)$ is the translation of a rule $(r')$ being a contraction (resp. weakening, resp. dereliction) which is also on top of a (\mcut) in $\pi$. We can take $\mathcal{R}'$ as the step commuting $(r')$ under the cut.
        \item If it is a principal case on a contraction or a weakening $(r)$ on a formula $\wn A$ (resp. $\oc A$), it means that each proof cut-connected to $\wn A$ (resp. $\oc A$) ends with a promotion. As $\pi^\circ$ is the translation of a $\muLLmodinf$-proof, it means that $(r)$ is the translation of a weakening or contraction rule $(r')$ on a formula $\wn A'$ (resp. $\oc A'$) or $\lozenge A'$ (resp. $\Box A'$) on top of a (\mcut). It also means that all the proofs cut-connected to these formulas are promotions or modal rules (no other rules than a modal rule or a promotion in $\muLLmodinf$ translates to a derivation ending with a promotion). Therefore, $\mathcal{R'}$ is the principal case on $(r')$.
        \item If it is a dereliction/promotion key case, the dereliction (resp. promotion) is the translation of a dereliction (resp. promotion). We take $\mathcal{R}'$ to be the redex formed by these two rules.
        \item If it is the commutative step of a promotion $(r)$, it means that all the proofs of the contexts of the (\mcut) are promotion rules.
        This means that all these rules come from the translation of promotion or modal rules.
        We need to ensure that each (\mcut) with a context full of promotions or modal rules is covered by the (\mcut)-reductions of $\muLLmodinf$:
        \begin{itemize}
            \item The commutation of $(\ocpromloz)$ (or $(\wnpromloz)$) is covered by the first commutative case in \Cref{fig:mullmodinfexpcommcutstep}.
            \item If it is a modal rule that is ready to be commuted, then other rules are necessarily modal rules and therefore is covered by the second commutative case in \Cref{fig:mullmodinfexpcommcutstep}.
        \end{itemize}
    \end{itemize}
\end{proof}

We use the two previous lemmas, to prove the following:
\begin{coro}
\label{mumodFairredSeqTranslationFiniteness}
If $\lambda\in\omega+1$ \& $(\pi_i)_{i\in1+\lambda}$ is a  fair \muLLmodinf{} reduction sequence, then:
\begin{itemize}
\item there exists a fair \muLLinf{} reduction sequence $(\theta_i)_{i\in1+\lambda'}$ with $\lambda'\in\omega+1$;
\item there exists a sequence of strictly increasing $(\varphi(i))_{i\in1+\lambda}$ of elements of $1+\lambda'$;
\item such that for each $i$, $\theta_{\varphi(i)}=\pi_i^\circ$.
\end{itemize}
\end{coro}
\begin{proof}
    We construct the sequence by induction on the steps of $(\pi_i)_{i\in1+\lambda}$.
    \begin{itemize}
        \item For $i=0$: $\theta_0=\pi_0^\circ$ and $\varphi(0)=0$:
    
        \item For $i+1$, suppose we have defined both sequences up to rank $i$. We use \Cref{mumodredSeqTranslationFiniteness} on the step $\pi_i \rightsquigarrow \pi_{i+1}$ to get a finite sequence of reduction $\pi_i = \theta'_0 \rightsquigarrow \dots \rightsquigarrow \theta'_n = \pi_{i+1}$.
        We then construct both sequences by setting $\varphi(i+1):=\varphi(i)+n, \quad \theta_{\varphi(i)+j} := {\theta'}_j$ (for $j\in\llbracket 0, n\rrbracket$).
    \end{itemize}
    Fairness of $(\theta_i)_{i\in1+\lambda'}$ follows from \Cref{mumodmcutonestepCompleteness} and from the fact that after the translation of an (\mcut)-step, $\pi^\circ\rightsquigarrow {\pi'}^\circ$, each residual of a redex $\mathcal{R}$ of $\pi^\circ$, is contained in the translations of residuals of the associated redex $\mathcal{R}'$ of \Cref{mumodmcutonestepCompleteness}.
\end{proof}

\subsection{Cut-elimination for  \muLLmodinf{}}
Finally, we can prove the main theorem of the section:
\begin{proof}[of \Cref{muLLmodinfCutElim}]
    Let $(\pi_i)_{i\in 1+\lambda}$ be a fair \muLLmodinf{} reduction sequence.
    We use \Cref{mumodFairredSeqTranslationFiniteness} and get a fair \muLLinf{} reduction sequence $(\theta_i)_{i\in1+\lambda'}$ and a sequence $(\varphi(i))_{i\in1+\lambda}$ of natural numbers.
    By \Cref{muLLinfCutElim}, we know that $(\theta_i)_{i\in\omega}$ converges to a cut-free proof $\theta$ of \muLLinf{}.
    Now suppose for the sake of contradiction that $(\pi_i)_{i\in 1+\lambda}$ does not converge to an $(\mcut)$-free pre-proof.
    There is a $j$ and a path $p$ such that for each proof $\pi_{j'}$, with $j'\geq j$, there is an $(\mcut)$-rule at the end of path $p$. This means that the translation of $p$ leads to an $(\mcut)$ for each proof $\theta_{j'}$ with $j'\geq\varphi(j)$, contradicting the convergence of $(\theta_j)$ to a cut-free proof.
    Moreover $(\pi_i)$ converges to a pre-proof $\pi$ such that $\pi^\circ=\theta$, since $\theta_{\varphi(j)}$ equals $\pi_j^\circ$ under multicuts.
    Since $\theta$ is cut-free, $\pi$ is both valid and cut-free by \Cref{mumdocircValidityRobustness}.
\end{proof}

\section{Cut-elimination of \muLKmodinf{}}\label{section:muLKmodinfcutElim}

\begin{figure}[t]
    \centering
    \begin{subfigure}{.4\textwidth}
    $
    \AIC{{\color{blue}\wn\trans{\A}}\vdash \wn\trans{A}, \wn\trans{\B}}
    \RL{\boxprom}
    \UIC{{\color{blue}\Box\wn\trans{\A}}\vdash\Box \wn\trans{A}, \lozenge\wn\trans{\B}}
    \RL{\wnde, \ocpromloz}
    \UIC{{\color{blue}\Box\wn\trans{\A}}\vdash\wn\oc\Box \wn\trans{A}, \lozenge\wn\trans{\B}}
    \doubleLine
    \RL{\wnde, \ocpromloz, {\color{blue}\ocde}}
    \UIC{{\color{blue}\oc\Box\wn\trans{\A}}\vdash \wn\oc\Box \wn\trans{A}, \wn\oc\lozenge\wn\trans{\B}}
    \DP
    $
    \caption{Failed linear translation of the $\Box$ rule with non-empty antecedent}
    \label{fig:leftSideTransFail}
    \end{subfigure}\qquad
    \begin{subfigure}{.4\textwidth}
    $
    \AIC{\trans{\A}\vdash\wn\trans{A}, \wn\trans{\B}}
    \RL{{\color{blue} \ocde, \wnprom}}
    \doubleLine
    \UIC{{\color{blue}\oc}\wn\trans{\A}\vdash \wn\trans{A}, \wn\trans{\B}}
    \RL{{\color{blue}\ocpromloz}}
    \UIC{{\color{blue}\oc}\wn\trans{\A}\vdash {\color{blue}\oc}\wn\trans{A}, \wn\trans{\B}}
    \RL{\boxprom}
    \UIC{\Box{\color{blue}\oc}\wn\trans{\A}\vdash\Box {\color{blue}\oc}\wn\trans{A}, \lozenge\wn\trans{\B}}
    \RL{\wnde, \ocpromloz}
    \UIC{\Box{\color{blue}\oc}\wn\trans{\A}\vdash\wn\oc\Box {\color{blue}\oc}\wn\trans{A}, \lozenge\wn\trans{\B}}
    \doubleLine
    \RL{\wnde, \ocpromloz, \ocde}
    \UIC{\oc\Box{\color{blue}\oc}\wn\trans{\A}\vdash \wn\oc\Box {\color{blue}\oc}\wn\trans{A}, \wn\oc\lozenge\wn\trans{\B}}
    \DP
    $
    \caption{Successful translation of the $\Box$ rule with non-empty antecedent}
    \label{fig:LeftSideTransSuccess}
    \end{subfigure}
    \caption{Translating the Modal Rule with Non-Empty Antecedent}
    \end{figure}

We extend the linear translation of \muLKinf{} from~\cite{TABLEAUX23} to a translation from \muLKmodinf{} into \muLLmodinf.
The tentative translation of the modal formulas studied at the beginning of \Cref{section:ModalMu} actually only works with an empty antecedent. Adapting our example by adding an antecedent to it: $
\AIC{{\color{blue}\A}\vdash F, \B}
\RL{\boxprom}
\UIC{{\color{blue}\Box\A}\vdash \Box F, \lozenge\B}
\DP
$ results in a derivation attempt as shown in \Cref{fig:leftSideTransFail}.
If $\A$ contains more than two formulas, $(\wnprom)$ cannot be applied. However, adding an $\oc$-connective in the translation of $\Box$-formulas as $\trans{\Box A} := \oc\Box{\color{blue}\oc}\wn \trans{A}$, allows us to complete the derivation (see \Cref{fig:LeftSideTransSuccess}).

Based on this, we define a translation from \muLKmodinf{} into \muLLmodinf{}:
\begin{defi}[Linear translation of \muLKmodinf]
The \emph{translation $\trans{(-)}$} from\linebreak \muLKmodinf{} formulas to \muLLmodinf{} formulas  is defined by induction on formulas:
\vspace{-.3cm
}\begin{align*}
\trans{(A_1\rightarrow A_2)} & :=  \oc(\wn\trans{A_1} \multimap \wn\trans{A_2}) &
\trans{X} & := \oc X
&\trans{(\mu X. A)} & :=  \oc\mu X. \wn\trans{A}\\
\trans{(A_1\land A_2)} & :=  \oc(\wn\trans{A_1} \with \wn\trans{A_2})
&\trans{\true} & :=  \oc\top
&\trans{(\nu X. A)} & :=  \oc\nu X. \wn\trans{A}\\
\trans{(A_1\lor A_2)} & :=  \oc(\wn\trans{A_1} \oplus \wn\trans{A_2}) & \trans{\false} & := \oc 0 & \trans{(\lozenge A)} & := \oc\lozenge\wn \trans{A}\\
\trans{(A^\perp)} & := \oc (\wn\trans{A})^\perp & \trans{a} & := \oc a &  \trans{(\Box A)} & :=\oc\Box\oc\wn\trans{A}
\end{align*}
Sequents are translated as 
$\trans{(\A\vdash \B)} := \trans{\A}\vdash \wn \trans{\B}.$
\end{defi}
The following property must be kept in mind when defining rule translations and is proved by an induction the formulas of \muLKmodinf:
\begin{prop}
    For any \muLKmodinf{} formula $A$, $\trans{A}$ is an $\oc$-formula.
\end{prop}
We provide the translation of modal rules in \Cref{fig:LinearRuleTranslationMod}.
\defref{See \Cref{app:LinTranslation} for the translation of the remaining rules.}
\begin{figure*}[t]
    \centering
    $\begin{array}{c|c}
    \AIC{\B\vdash F, \A}
    \RL{\boxprom}
    \UIC{\Box\B\vdash\Box F,\lozenge\A}
    \DP
\qquad    \rightsquigarrow 
&
    \AIC{\B, F\vdash \A}
    \RL{\diaprom}
    \UIC{\Box\B, \lozenge F\vdash\lozenge\A}
    \DP
\qquad    \rightsquigarrow 
    \\[3ex]
\AIC{\trans{\B}\vdash\wn\trans{F}, \wn\trans{\A}}
    \doubleLine
    \RL{(\ocde, \wnprom)^{\#(\B)}}
    \UIC{\oc\wn\trans{\B}\vdash\wn\trans{F}, \wn\trans{\A}}
    \RL{\ocprom}
    \UIC{\oc\wn\trans{\B}\vdash\oc\wn\trans{F}, \wn\trans{\A}}
    \RL{\boxprom}
    \UIC{\Box\oc\wn\trans{\B}\vdash\Box\oc\wn\trans{F}, \lozenge\wn\trans{\A}}
    \doubleLine
    \RL{\ocde^{\#(\B)}}
    \UIC{\oc\Box\oc\wn\trans{\B}\vdash\Box\oc\wn\trans{F}, \lozenge\wn\trans{\A}}
    \doubleLine
    \RL{(\wnde, \ocpromloz)^{(\#(\A)+1)}}
    \UIC{\oc\Box\oc\wn\trans{\B}\vdash\wn\oc\Box\oc\wn\trans{F}, \wn\oc\lozenge\wn\trans{\A}}
    \DP
 \quad
&
 \quad    \AIC{\trans{\B}, \trans{F}\vdash \wn\trans{\A}}
    \doubleLine
    \RL{(\ocde, \wnprom)^{\#(\B)}}
    \UIC{\oc\wn\trans{\B}, \trans{F}\vdash \wn\trans{\A}}
    \RL{\wnprom}
    \UIC{\oc\wn\trans{\B}, \wn\trans{F}\vdash \wn\trans{\A}}
    \RL{\diaprom}
    \UIC{\Box\oc\wn\trans{\B}, \lozenge\wn\trans{F}\vdash \lozenge\wn\trans{\A}}
    \doubleLine
    \RL{\ocde^{\#(\B)+1}}
    \UIC{\oc\Box\oc\wn\trans{\B},\oc\lozenge\wn\trans{F}\vdash \lozenge\wn\trans{\A}}
    \doubleLine
    \RL{(\wnde, \ocpromloz)^{\#(\A)}}
    \UIC{\oc\Box\oc\wn\trans{\B}, \oc\lozenge\wn\trans{F}\vdash \wn\oc\lozenge\wn\trans{\A}}
    \DP
     \quad
    \end{array}$
    \caption{Linear translation of the modal rules}\label{fig:LinearRuleTranslationMod}
\end{figure*}
We define translations of proofs coinductively using rule translations.
As the smallest formula of a totally ordered set of translations is the translation of the smallest formula, and branches of $\trans{\pi}$ contain translations of threads from $\pi$ and vice-versa, we have:
\begin{lem}[Robustness of $\trans{(-)}$ to validity]
A pre-proof $\pi$ is valid if and only if $\trans{\pi}$ is valid.
\end{lem}

\begin{lem}[Composition of $\sk{-}$ and of $\trans{(-)}$]
    \label{linearSkComp}
\proofref{See proof in \Cref{sec:app:linearSkComp}}
For any \muLKmodinf{} pre-proof $\pi$, $\sk{\trans{\pi}}$ is equal to $\pi$.
\end{lem}

We define our rewriting system using the $\sk{-}$ translation:
\begin{defi}[(\mcut)-rewriting system of \muLKmodinf{}]
The $\mathsf{(mcut)}$-rewriting \linebreak system of \muLKmodinf{} is defined as the (\mcut)-system obtained from the \muLLmodinf{} (\mcut)-system by discarding the linear information of proofs in this system.
\end{defi}

Finally, we have the following theorem:
\begin{thm}
The 
reduction system of \muLKmodinf{} is infinitary weakly-normalizing.
\end{thm}
\begin{proof}
Consider a \muLKmodinf{} proof $\pi$ and a fair reduction sequence $\sigma_{\mathsf{L}}$ from $\trans{\pi}$. By \Cref{muLLmodinfCutElim}, $\sigma_{\mathsf{L}}$ converges to a cut-free \muLLmodinf{} proof. 
By applying $\sk{-}$ to each proof in the sequence, we obtain a sequence of valid \muLKmodinf{} proofs such that either $\sk{\pi_i} = \sk{\pi_{i+1}}$ or $\sk{\pi_i}$ reduces to $\sk{\pi_{i+1}}$ with one step of \muLKmodinf{} \mcut-reduction.
We obtain a \muLKmodinf{} cut-reduction sequence $\sigma_{\mathsf{K}}$ that is infinite and converges to a valid, cut-free \muLKmodinf{} proof.
\end{proof}

    \section{Conclusion}

We have introduced \muLLmodinf{}, a linear version of the modal $\mu$-calculus, along with its circular and non-wellfounded proof systems. We have proved a cut-elimination theorem for fair cut-elimination reduction sequences, generalizing previous results on the non-wellfounded proof theory of linear logic. 
Through a linear translation of the circular and non-wellfounded proof systems for the modal $\mu$-calculus (\muLKmodinf{}) to \muLLmodinf{}, we have obtained a cut-elimination theorem for the non-wellfounded sequent calculus of the modal $\mu$-calculus.

In our opinion, this work presents a new and interesting application of linear logic to the modal $\mu$-calculus, developing proof theories in both domains and highlighting the potential for cross-fertilization for the two communities. Indeed, this constitutes the first full syntactic cut-elimination theorem for a proof system modeling the full modal $\mu$-calculus.

Furthermore, the fine-grained cut-elimination inherited from linear logic \linebreak opens up the possibility of developing a non-trivial cut-elimination equivalence on \muLKmodinf{} proofs. This, in turn, could lead to the design of a denotational semantics for proofs of modal $\mu$-calculus, a question that was previously out of reach due to both the lack of a syntactic cut-elimination theorem and the absence of structure in proofs of the modal $\mu$-calculus. 

\medskip

However, this is not the first work studying both linear logic and modal logic. Notably, studies on modal logic S4 through linear logic have addressed both its intuitionistic fragment, proving cut-admissibility~\cite{FukudaYoshimizu19}, and its classical version~\cite{MartiniMasini94,Schellinx96}.
Our work differs from these studies in two key aspects: (i) We employ a non-wellfounded setting. (ii) The logic S4 treats modalities as linear logic exponentials with a standard promotion and dereliction, making it a sub-system of the subexponentials~\cite{Nigam09}. Our work is different because we use a functorial promotion for modal rules, which requires us to examine the interaction between functorial and non-functorial promotion.

\medskip

%

An important direction for future work is to explore whether our linear-logical modal mu-calculus can be adapted to wellfounded proof-systems of linear logic with fixed-points in a \muLLmod{} sequent calculus. We aim to investigate if our methodology can be extended to obtain a cut-elimination theorem for the finitary modal $\mu$-sequent calculus via a linear translation from \muLKmod{} to \muLLmod{}, building upon Baelde's results~\cite{DBLP:journals/tocl/Baelde12}. This question presents significant challenges due to the complex structure of fixed-point rules in finitary \muLL{}. While omitting linear information from \muLLmod{} proofs should yield \muLKmod{} proofs that simulate cut-elimination, designing a linear translation from \muLKmod{} to \muLLmod{} that commutes with cut-elimination, as shown in this paper, remains a non-trivial task.

From the linear logic-theoretic point of view, our system \muLLmodinf{} can be viewed as a linear logic with two sets 
of exponential modalities satisfying different structural rules and exponentials. This is akin to so-called {\it light 
logics}~\cite{lll,sll}, that are variants of linear logics developed by taming the power of exponential modalities in 
order to control the complexity of cut-elimination (for instance constraining the $\wn$-context of a promotion to be 
immediately derelicted after a promotion ensures that typable programs have at most elementary complexity~\cite{lll}).

Building on these ideas, we propose to develop a uniform framework for proving cut elimination in wellfounded proof systems of linear logic fragments that encompass both light logics and modal logics. This framework could be inspired by the subexponential system of Nigam \& Miller~\cite{Nigam09} and the first author with Laurent~\cite{TLLA21}, but extended to include functorial promotion. Such a system would allow us to treat light logics and our linear-logical modal mu-calculus as instances of a more general linear logic system.

In a similar direction, we shall investigate whether our approach of reducing cut-elimination to that of simpler systems can be useful for the circular approaches to session types \cite{lmcs:5675,10.1007/978-3-662-45917-1_11} as well as to the study of productivity and normalization properties in functional reactive programming~\cite{DBLP:conf/popl/CaveFPP14}.


\subsubsection{Funding.} This work was partially funded by the ANR project RECIPROG, project reference ANR-21-CE48-019-01.

\subsubsection{Acknowledgments.} The authors would like to thank the anonymous reviewers for their constructive and valuable comments.

    \bibliographystyle{spmpsci}

 \bibliography{biblio}

\begin{thebibliography}{10}
\providecommand{\url}[1]{{#1}}
\providecommand{\urlprefix}{URL }
\expandafter\ifx\csname urlstyle\endcsname\relax
  \providecommand{\doi}[1]{DOI~\discretionary{}{}{}#1}\else
  \providecommand{\doi}{DOI~\discretionary{}{}{}\begingroup
  \urlstyle{rm}\Url}\fi

\bibitem{Bahareh17}
Afshari, B., Leigh, G.E.: Cut-free completeness for modal mu-calculus.
\newblock In: 2017 32nd Annual ACM/IEEE Symposium on Logic in Computer Science
  (LICS), pp. 1--12 (2017).
\newblock \doi{10.1109/LICS.2017.8005088}

\bibitem{DBLP:journals/tocl/Baelde12}
Baelde, D.: Least and greatest fixed points in linear logic.
\newblock {ACM} Trans. Comput. Log. \textbf{13}(1), 2:1--2:44 (2012).
\newblock \doi{10.1145/2071368.2071370}.
\newblock \urlprefix\url{https://doi.org/10.1145/2071368.2071370}

\bibitem{LICS22}
Baelde, D., Doumane, A., Kuperberg, D., Saurin, A.: Bouncing threads for
  circular and non-wellfounded proofs: Towards compositionality with circular
  proofs.
\newblock In: Proceedings of the 37th Annual ACM/IEEE Symposium on Logic in
  Computer Science, LICS '22. Association for Computing Machinery, New York,
  NY, USA (2022).
\newblock \doi{10.1145/3531130.3533375}.
\newblock \urlprefix\url{https://doi.org/10.1145/3531130.3533375}

\bibitem{CSL16}
Baelde, D., Doumane, A., Saurin, A.: {Infinitary Proof Theory: the
  Multiplicative Additive Case}.
\newblock In: CSL 2016, \emph{LIPIcs}, vol.~62, pp. 42:1--42:17 (2016).
\newblock \doi{10.4230/LIPIcs.CSL.2016.42}.
\newblock \urlprefix\url{http://drops.dagstuhl.de/opus/volltexte/2016/6582}

\bibitem{TLLA21}
Bauer, E., Laurent, O.: Super exponentials in linear logic.
\newblock Electronic Proceedings in Theoretical Computer Science \textbf{353},
  50–73 (2021).
\newblock \doi{10.4204/eptcs.353.3}.
\newblock \urlprefix\url{http://dx.doi.org/10.4204/EPTCS.353.3}

\bibitem{BRUNNLER12}
Brünnler, K., Studer, T.: Syntactic cut-elimination for a fragment of the
  modal mu-calculus.
\newblock Annals of Pure and Applied Logic \textbf{163}(12), 1838--1853 (2012).
\newblock \doi{https://doi.org/10.1016/j.apal.2012.04.006}.
\newblock
  \urlprefix\url{https://www.sciencedirect.com/science/article/pii/S0168007212000760}

\bibitem{BussHandbook}
Buss, S.R. (ed.): Handbook of Proof Theory.
\newblock Elsevier, New York (1998)

\bibitem{DBLP:conf/popl/CaveFPP14}
Cave, A., Ferreira, F., Panangaden, P., Pientka, B.: Fair reactive programming.
\newblock In: S.~Jagannathan, P.~Sewell (eds.) The 41st Annual {ACM}
  {SIGPLAN-SIGACT} Symposium on Principles of Programming Languages, {POPL}
  '14, San Diego, CA, USA, January 20-21, 2014, pp. 361--372. {ACM} (2014).
\newblock \doi{10.1145/2535838.2535881}.
\newblock \urlprefix\url{https://doi.org/10.1145/2535838.2535881}

\bibitem{DBLP:conf/icfp/CurienH00}
Curien, P., Herbelin, H.: The duality of computation.
\newblock In: ICFP 2000, pp. 233--243. {ACM} (2000).
\newblock \doi{10.1145/351240.351262}.
\newblock \urlprefix\url{https://doi.org/10.1145/351240.351262}

\bibitem{DanosJS97}
Danos, V., Joinet, J., Schellinx, H.: A new deconstructive logic: Linear logic.
\newblock J. Symb. Log. \textbf{62}(3), 755--807 (1997).
\newblock \doi{10.2307/2275572}.
\newblock \urlprefix\url{https://doi.org/10.2307/2275572}

\bibitem{lmcs:5675}
Derakhshan, F., Pfenning, F.: Circular proofs as session-typed processes: A
  local validity condition.
\newblock Logical Methods in Computer Science \textbf{Volume 18, Issue 2}, 8
  (2022).
\newblock \doi{10.46298/lmcs-18(2:8)2022}.
\newblock \urlprefix\url{https://lmcs.episciences.org/5675}

\bibitem{aminaphd}
Doumane, A.: On the infinitary proof theory of logics with fixed points.
\newblock Phd thesis, Paris Diderot University (2017)

\bibitem{FukudaYoshimizu19}
Fukuda, Y., Yoshimizu, A.: A linear-logical reconstruction of intuitionistic
  modal logic {S4}.
\newblock CoRR \textbf{abs/1904.10605} (2019).
\newblock \urlprefix\url{http://arxiv.org/abs/1904.10605}

\bibitem{girard87}
Girard, J.: Linear logic.
\newblock Theor. Comput. Sci. \textbf{50}, 1--102 (1987).
\newblock \doi{10.1016/0304-3975(87)90045-4}.
\newblock \urlprefix\url{https://doi.org/10.1016/0304-3975(87)90045-4}

\bibitem{lll}
Girard, J.Y.: Light linear logic \textbf{143}(2), 175--204 (1998).
\newblock \doi{10.1006/inco.1998.2700}

\bibitem{JAGER2008270}
Jäger, G., Kretz, M., Studer, T.: Canonical completeness of infinitary $\mu$.
\newblock The Journal of Logic and Algebraic Programming \textbf{76}(2),
  270--292 (2008).
\newblock \doi{https://doi.org/10.1016/j.jlap.2008.02.005}.
\newblock
  \urlprefix\url{https://www.sciencedirect.com/science/article/pii/S1567832608000209}.
\newblock Logic and Information: From Logic to Constructive Reasoning

\bibitem{DBLP:journals/tcs/Kozen83}
Kozen, D.: Results on the propositional mu-calculus.
\newblock Theor. Comput. Sci. \textbf{27}, 333--354 (1983).
\newblock \doi{10.1016/0304-3975(82)90125-6}.
\newblock \urlprefix\url{https://doi.org/10.1016/0304-3975(82)90125-6}

\bibitem{sll}
Lafont, Y.: Soft linear logic and polynomial time.
\newblock Theor. Comput. Sci. \textbf{318}(1--2), 163--180 (2004).
\newblock \doi{10.1016/j.tcs.2003.10.018}

\bibitem{MartiniMasini94}
Martini, S., Masini, A.: A modal view of linear logic.
\newblock Journal of Symbolic Logic \textbf{59}(3), 888–899 (1994).
\newblock \doi{10.2307/2275915}

\bibitem{miller-unity-computational-logic}
Miller, D.: Finding unity in computational logic.
\newblock In: Proceedings of the 2010 ACM-BCS Visions of Computer Science
  Conference, ACM-BCS '10. BCS Learning \& Development Ltd., Swindon, GBR
  (2010)

\bibitem{Mints12}
Mints, G.: Effective cut-elimination for a fragment of modal mu-calculus.
\newblock Studia Logica: An International Journal for Symbolic Logic
  \textbf{100}(1/2), 279--287 (2012).
\newblock \urlprefix\url{http://www.jstor.org/stable/41475226}

\bibitem{MintsStuder12}
Mints, G., Studer, T.: Cut-elimination for the mu-calculus with one variable.
\newblock Electronic Proceedings in Theoretical Computer Science \textbf{77},
  47–54 (2012).
\newblock \doi{10.4204/eptcs.77.7}.
\newblock \urlprefix\url{http://dx.doi.org/10.4204/EPTCS.77.7}

\bibitem{Nigam09}
Nigam, V., Miller, D.: Algorithmic specifications in linear logic with
  subexponentials.
\newblock In: PPDP 2009, p. 129–140. ACM, New York, NY, USA (2009).
\newblock \doi{10.1145/1599410.1599427}.
\newblock \urlprefix\url{https://doi.org/10.1145/1599410.1599427}

\bibitem{NIWINSKI96}
Niwiński, D., Walukiewicz, I.: Games for the $\mu$-calculus.
\newblock Theoretical Computer Science \textbf{163}(1), 99--116 (1996).
\newblock \doi{https://doi.org/10.1016/0304-3975(95)00136-0}.
\newblock
  \urlprefix\url{https://www.sciencedirect.com/science/article/pii/0304397595001360}

\bibitem{PagTor10tcs}
Pagani, M., Tortora~de Falco, L.: Strong normalization property for second
  order linear logic.
\newblock Theoretical Computer Science \textbf{411}(2), 410--444 (2010)

\bibitem{DBLP:conf/fossacs/Santocanale02}
Santocanale, L.: A calculus of circular proofs and its categorical semantics.
\newblock In: M.~Nielsen, U.~Engberg (eds.) Foundations of Software Science and
  Computation Structures, 5th International Conference, {FOSSACS} 2002. Held as
  Part of the Joint European Conferences on Theory and Practice of Software,
  {ETAPS} 2002 Grenoble, France, April 8-12, 2002, Proceedings, \emph{Lecture
  Notes in Computer Science}, vol. 2303, pp. 357--371. Springer (2002).
\newblock \doi{10.1007/3-540-45931-6\_25}.
\newblock \urlprefix\url{https://doi.org/10.1007/3-540-45931-6\_25}

\bibitem{Santocanale13}
Santocanale, L.H.M.l., Fortier, J.: {Cuts for circular proofs: semantics and
  cut-elimination}.
\newblock In: {Computer Science Logic 2013}, \emph{LIPIcs}, vol.~23, pp.
  248--262. {Schloss Dagstuhl - Leibniz-Zentrum fuer Informatik}, Torino, Italy
  (2013).
\newblock \doi{10.4230/LIPIcs.CSL.2013.248}.
\newblock \urlprefix\url{https://hal.science/hal-01260986}

\bibitem{TABLEAUX23}
Saurin, A.: {A linear perspective on cut-elimination for non-wellfounded
  sequent calculi with least and greatest fixed points}.
\newblock TABLEAUX '23. Springer (2023).
\newblock \urlprefix\url{https://hal.science/hal-04169137}

\bibitem{Schellinx96}
Schellinx, H.: A Linear Approach to Modal Proof Theory, pp. 33--43.
\newblock Springer Netherlands, Dordrecht (1996).
\newblock \doi{10.1007/978-94-017-2798-3_3}.
\newblock \urlprefix\url{https://doi.org/10.1007/978-94-017-2798-3_3}

\bibitem{stirling}
Stirling, C.: {Modal and Temporal Logics}.
\newblock In: {Handbook of Logic in Computer Science}. Oxford University Press
  (1992).
\newblock \doi{10.1093/oso/9780198537618.003.0005}.
\newblock \urlprefix\url{https://doi.org/10.1093/oso/9780198537618.003.0005}

\bibitem{10.1007/978-3-662-45917-1_11}
Toninho, B., Caires, L., Pfenning, F.: Corecursion and non-divergence in
  session-typed processes.
\newblock In: M.~Maffei, E.~Tuosto (eds.) Trustworthy Global Computing, pp.
  159--175. Springer Berlin Heidelberg, Berlin, Heidelberg (2014)

\end{thebibliography}

\clearpage

\appendix

\section{Appendix on the \Cref{section:definitionsOfKnownSequentCalculi}}
\label{app:definitionsOfKnownSequentCalculi}

\subsection{Details of definitions of Linear Logic}
\label{app:lldefdetails}
\begin{defi}[Positive and negative occurrence of a fixed-point variable]
    Let $X\in \Var$ be a fixed-point variable, one defines the fact, for $X$, to occur positively (resp. negatively) in a pre-formula by induction on the structure of pre-formulas: 
    \begin{itemize}
    \item The variable $X$ occurs positively in $X$.
    \item The variable $X$ occurs positively (resp. negatively) in ${c} (F_1, \dots, F_n)$, if there is some $1\leq i \leq n$ such that $X$ occurs positively (resp. negatively) in $F_i$ for $c\in\{\otimes, \parr, \with, \oplus, \oc, \wn\}$.
    \item The variable $X$ occurs positively (resp. negatively) in $F^\perp$ if $X$ occurs negatively (resp. positively) in $F$.
    \item The variable $X$ occurs positively (resp. negatively) in $F\multimap G$ if $X$ occurs either positively (resp. negatively) in $G$ or negatively (resp. positively) in $F$.
    \item The variable $X$ occurs positively (resp. negatively) in $\delta Y. G$ (with $Y\neq X$) if it occurs positively (resp. negatively) in $G$ (for $\delta\in\{\mu, \nu\}$).
    \end{itemize}
    
    \end{defi}

\section{Appendix on the \Cref{section:ModalMu}}

\subsection{Details on the discussion about robustness of \muLLmodinf{}}
\label{app:DetailsDiscussionRobustnessLLbox}

As said in the core of the paper, taking the $(\wnprom/\wnwk)$ principal case:
$$
\AIC{\A_1\vdash \B_1}
\RL{\wnwk}
\UIC{\A_1\vdash \wn C, \B_1}
\AIC{\oc\A_2, C\vdash\wn\B_2}
\RL{\wnprom}
\UIC{\oc\A_2\wn C\vdash \wn\B_2}
\RL{\cut}
\BIC{\A_1, \oc\A_2\vdash \B_1,\wn\B_2}
\DP
\rightsquigarrow
\AIC{\A_1\vdash \B_1}
\RL{\wnwk, \ocwk}
\doubleLine
\UIC{\A_1, \oc\A_2\vdash \B_2, \wn\B_2}
\DP
$$
and adding modal-formulas to the context, naturally requires to be able to weaken $\lozenge/\Box$-formulas (this corresponds to the reason for the design of the promotion rule in \LL{}):
$$
\AIC{\A_1\vdash \B_1}
\RL{\wnwk}
\UIC{\A_1\vdash \wn C, \B_1}
\AIC{\oc\A_2{\color{blue}, \Box\A_3}, C\vdash \wn\B_2{\color{blue}, \lozenge\B_3}}
\RL{\wnprom}
\UIC{\oc\A_2, {\color{blue}\Box\A_3,} \wn C\vdash \wn\B_2{\color{blue}, \lozenge\B_3}}
\RL{\cut}
\BIC{\A_1, \oc\A_2, {\color{blue}\Box\A_3}\vdash \B_2, \wn\B_2{\color{blue}, \lozenge\B_3}}
\DP
\rightsquigarrow
\AIC{\A_1\vdash \B_1}
\RL{\wnwk, \ocwk{\color{blue}, \boxwk, \diawk}}
\doubleLine
\UIC{\A_1, \oc\A_2{\color{blue}, \Box\A_3}\vdash \B_2, \wn\B_2{\color{blue}, \lozenge\B_3}}
\DP
$$
Moreover, the weakening on $\lozenge$ (and dually on $\Box$) is necessary to preserve the cut-elimination property, as the sequent $ \vdash \lozenge \bot, 1 $ is provable with (\cut) and without $(\diawk)$:
$$
\AIC{}
\RL{1}
\UIC{\vdash 1}
\RL{\wnwk}
\UIC{\vdash 1, \wn\lozenge \bot}
\AIC{}
\RL{\ax}
\UIC{\lozenge\bot\vdash \lozenge \bot}
\RL{\wnprom}
\UIC{\wn \lozenge \bot\vdash  \lozenge \bot}
\RL{\cut}
\BIC{\vdash \lozenge \bot, 1}
\DP
$$
but is unprovable without (\cut) and without $(\diawk)$ as we cannot apply any rules on such a sequent.
Similarly, the $(\wncontr/\wnprom)$ key-case naturally asks to be able to contract $\lozenge$-formulas.

$$\hspace{-6cm}
\AIC{\A_1\vdash \wn C, \wn C, \B_1}
\RL{\wncontr}
\UIC{\A_1\vdash \wn C, \B_1}
\AIC{\oc\A_2, {\color{blue}\Box\A_3}, C\vdash \wn\B_2, {\color{blue}\lozenge\B_3}}
\RL{\wnprom}
\UIC{\oc\A_2, {\color{blue}\Box\A_3}, \wn C\vdash \wn\B_2, {\color{blue}\lozenge\B_3}}
\RL{\cut}
\BIC{\A_1, \oc\A_2, {\color{blue}\Box\A_3}\vdash \B_1,\wn\B_2, {\color{blue}\lozenge\B_3}}
\DP
\rightsquigarrow$$
$$\hspace{2cm}
\AIC{\A_1\vdash \wn C, \wn C, \B_1}
\AIC{\oc\A_2, {\color{blue}\Box\A_3}, C\vdash \wn\B_2, {\color{blue}\lozenge\B_3}}
\RL{\wnprom}
\UIC{\oc\A_2, {\color{blue}\Box\A_3}, \wn C\vdash \wn\B_2, {\color{blue}\lozenge\B_3}}
\RL{\cut}
\BIC{\A_1, \oc\A_2, {\color{blue}\Box\A_3}\vdash\wn C, \B_1,\wn\B_2, {\color{blue}\lozenge\B_3}}
\AIC{\oc\A_2, {\color{blue}\Box\A_3}, C\vdash \wn\B_2, {\color{blue}\lozenge\B_3}}
\RL{\wnprom}
\UIC{\oc\A_2, {\color{blue}\Box\A_3}, \wn C\vdash \wn\B_2, {\color{blue}\lozenge\B_3}}
\RL{\cut}
\BIC{\A_1, \oc\A_2,\oc\A_2, {\color{blue}\Box\A_3}, {\color{blue}\Box\A_3}\vdash \B_1,\wn\B_2, \wn\B_2, {\color{blue}\lozenge\B_3}, {\color{blue}\lozenge\B_3}}
\RL{\diacontr}
\UIC{\A_1, \oc\A_2,\oc\A_2, {\color{blue}\Box\A_3}\vdash \B_1,\wn\B_2, \wn\B_2, {\color{blue}\lozenge\B_3}}
\RL{\wncontr, \occontr{}}
\doubleLine
\UIC{\A_1, \oc\A_2, {\color{blue}\Box\A_3}\vdash \B_1,\wn\B_2, {\color{blue}\lozenge\B_3}}
\DP.
$$
By examining the following proof, one can once again see the necessity of $(\diacontr)$ for preserving cut-elimination in such a system:
$$
\AIC{}
\RL{\ax}
\UIC{\lozenge 1\vdash \lozenge 1}
\RL{\wnprom}
\UIC{\wn\lozenge 1\vdash \lozenge 1}
\AIC{}
\RL{\ax}
\UIC{\vdash \lozenge 1, \Box\bot}
\RL{\wnde}
\UIC{\vdash \wn\lozenge 1, \Box\bot}
\AIC{}
\RL{\ax}
\UIC{\vdash \lozenge 1, \Box\bot}
\RL{\wnde}
\UIC{\vdash \wn\lozenge 1, \Box\bot}
\RL{\otimes}
\BIC{\vdash \wn\lozenge 1, \wn\lozenge 1, \Box \bot\otimes \Box \bot}
\RL{\wncontr}
\UIC{\vdash \wn\lozenge 1, \Box \bot\otimes \Box \bot}
\RL{\cut}
\BIC{\vdash \lozenge 1, \Box \bot\otimes \Box \bot}
\DP
.$$

The conclusion sequent is unprovable, as the only rule that can be applied on it is a $(\otimes)$, leaving us with an unprovable sequent:
$$
\AIC{}
\RL{\ax}
\UIC{\vdash \Box\bot, \lozenge 1}
\AIC{\text{unprovable}}
\noLine
\UIC{\vdash}
\RL{\bot}
\UIC{\vdash \bot}
\RL{\boxprom}
\UIC{\vdash \Box\bot}
\RL{\otimes}
\BIC{\vdash \lozenge 1, \Box \bot\otimes \Box \bot}
\DP.
$$

\subsection{Details on \sk{-}-translation}
\label{app:skTranslation}

\begin{defi}[\muLKmodinf-Skeleton]
    \label{def:app:skTranslation}
    We define skeleton of formulas by induction:\\ 
$\begin{array}{rcl rcl rcl}
\sk{1}&=&\true\quad&
\sk{F\otimes G} &=& \sk{F} \land \sk{G}\quad&
\sk{F\multimap G} &=&  \sk F \rightarrow \sk G\\
\sk{\bot} &=& \false\quad&
\sk{F\parr G} &=& \sk{F}\lor\sk{G}\quad&
\sk {F^\perp} &= &\sk{F}^\perp\\
\sk{\top} &=& \true \quad&
\sk{F\with G} &=& \sk F\land \sk G \quad&
\sk{\lozenge F} & = & \lozenge \sk{F}\\
\sk{0} &=& \false \quad&
\sk{F\oplus G} &=& \sk F \lor \sk G \quad&
\sk{\Box F} & = & \Box \sk{F}\\
\sk a &=& a \quad&
\sk{\wn F} &=& \sk F \quad&
\sk{\oc F} &=& \sk F\\
\sk X &=& X \quad&
\sk{\mu X. F} &=& \mu X. \sk F \quad&
\sk{\nu X. F} &=& \nu X. \sk F
 \end{array}$

Sequents of \muLLmodinf{} are translated to sequent of skeletons of these formulas.
Rules are translated straightforwardly by forgetting the linear information, translation are given in \Cref{fig:skNonModruletranslation,fig:skModruletranslation}.

Translations of pre-proofs are obtained co-inductively by applying rule translations.
\end{defi}

\begin{figure*}
\begin{align*}
\hspace{-3cm}\AIC{}
\RL{\ax}
\UIC{F\vdash F}
\DP
\quad&\rightsquigarrow\quad
\AIC{}
\RL{\ax}
\UIC{\sk{F}\vdash \sk{F}}
\DP\\
\hspace{-2cm}\AIC{\A_1\vdash F, \B_1}
\AIC{\A_2, F \vdash \B_2}
\RL{\cut}
\BIC{\A_1, \A_2\vdash \B_1, \B_2}
\DP\quad&\rightsquigarrow\quad
\AIC{\sk{\A_1}\vdash \sk{F}, \sk{\B_1}}
\AIC{\sk{\A_2}, \sk{F} \vdash \sk{\B_2}}
\RL{\cut}
\BIC{\sk{\A_1}, \sk{\A_2}\vdash \sk{\B_1}, \sk{\B_2}}
\DP\\
\hspace{-2cm}\AIC{\A_1\vdash F, \B_1}
\AIC{\A_2\vdash G, \B_2}
\RL{\otimes_r}
\BIC{\A_1, \A_2\vdash F\otimes G, \B_1, \B_2}
\DP\quad&\rightsquigarrow
\quad
{\tiny\AIC{\sk{\A_1}\vdash \sk{F}, \sk{\B_1}}
\doubleLine
\RL{\lkwk_l, \lkwk_r}
\UIC{\sk{\A_1}, \sk{\A_2}\vdash \sk{F}, \sk{\B_1}, \sk{\B_2}}
\AIC{\sk{\A_2}\vdash \sk{G}, \sk{\B_2}}
\doubleLine
\RL{\lkwk_l, \lkwk_r}
\UIC{\sk{\A_1}, \sk{\A_2}\vdash \sk{G}, \sk{\B_1}, \sk{\B_2}}
\RL{\land_r}
\BIC{\sk{\A_1}, \sk{\A_2}\vdash \sk{F}\land\sk{G}, \sk{\B_1}, \sk{\B_2}}
\DP}\\
\hspace{-2cm}\AIC{\A, F, G\vdash \B}
\RL{\otimes_l}
\UIC{\A, F\otimes G\vdash\B}
\DP\quad&\rightsquigarrow\quad
\AIC{\sk{\A}, \sk{F}, \sk{G}\vdash\sk{\B}}
\RL{\land_l^1, \land_l^2}
\doubleLine
\UIC{\sk{\A}, \sk{F}\land\sk{G}, \sk{F}\land\sk{G}\vdash \sk{\B}}
\RL{\lkcontr_l}
\UIC{\sk{\A}, \sk{F}\land\sk{G}\vdash \sk{\B}}
\DP\\
\hspace{-2cm}\AIC{}
\RL{1_r}
\UIC{\vdash 1}
\DP\quad&\rightsquigarrow\quad
\AIC{}
\RL{\true_r}
\UIC{\vdash \true}
\DP\\
\hspace{-2cm}\AIC{\A\vdash\B}
\RL{1_l}
\UIC{\A, 1\vdash\B}
\DP\quad&\rightsquigarrow\quad
\AIC{\sk{\A}\vdash \sk{\B}}
\RL{\wk_l}
\UIC{\sk{\A}, \true\vdash \sk{\B}}
\DP\\
\hspace{-2cm}
\AIC{\A_1, F\vdash \B_1}
\AIC{\A_2, G\vdash \B_2}
\RL{\parr_l}
\BIC{\A_1, \A_2, F\parr G\vdash \B_1, \B_2}
\DP\quad&\rightsquigarrow\quad
{\tiny\AIC{\sk{\A_1}, \sk{F}\vdash \sk{\B_1}}
\RL{\wk_l, \wk_r}
\doubleLine
\UIC{\sk{\A_1}, \sk{\A_2}, \sk{F}\vdash \sk{\B_1}, \sk{\B_2}}
\AIC{\sk{\A_2}, \sk{G}\vdash \sk{\B_2}}
\RL{\wk_l, \wk_r}
\doubleLine
\UIC{\sk{\A_1}, \sk{\A_2}, \sk{G}\vdash \sk{\B_1}, \sk{\B_2}}
\RL{\lor_l}
\BIC{\sk{\A_1}, \sk{\A_2}, \sk{F}\lor\sk{G}\vdash \sk{\B_1}, \sk{\B_2}}
\DP}\\
\hspace{-2cm}\AIC{\A \vdash F, G, \B}
\RL{\parr_r}
\UIC{\A\vdash F\parr G,\B}
\DP\quad&\rightsquigarrow\quad
\AIC{\sk{\A} \vdash \sk{F}, \sk{G}, \sk{\B}}
\RL{\lor_r^1, \lor_r^2}
\doubleLine
\UIC{\sk{\A}\vdash \sk{F}\lor\sk{G}, \sk{F}\lor\sk{G}, \sk{\B}}
\RL{\lkcontr_r}
\UIC{\sk{\A}\vdash \sk{F}\lor\sk{G}, \sk{\B}}
\DP\\
\hspace{-2cm}\AIC{}
\RL{\bot_l}
\UIC{\bot \vdash}
\DP\quad&\rightsquigarrow\quad
\AIC{}
\RL{\false_l}
\UIC{\false\vdash}
\DP\\
\hspace{-2cm}\AIC{\A\vdash\B}
\RL{\bot_r}
\UIC{\A\vdash\bot, \B}
\DP\quad&\rightsquigarrow\quad
\AIC{\sk{\A}\vdash \sk{\B}}
\RL{\wk_r}
\UIC{\sk{\A}\vdash \false, \sk{\B}}
\DP\\
\hspace{-2cm}\AIC{\A\vdash F, \B}
\AIC{\A\vdash G, \B}
\RL{\with_r}
\BIC{\A\vdash F\with G, \B}
\DP
\quad&\rightsquigarrow\quad
\AIC{\sk{\A}\vdash \sk{F}, \sk{\B}}
\AIC{\sk{\A}\vdash \sk{G}, \sk{\B}}
\RL{\land_r}
\BIC{\sk{\A}\vdash \sk{F}\land\sk{G}, \sk{\B}}
\DP\\
\hspace{-2cm}\AIC{\A, F\vdash \B}
\RL{\with_l^1}
\UIC{\A, F\with G\vdash \B}
\DP
\quad&\rightsquigarrow\quad
\AIC{\sk{\A}, \sk{F}\vdash \sk{\B}}
\RL{\land_l^1}
\UIC{\sk{\A}, \sk{F}\land \sk{G}\vdash \sk{\B}}
\DP\\
\hspace{-2cm}\AIC{\A, G\vdash \B}
\RL{\with_l^2}
\UIC{\A, F\with G\vdash \B}
\DP
\quad&\rightsquigarrow\quad
\AIC{\sk{\A}, \sk{G}\vdash \sk{\B}}
\RL{\land_l^2}
\UIC{\sk{\A}, \sk{F}\land \sk{G}\vdash \sk{\B}}
\DP\\
\hspace{-2cm}\AIC{}
\RL{\top_r}
\UIC{\A\vdash\top, \B}
\DP\quad&\rightsquigarrow\quad
\AIC{}
\RL{\true_r}
\UIC{\sk{\A}\vdash \true, \sk{\B}}
\DP\\
\hspace{-2cm}\AIC{\A, F\vdash \B}
\AIC{\A, G\vdash \B}
\RL{\oplus_l}
\BIC{\A, F\oplus G\vdash \B}
\DP
\quad&\rightsquigarrow\quad
\AIC{\sk{\A}, \sk{F}\vdash \sk{\B}}
\AIC{\sk{\A}, \sk{G}\vdash \sk{\B}}
\RL{\lor_l}
\BIC{\sk{\A}, \sk{F}\lor\sk{G}\vdash \sk{\B}}
\DP\\
\hspace{-2cm}\AIC{\A\vdash F, \B}
\RL{\oplus_r^1}
\UIC{\A \vdash F\oplus G, \B}
\DP\quad&\rightsquigarrow\quad
\AIC{\sk{\A}\vdash \sk{F}, \sk{\B}}
\RL{\lor_r^1}
\UIC{\sk{\A} \vdash \sk{F}\lor \sk{G}, \sk{\B}}
\DP
\\
\hspace{-2cm}\AIC{\A\vdash G, \B}
\RL{\oplus_r^2}
\UIC{\A \vdash F\oplus G, \B}
\DP\quad&\rightsquigarrow\quad
\AIC{\sk{\A}\vdash \sk{G}, \sk{\B}}
\RL{\lor_r^2}
\UIC{\sk{\A} \vdash \sk{F}\lor \sk{G}, \sk{\B}}
\DP
\\
\hspace{-2cm}\AIC{}
\RL{0_l}
\UIC{\A, 0\vdash\B}
\DP\quad&\rightsquigarrow\quad
\AIC{}
\RL{\false_l}
\UIC{\sk{\A}, \false \vdash \sk{\B}}
\DP\\
\hspace{-2cm}\AIC{\A, F\vdash G, \B}
\RL{\multimap_r}
\UIC{\A\vdash F\multimap G, \B}
\DP
\quad&\rightsquigarrow\quad
\AIC{\sk{\A}, \sk{F}\vdash \sk{G}, \sk{\B}}
\RL{\rightarrow_r}
\UIC{\sk{\A}\vdash \sk{F}\rightarrow \sk{G}, \sk{\B}}
\DP\\
\hspace{-2cm}\AIC{\A_1\vdash F, \B_1}
\AIC{\A_2, G\vdash \B_2}
\RL{\multimap_l}
\BIC{\A_1, \A_2, F\multimap G\vdash \B_1, \B_2}
\DP
\quad&\rightsquigarrow\quad
\AIC{\sk{\A_1}\vdash \sk{F}, \sk{\B_1}}
\AIC{\sk{\A_2},\sk{G}\vdash \sk{\B_2}}
\RL{\rightarrow_l}
\BIC{\sk{\A_1}, \sk{\A_2}, \sk{F}\rightarrow\sk{G}\vdash \sk{\B_1}, \sk{\B_2}}
\DP\\
\hspace{-2cm}\AIC{\A, A\vdash \B}
\RL{\perp_l}
\UIC{\A\vdash A^\perp, \B}
\DP
\quad&\rightsquigarrow\quad
\AIC{\sk\A\vdash \sk A, \sk\B}
\RL{\perp_l}
\UIC{\sk\A, \sk{A}^\perp\vdash \sk\B}
\DP
\\
\hspace{-2cm}\AIC{\A, A\vdash \B}
\RL{\perp_r}
\UIC{\A\vdash A^\perp, \B}
\DP
\quad&\rightsquigarrow\quad
\AIC{\sk\A, \sk A\vdash \sk\B}
\RL{\perp_r}
\UIC{\sk\A\vdash \sk{A}^\perp, \sk\B}
\DP
\\
\hspace{-2cm}\AIC{\A, F[\mu X. F/X]\vdash \B}
\RL{\mu_l}
\UIC{\A, \mu X. F\vdash\B}
\DP
\quad&\rightsquigarrow\quad
\AIC{\sk{\A}, \sk{F}[\mu X. \sk{F}/X]\vdash \sk{\B}}
\RL{\mu_l}
\UIC{\sk{\A}, \mu X. \sk{F}\vdash\sk{\B}}
\DP\\
\hspace{-2cm}\AIC{\A\vdash F[\mu X. F/X], \B}
\RL{\mu_r}
\UIC{\A\vdash\mu X. F, \B}
\DP
\quad&\rightsquigarrow\quad
\AIC{\sk{\A}\vdash\sk{F}[\mu X. \sk{F}/X], \sk{\B}}
\RL{\mu_r}
\UIC{\sk{\A}\vdash\mu X. \sk{F}, \sk{\B}}
\DP\\
\hspace{-2cm}\AIC{\A, F[\nu X. F/X]\vdash \B}
\RL{\nu_l}
\UIC{\A, \nu X. F\vdash\B}
\DP
\quad&\rightsquigarrow\quad
\AIC{\sk{\A}, \sk{F}[\nu X. \sk{F}/X]\vdash \sk{\B}}
\RL{\nu_l}
\UIC{\sk{\A}, \nu X. \sk{F}\vdash\sk{\B}}
\DP\\
\hspace{-2cm}\AIC{\A\vdash F[\nu X. F/X], \B}
\RL{\nu_r}
\UIC{\A\vdash\nu X. F, \B}
\DP
\quad&\rightsquigarrow\quad
\AIC{\sk{\A}\vdash\sk{F}[\nu X. \sk{F}/X], \sk{\B}}
\RL{\nu_r}
\UIC{\sk{\A}\vdash\nu X. \sk{F}, \sk{\B}}
\DP
\end{align*}
\caption{$\sk{-}$-translation of non-modal rules}\label{fig:skNonModruletranslation}
\end{figure*}

\begin{figure*}
\begin{align*}
\AIC{\A\vdash \B}
\RL{\ocwk}
\UIC{\A, \oc F\vdash \B}
\DP
\quad&\rightsquigarrow\quad
\AIC{\sk{\A}\vdash \sk{\B}}
\RL{\wk_l}
\UIC{\sk{\A}, \sk{F}\vdash \sk{\B}}
\DP\\
\AIC{\A, \oc F, \oc F\vdash \B}
\RL{\occontr{}}
\UIC{\A, \oc F\vdash \B}
\DP
\quad&\rightsquigarrow\quad
\AIC{\sk{\A}, \sk{F}, \sk{F}\vdash \sk{\B}}
\RL{\cntr_l}
\UIC{\sk{\A}, \sk{F}\vdash \sk{\B}}
\DP\\
\AIC{\A, F\vdash \B}
\RL{\ocde}
\UIC{\A, \oc F\vdash\B}
\DP
\quad&\rightsquigarrow\quad
\AIC{\sk{\A}, \sk{F}\vdash \sk{\B}}
\DP\\
\AIC{\oc \A\vdash F, \wn \B}
\RL{\ocprom}
\UIC{\oc\A\vdash \oc F, \wn\B}
\DP
\quad&\rightsquigarrow\quad
\AIC{\sk{\A} \vdash \sk{F}, \sk{\B}}
\DP
\\
\AIC{\A\vdash \B}
\RL{\wnwk}
\UIC{\A\vdash \wn F, \B}
\DP
\quad&\rightsquigarrow\quad
\AIC{\sk{\A}\vdash \sk{\B}}
\RL{\wk_r}
\UIC{\sk{\A}, \sk{F}\vdash \sk{\B}}
\DP\\
\AIC{\A\vdash \wn F, \wn F, \B}
\RL{\wncontr{}}
\UIC{\A\vdash \wn F, \B}
\DP
\quad&\rightsquigarrow\quad
\AIC{\sk{\A}\vdash \sk{F}, \sk{F}, \sk{\B}}
\RL{\cntr_r}
\UIC{\sk{\A}\vdash \sk{F}, \sk{\B}}
\DP\\
\AIC{\A\vdash F, \B}
\RL{\wnde}
\UIC{\A\vdash\wn F, \B}
\DP
\quad&\rightsquigarrow\quad
\AIC{\sk{\A}, \sk{F}\vdash \sk{\B}}
\DP\\
\AIC{\oc \A, F\vdash \wn \B}
\RL{\wnprom}
\UIC{\oc\A, \wn F\vdash \wn\B}
\DP
\quad&\rightsquigarrow\quad
\AIC{\sk{\A}, \sk{F}\vdash \sk{\B}}
\DP\\
\AIC{\A\vdash A, \B}
\RL{\boxprom}
\UIC{\Box \A\vdash \Box A, \lozenge\B}
\DP
\quad&\rightsquigarrow\quad
\AIC{\sk{\A}\vdash \sk{A}, \sk{\B}}
\RL{\boxprom}
\UIC{\Box \sk{\A}\vdash \Box \sk{A}, \lozenge\sk{\B}}
\DP\\
\AIC{\A, A\vdash \B}
\RL{\diaprom}
\UIC{\Box \A, \lozenge A\vdash \lozenge\B}
\DP
\quad&\rightsquigarrow\quad
\AIC{\sk{\A}, \sk{A}\vdash \sk{\B}}
\RL{\diaprom}
\UIC{\Box \sk{\A}, \lozenge\sk{A}\vdash \lozenge\sk{\B}}
\DP
\end{align*}
\caption{$\sk{-}$-translation of modal rules of \muLLinf{}}\label{fig:skModruletranslation}
\end{figure*}

\subsection{Details on proofs of \Cref{skeletonValidity} }
\label{sec:app:skeletonvalidity}

\begin{prop}[Robustness of the skeleton to validity]\label{app:skeletonValidity}
    If $\pi$ is a \muLLmodinf{} valid pre-proof then $\sk{\pi}$ is a \muLKmodinf{} valid pre-proof, and vice-versa.
    \end{prop}
    \begin{proof}
    This comes from the fact that (i) minimal formula of a set of translated formulas is the translation of the minimal formula of the set of initial formulas; (ii) translations of branches contains all the translations of formulas of the initial branch and vice-versa.
    \end{proof}

\section{Appendix on the \Cref{section:mullmodinfCutElim}}
\label{app:mullmodinfCutElim}

\subsection{Details on the multicut rule (\Cref{multicutdef})}\label{app:multicutdef}
The multi-cut rule is a rule with an arbitrary number of hypotheses:
$$
\AIC{\A_1\vdash \B_1}
\AIC{\dots}
\AIC{\A_n\vdash \B_n}
\RL{\rmcutpar}
\TIC{\A\vdash\B}
\DP$$
Let $C_1:=\{(1, i, j)\mid i\in\{ 1,\dots, n\}, j\in\{ 1,\dots, \#\A_i\}\}$,$C_2:=\{(2, i, j)\mid i\in\{ 1,\dots, n\}, j\in\{ 1,\dots, \#\B_i\}\}$, $\iota$ is a map from $\{1\}\times\{ 1,\dots, \# \A\} \cup \{2\}\times \{ 1,\dots, \# \B\}$ to $C = C_1\cup C_2$ and $\cutrel$ is a relation on $C$:
\begin{itemize}
\item Elements of $(k, n)$ are sent on $C_k$;
\item The map $\iota$ is injective;
\item If $(k,i,j)\cutrel (k',i',j')$ then $k\neq k'$;
\item The relation $\cutrel$ is defined for $C\setminus \iota$, and is total for this set;
\item The relation $\cutrel$ is symmetric;
\item Each index can be related at most once to another one;
\item If $(1,i,j) \cutrel (2,i',j')$, then the $\A_i[j] = \B_{i'}[j']$;
\item The projection of $\cutrel$ on the second element is acyclic and connected.
\end{itemize}

\subsection{Details on the restriction of a multicut context (\Cref{multicutRestriction})}\label{app:multicutRestriction}
\begin{defi}[Restriction of a multicut context]
Let $
\AIC{\C}
\RL{\rmcutpar}
\UIC{s}
\DP
$ be a multicut-occurrence such that $\C = s_1\quad\dots\quad s_n$ and let $s_i := F_1, \dots, F_{k_i}\vdash G_1, \dots, G_{r_i} $, we define $\C_{F_j}$ (resp. $\C_{G_j}$) with $F_j\in s_i$ (resp. $G_j\in s_i$) to be the least sub-context of $\C$ such that:
\begin{itemize}
\item The sequent $s_i$ is in $\C_{F_j}$ (resp. $\C_{G_j}$);
\item If there exists $l$ such that $(1,i,j)\cutrel(2,k,l)$ or $(2,i,j)\cutrel(1,k,l)$ then $s_k\in\C_{F_j}$ (resp. $s_k\in\C_{G_j}$);
\item For any $k\neq i$, if there exists $l$ such that $(1,k,l)\cutrel(2,k',l')$ or $(2,k,l)\cutrel(1,k',l')$ and that $s_k\in\C_{F_j}$ (resp. $s_k\in\C_{G_j}$) then $s_{k'}\in\C_{F_j}$ ($s_{k'}\in\C_{G_j}$).
\end{itemize}
We then extend the notation to contexts, setting $\C_\emptyset := \emptyset$ and $\C_{F, \A} := \C_F\cup\C_\A$.
\end{defi}

\subsection{Full (\mcut)-reduction steps}
\label{app:mcutSteps}

\begin{defi}[\muMALLinf{} \& \muLLinf{} (\mcut)-reduction steps]
Reduction steps of \muMALLinf{} are given in \Cref{fig:mcutMALLPrincip,fig:mcutMLLComm,fig:mcutALLComm}.

Reduction steps of \muLLmodinf{} are the reduction steps of \muMALLinf{} together with the steps of \Cref{fig:app:mullmodinfexpcommcutstep1,fig:app:mullmodinfexpcommcutstep2,fig:appmullmodinfexpprincipcutstep1,fig:appmullmodinfexpprincipcutstep2}.

Reduction steps of \muLLinf{} are the $\{\Box, \lozenge\}$-free reduction steps of \muLLmodinf{}.
\end{defi}

\begin{figure*}
    \scalebox{0.8}{
    \begin{minipage}{\linewidth}
    \begin{align*}
    \AIC{\A_1\vdash A_1, \B_1}
    \AIC{\A_2\vdash A_2, \B_2}
    \RL{\otimes_r}
    \BIC{\A_1, \A_2\vdash A_1\otimes A_2, \B_1, \B_2}
    \AIC{\A_3, A_1, A_2\vdash \B_3}
    \RL{\otimes_l}
    \UIC{\A_3, A_1\otimes A_2\vdash \B_3}
    \AIC{\C}
    \RL{\rmcutpar}
    \TIC{\A\vdash \B}
    \DP \quad&\rightsquigarrow\\
    &\hspace{-6cm}
    \AIC{\A_1\vdash A_1, \B_1}
    \AIC{\A_2\vdash A_2, \B_2}
    \AIC{\A_3, A_1, A_2\vdash \B_3}
    \AIC{\C}
    \RL{\rmcutparprime}
    \QIC{\A\vdash \B}
    \DP\\[2ex]
    \AIC{\A_1, A_1\vdash \B_1}
    \AIC{\A_2, A_2\vdash \B_2}
    \RL{\parr_l}
    \BIC{\A_1, \A_2, A_1\parr A_2\vdash \B_1, \B_2}
    \AIC{\A_3\vdash A_1, A_2,\B_3}
    \RL{\parr_r}
    \UIC{\A_3\vdash A_1\parr A_2, \B_3}
    \AIC{\C}
    \RL{\rmcutpar}
    \TIC{\A\vdash \B}
    \DP \quad&\rightsquigarrow\\
    &\hspace{-6cm}
    \AIC{\A_1, A_1\vdash \B_1}
    \AIC{\A_2, A_2\vdash \B_2}
    \AIC{\A_3\vdash A_1, A_2, \B_3}
    \AIC{\C}
    \RL{\rmcutparprime}
    \QIC{\A\vdash \B}
    \DP\\[2ex]
    \AIC{\A_1\vdash A_1, \B_1}
    \AIC{\A_2, A_2\vdash \B_2}
    \RL{\multimap_l}
    \BIC{\A_1, \A_2, A_1\multimap A_2\vdash \B_1, \B_2}
    \AIC{\A_3, A_1\vdash A_2,\B_3}
    \RL{\multimap_r}
    \UIC{\A_3\vdash A_1\multimap A_2, \B_3}
    \AIC{\C}
    \RL{\rmcutpar}
    \TIC{\A\vdash \B}
    \DP \quad&\rightsquigarrow\\
    &\hspace{-6cm}
    \AIC{\A_1\vdash A_1, \B_1}
    \AIC{\A_2, A_2\vdash \B_2}
    \AIC{\A_3, A_1\vdash A_2, \B_3}
    \AIC{\C}
    \RL{\rmcutparprime}
    \QIC{\A\vdash \B}
    \DP\\[2ex]
    \AIC{\A_1\vdash\B_1}
    \RL{1_l}
    \UIC{\A_1, 1\vdash \B_1}
    \AIC{}
    \RL{1_r}
    \UIC{\vdash 1}
    \AIC{\C}
    \RL{\rmcutpar}
    \TIC{\A\vdash\B}
    \DP
    \quad&\rightsquigarrow\quad
    \AIC{\A_1\vdash\B_1}
    \AIC{\C}
    \RL{\rmcutparprime}
    \BIC{\A\vdash\B}
    \DP\\[2ex]
    \AIC{\A_1\vdash\B_1}
    \RL{\bot_r}
    \UIC{\A_1\vdash\bot, \B_1}
    \AIC{}
    \RL{\bot_l}
    \UIC{\bot\vdash}
    \AIC{\C}
    \RL{\rmcutpar}
    \TIC{\A\vdash\B}
    \DP
    \quad&\rightsquigarrow\quad
    \AIC{\A_1\vdash\B_1}
    \AIC{\C}
    \RL{\rmcutparprime}
    \BIC{\A\vdash\B}
    \DP\\[2ex]
    \AIC{\A_1\vdash F_i,\B_1}
    \RL{\oplus_r^i}
    \UIC{\A_1\vdash F_1\oplus F_2, \B_1}
    \AIC{\A_2, F_1\vdash \B_2}
    \AIC{\A_2, F_2\vdash \B_2}
    \RL{\oplus_l}
    \BIC{\A_2, F_1\oplus F_2\vdash \B_2}
    \AIC{\C}
    \RL{\rmcutpar}
    \TIC{\A\vdash\B}
    \DP
    \quad &\rightsquigarrow\\[2ex]
    & \hspace{-5cm}\AIC{\A_1\vdash F_i,\B_1}
    \AIC{\A_2, F_i\vdash \B_2}
    \AIC{\C}
    \RL{\rmcutpar}
    \TIC{\A\vdash\B}
    \DP\\[2ex]
    \AIC{\A_1, F_i\vdash\B_1}
    \RL{\with_l^i}
    \UIC{\A_1, F_1\with F_2\vdash\B_1}
    \AIC{\A_2\vdash F_1, \B_2}
    \AIC{\A_2\vdash F_2, \B_2}
    \RL{\with_r}
    \BIC{\A_2\vdash F_1\with F_2, \B_2}
    \AIC{\C}
    \RL{\rmcutpar}
    \TIC{\A\vdash\B}
    \DP
    \quad &\rightsquigarrow\\[2ex]
    & \hspace{-5cm}
    \AIC{\A_1, F_i\vdash \B_1}
    \AIC{\A_2\vdash F_i, \B_2}
    \AIC{\C}
    \RL{\rmcutpar}
    \TIC{\A\vdash\B}
    \DP\\[2ex]
    \AIC{\A_1, A\vdash \B_1}
    \RL{(-)^\perp_r}
    \UIC{\A_1\vdash A^\perp, \B_1}
    \AIC{\A_2\vdash A, \B_2}
    \RL{(-)^\perp_l}
    \UIC{\A_2, A^\perp\vdash\B_2}
    \AIC{\C}
    \RL{\rmcutpar}
    \TIC{\A\vdash \B}
    \DP
    \quad &\rightsquigarrow\\[2ex]
    & \hspace{-5cm}
    \AIC{\A_1, A\vdash\B_1}
    \AIC{\A_2\vdash A, \B_2}
    \AIC{\C}
    \RL\rmcutparprime
    \TIC{\A\vdash\B}
    \DP\\[2ex]
    \AIC{}
    \RL{\ax}
    \UIC{A\vdash A}
    \AIC{\C}
    \RL\rmcutpar
    \BIC{\A\vdash \B}
    \DP\quad&\rightsquigarrow\quad
    \AIC{\C}
    \RL\rmcutparprime
    \UIC{\A\vdash\B}
    \DP\\[2ex]
    \AIC{\A_1\vdash F, \B_1}
    \AIC{\A_2, F\vdash \B_2}
    \RL{\cut}
    \BIC{\A_1, \A_2\vdash \B_1, \B_2}
    \AIC{\C}
    \RL{\rmcutpar}
    \BIC{\A\vdash\B}
    \DP
    \quad &\rightsquigarrow\\[2ex]
    & \hspace{-5cm}
    \AIC{\A_1\vdash F, \B_1}
    \AIC{\A_2, F\vdash \B_2}
    \AIC{\C}
    \RL{\rmcutparprime}
    \TIC{\A\vdash\B}
    \DP\\[2ex]
    \AIC{\A_1, F[X:=\delta X. F]\vdash \B_1}
    \RL{\delta_l}
    \UIC{\A_1, \delta X. F\vdash \B_1}
    \AIC{\A_2\vdash F[X:=\delta X. F], \B_2}
    \RL{\delta_r}
    \UIC{\A_2\vdash \delta X. F, \B_2}
    \AIC{\C}
    \RL{\rmcutpar}
    \TIC{\A\vdash\B}
    \DP\quad&\rightsquigarrow\\
    &\hspace{-6.5cm}
    \AIC{\A_1, F[X:=\delta X. F]\vdash \B_1}
    \AIC{\A_2\vdash F[X:=\delta X. F], \B_2}
    \AIC{\C}
    \RL{\rmcutpar}
    \TIC{\A\vdash\B}
    \DP
\end{align*}
\end{minipage}
    } 
    \caption{Principal (\mcut)-step of $\muMALLinf{}$ --- $\delta\in\{\mu, \nu\}$}\label{fig:mcutMALLPrincip}
\end{figure*}

\begin{figure*}
    \centering
    \scalebox{0.8}{
    \begin{minipage}{\linewidth}
    \begin{align*}
    \AIC{}
    \RL{\ax}
    \UIC{A\vdash A}
    \RL{\rmcutpar}
    \UIC{A\vdash A}
    \DP\quad&\rightsquigarrow\quad
    \AIC{}
    \RL{\ax}
    \UIC{A\vdash A}
    \DP\\
    \AIC{}
    \RL{1_r}
    \UIC{\vdash 1}
    \RL{\rmcutpar}
    \UIC{\vdash 1}
    \DP\quad&\rightsquigarrow\quad
    \AIC{}
    \RL{1_r}
    \UIC{\vdash 1}
    \DP\\
    \AIC{}
    \RL{\bot_l}
    \UIC{\bot\vdash}
    \RL{\rmcutpar}
    \UIC{\bot\vdash}
    \DP\quad&\rightsquigarrow\quad
    \AIC{}
    \RL{\bot_l}
    \UIC{\bot\vdash}
    \DP\\
    \AIC{\A\vdash \B}
    \RL{\bot_r}
    \UIC{\A\vdash \bot, \B}
    \AIC{\C}
    \RL{\rmcutpar}
    \BIC{\A'\vdash\bot, \B'}
    \DP
    \quad&\rightsquigarrow\quad
    \AIC{\A\vdash \B}
    \AIC{\C}
    \RL{\rmcutparprime}
    \BIC{\A'\vdash \B'}
    \RL{\bot_r}
    \UIC{\A'\vdash\bot, \B'}
    \DP\\
    \AIC{\A\vdash \B}
    \RL{1_l}
    \UIC{\A, 1\vdash\B}
    \AIC{\C}
    \RL{\rmcutpar}
    \BIC{\A', 1\vdash\B'}
    \DP
    \quad&\rightsquigarrow\quad
    \AIC{\A\vdash \B}
    \AIC{\C}
    \RL{\rmcutparprime}
    \BIC{\A'\vdash \B'}
    \RL{1_l}
    \UIC{\A', 1\vdash\B'}
    \DP\\
    \AIC{\A\vdash A_1, A_2,\B}
    \RL{\parr_r}
    \UIC{\A\vdash A_1\parr A_2, \B}
    \AIC{\C}
    \RL\rmcutpar
    \BIC{\A'\vdash A_1\parr A_2, \B'}
    \DP\quad&\rightsquigarrow \quad
    \AIC{\A\vdash A_1, A_2, \B}
    \AIC{\C}
    \RL{\rmcutparprime}
    \BIC{\A'\vdash A_1, A_2, \B'}
    \RL{\parr_r}
    \UIC{\A'\vdash A_1\parr A_2, \B'}
    \DP\\
    \AIC{\A, A_1, A_2\vdash \B}
    \RL{\otimes_l}
    \UIC{\A, A_1\otimes A_2\vdash \B}
    \AIC{\C}
    \RL\rmcutpar
    \BIC{\A', A_1\otimes A_2\vdash \B'}
    \DP\quad&\rightsquigarrow \quad
    \AIC{\A, A_1, A_2\vdash\B}
    \AIC{\C}
    \RL{\rmcutparprime}
    \BIC{\A', A_1, A_2\vdash \B'}
    \RL{\otimes_l}
    \UIC{\A',A_1\otimes A_2\vdash \B'}
    \DP\\
    \AIC{\A, A_1\vdash A_2, \B}
    \RL{\multimap_r}
    \UIC{\A\vdash A_1\multimap A_2, \B}
    \AIC{\C}
    \RL\rmcutpar
    \BIC{\A'\vdash A_1\multimap A_2, \B'}
    \DP\quad&\rightsquigarrow \quad
    \AIC{\A, A_1\vdash A_2, \B}
    \AIC{\C}
    \RL{\rmcutparprime}
    \BIC{\A', A_1\vdash A_2, \B'}
    \RL{\multimap_r}
    \UIC{\A'\vdash A_1\multimap A_2, \B'}
    \DP\\
    \AIC{\A', A\vdash\B'}
    \RL{(-)^\perp_r}
    \UIC{\A'\vdash A^\perp, \B'}
    \AIC{\C}
    \RL{\rmcutpar}
    \BIC{\A\vdash A^\perp, \B}
    \DP\quad&\rightsquigarrow\quad
    \AIC{\A', A\vdash\B'}
    \AIC{\C}
    \RL{\rmcutparprime}
    \BIC{\A, A\vdash\B}
    \RL{(-)^\perp_r}
    \UIC{\A\vdash A^\perp, \B}
    \DP\\
    \AIC{\A'\vdash A, \B'}
    \RL{(-)^\perp_l}
    \UIC{\A', A^\perp\vdash\B'}
    \AIC{\C}
    \RL{\rmcutpar}
    \BIC{\A, A^\perp\vdash\B}
    \DP\quad&\rightsquigarrow\quad
    \AIC{\A'\vdash A, \B'}
    \AIC{\C}
    \RL{\rmcutparprime}
    \BIC{\A\vdash A, \B}
    \RL{(-)^\perp_l}
    \UIC{\A, A^\perp\vdash\B}
    \DP\\
    \AIC{\A'_1\vdash A_1, \B'_1}
    \AIC{\A'_2\vdash A_2, \B'_2}
    \RL{\otimes_r}
    \BIC{\tikzmark{otimesrcommprect3}\A'_1, \tikzmark{otimesrcommprect4}\A'_2\vdash A_1\otimes A_2, \B'_1\tikzmark{otimesrcommprect5}, \B'_2\tikzmark{otimesrcommprect6}}
    \AIC{\tikzmark{otimesrcommprect1}\C_{\A'_1, \B'_1}}
    \AIC{\tikzmark{otimesrcommprect2}\C_{\A'_2, \B'_2}}
    \RL{\rmcutpar}
    \TIC{\tikzmark{otimesrcommprecd11}\A_1, \A_2\tikzmark{otimesrcommprecd21}\vdash A_1\otimes A_2, \tikzmark{otimesrcommprecd12}\B_1, \B_2\tikzmark{otimesrcommprecd22}}
    \DP \quad&\rightsquigarrow\\
    &\hspace{-7cm}
    \AIC{\A'_1\vdash A_1, \B'_1}
    \AIC{\C_{\A'_1, \B'_1}}
    \RL{\rmcutparprime}
    \BIC{\A_1\vdash A_1, \B_1}
    \AIC{\A'_2\vdash A_2, \B'_2}
    \AIC{\C_{\A'_2, \B'_2}}
    \RL{\rmcutpardouble}
    \BIC{\A_2\vdash A_2, \B_2}
    \RL{\otimes_r}
    \BIC{\A_1, \A_2\vdash A_1\otimes A_2, \B_1, \B_2}
    \DP\\
    \AIC{\A'_1, A_1\vdash \B'_1}
    \AIC{\A'_2, A_2\vdash \B'_2}
    \RL{\parr_l}
    \BIC{\tikzmark{parrlcommprect3}\A'_1, \tikzmark{parrlcommprect4}\A'_2, A_1\parr A_2\vdash \B'_1\tikzmark{parrlcommprect5}, \B'_2\tikzmark{parrlcommprect6}}
    \AIC{\tikzmark{parrlcommprect1}\C_{\A'_1, \B'_1}}
    \AIC{\tikzmark{parrlcommprect2}\C_{\A'_2, \B'_2}}
    \RL{\rmcutpar}
    \TIC{\tikzmark{parrlcommprecd11}\A_1, \A_2\tikzmark{parrlcommprecd21}, A_1\parr A_2\vdash \tikzmark{parrlcommprecd12}\B_1, \B_2\tikzmark{parrlcommprecd22}}
    \DP \quad&\rightsquigarrow\\
    &\hspace{-7cm}
    \AIC{\A'_1, A_1\vdash\B'_1}
    \AIC{\C_{\A'_1, \B'_1}}
    \RL{\rmcutparprime}
    \BIC{\A_1, A_1\vdash\B_1}
    \AIC{\A'_2, A_2\vdash\B'_2}
    \AIC{\C_{\A'_2, \B'_2}}
    \RL{\rmcutpardouble}
    \BIC{\A_2, A_2\vdash\B_2}
    \RL{\parr_l}
    \BIC{\A_1, \A_2, A_1\parr A_2\vdash\B_1, \B_2}
    \DP\\
    \AIC{\A'_1\vdash A_1, \B'_1}
    \AIC{\A'_2, A_2\vdash \B'_2}
    \RL{\multimap_l}
    \BIC{\tikzmark{multimaplcommprect3}\A'_1, \tikzmark{multimaplcommprect4}\A'_2, A_1\multimap A_2\vdash \B'_1\tikzmark{multimaplcommprect5}, \B'_2\tikzmark{multimaplcommprect6}}
    \AIC{\tikzmark{multimaplcommprect1}\C_{\A'_1, \B'_1}}
    \AIC{\tikzmark{multimaplcommprect2}\C_{\A'_2, \B'_2}}
    \RL{\rmcutpar}
    \TIC{\tikzmark{multimaplcommprecd11}\A_1, \A_2\tikzmark{multimaplcommprecd21}, A_1\multimap A_2\vdash \tikzmark{multimaplcommprecd12}\B_1, \B_2\tikzmark{multimaplcommprecd22}}
    \DP \quad&\rightsquigarrow\\
    &\hspace{-7cm}
    \AIC{\A'_1\vdash A_1, \B'_1}
    \AIC{\C_{\A'_1, \B'_1}}
    \RL{\rmcutparprime}
    \BIC{\A_1 \vdash A_1, \B_1}
    \AIC{\A'_2, A_2\vdash\B'_2}
    \AIC{\C_{\A'_2, \B'_2}}
    \RL{\rmcutpardouble}
    \BIC{\A_2, A_2\vdash\B_2}
    \RL{\multimap_l}
    \BIC{\A_1, \A_2, A_1\multimap A_2\vdash\B_1, \B_2}
    \DP
    \end{align*}
    \begin{tikzpicture}[overlay,remember picture,-,line cap=round,line width=0.1cm]
        \draw[rounded corners, smooth=2,red, opacity=.3] ([xshift=1mm,yshift=1mm] pic cs:otimesrcommprecd12) to ([xshift=2mm,yshift=1mm] pic cs:otimesrcommprect1);
        \draw[rounded corners, smooth=2,red, opacity=.3] ([xshift=1mm,yshift=1mm] pic cs:otimesrcommprecd11) to ([xshift=2mm,yshift=1mm] pic cs:otimesrcommprect1);
        \draw[rounded corners, smooth=2,red, opacity=.3] ([xshift=1mm,yshift=1mm] pic cs:otimesrcommprecd11) to ([yshift=2mm] pic cs:otimesrcommprect3);
        \draw[rounded corners, smooth=2,red, opacity=.3] ([xshift=1mm,yshift=1mm] pic cs:otimesrcommprecd12) to ([xshift=-2mm,yshift=2mm] pic cs:otimesrcommprect5);
        \draw[rounded corners, smooth=2,blue, opacity=.3] ([xshift=-1mm] pic cs:otimesrcommprecd21) to ([xshift=2mm,yshift=1mm] pic cs:otimesrcommprect2);
        \draw[rounded corners, smooth=2,blue, opacity=.3] ([xshift=-1mm] pic cs:otimesrcommprecd22) to ([xshift=2mm,yshift=1mm] pic cs:otimesrcommprect2);
        \draw[rounded corners, smooth=2,blue, opacity=.3] ([xshift=-1mm] pic cs:otimesrcommprecd21) to ([yshift=2mm] pic cs:otimesrcommprect4);
        \draw[rounded corners, smooth=2,blue, opacity=.3] ([xshift=-1mm] pic cs:otimesrcommprecd22) to ([xshift=-2mm,yshift=2mm] pic cs:otimesrcommprect6);
        \draw[rounded corners, smooth=2,red, opacity=.3] ([xshift=1mm,yshift=1mm] pic cs:parrlcommprecd12) to ([xshift=2mm,yshift=1mm] pic cs:parrlcommprect1);
        \draw[rounded corners, smooth=2,red, opacity=.3] ([xshift=1mm,yshift=1mm] pic cs:parrlcommprecd11) to ([xshift=2mm,yshift=1mm] pic cs:parrlcommprect1);
        \draw[rounded corners, smooth=2,red, opacity=.3] ([xshift=1mm,yshift=1mm] pic cs:parrlcommprecd11) to ([yshift=2mm] pic cs:parrlcommprect3);
        \draw[rounded corners, smooth=2,red, opacity=.3] ([xshift=1mm,yshift=1mm] pic cs:parrlcommprecd12) to ([xshift=-2mm,yshift=2mm] pic cs:parrlcommprect5);
        \draw[rounded corners, smooth=2,blue, opacity=.3] ([xshift=-1mm] pic cs:parrlcommprecd21) to ([xshift=2mm,yshift=1mm] pic cs:parrlcommprect2);
        \draw[rounded corners, smooth=2,blue, opacity=.3] ([xshift=-1mm] pic cs:parrlcommprecd22) to ([xshift=2mm,yshift=1mm] pic cs:parrlcommprect2);
        \draw[rounded corners, smooth=2,blue, opacity=.3] ([xshift=-1mm] pic cs:parrlcommprecd21) to ([yshift=2mm] pic cs:parrlcommprect4);
        \draw[rounded corners, smooth=2,blue, opacity=.3] ([xshift=-1mm] pic cs:parrlcommprecd22) to ([xshift=-2mm,yshift=2mm] pic cs:parrlcommprect6);
        \draw[rounded corners, smooth=2,red, opacity=.3] ([xshift=1mm,yshift=1mm] pic cs:multimaplcommprecd12) to ([xshift=2mm,yshift=1mm] pic cs:multimaplcommprect1);
        \draw[rounded corners, smooth=2,red, opacity=.3] ([xshift=1mm,yshift=1mm] pic cs:multimaplcommprecd11) to ([xshift=2mm,yshift=1mm] pic cs:multimaplcommprect1);
        \draw[rounded corners, smooth=2,red, opacity=.3] ([xshift=1mm,yshift=1mm] pic cs:multimaplcommprecd11) to ([yshift=2mm] pic cs:multimaplcommprect3);
        \draw[rounded corners, smooth=2,red, opacity=.3] ([xshift=1mm,yshift=1mm] pic cs:multimaplcommprecd12) to ([xshift=-2mm,yshift=2mm] pic cs:multimaplcommprect5);
        \draw[rounded corners, smooth=2,blue, opacity=.3] ([xshift=-1mm] pic cs:multimaplcommprecd21) to ([xshift=2mm,yshift=1mm] pic cs:multimaplcommprect2);
        \draw[rounded corners, smooth=2,blue, opacity=.3] ([xshift=-1mm] pic cs:multimaplcommprecd22) to ([xshift=2mm,yshift=1mm] pic cs:multimaplcommprect2);
        \draw[rounded corners, smooth=2,blue, opacity=.3] ([xshift=-1mm] pic cs:multimaplcommprecd21) to ([yshift=2mm] pic cs:multimaplcommprect4);
        \draw[rounded corners, smooth=2,blue, opacity=.3] ([xshift=-1mm] pic cs:multimaplcommprecd22) to ([xshift=-2mm,yshift=2mm] pic cs:multimaplcommprect6);
    \end{tikzpicture}
\end{minipage}
}
    \caption{Commutative (\mcut)-step of the multiplicative fragment of $\MALL$}\label{fig:mcutMLLComm}
\end{figure*}

\begin{figure*}
\begin{align*}
    \AIC{}
    \RL{\top_r}
    \UIC{\A'\vdash\top, \B'}
    \AIC{\C}
    \RL\rmcutpar
    \BIC{\A\vdash \top, \B}
    \DP\quad&\rightsquigarrow\quad
    \AIC{}
    \RL{\top_r}
    \UIC{\A\vdash \top, \B}
    \DP\\[2ex]
    \AIC{}
    \RL{0_l}
    \UIC{\A', 0\vdash\B'}
    \AIC{\C}
    \RL\rmcutpar
    \BIC{\A, 0\vdash\B}
    \DP\quad&\rightsquigarrow\quad
    \AIC{}
    \RL{0_l}
    \UIC{\A, 0\vdash \B}
    \DP\\[2ex]
    \AIC{\A\vdash F',\B}
    \RL{r}
    \UIC{\A\vdash F, \B}
    \AIC{\C}
    \RL\rmcutpar
    \BIC{\A'\vdash F, \B'}
    \DP\quad&\rightsquigarrow \quad
    \AIC{\A\vdash F', \B}
    \AIC{\C}
    \RL{\rmcutpar}
    \BIC{\A'\vdash F', \B'}
    \RL{r}
    \UIC{\A'\vdash F, \B'}
    \DP\quad \\[2ex]
    \AIC{\A, F'\vdash \B}
    \RL{r}
    \UIC{\A, F\vdash \B}
    \AIC{\C}
    \RL\rmcutpar
    \BIC{\A',F\vdash \B'}
    \DP\quad&\rightsquigarrow \quad
    \AIC{\A, F'\vdash \B}
    \AIC{\C}
    \RL{\rmcutpar}
    \BIC{\A', F'\vdash \B'}
    \RL{r}
    \UIC{\A', F\vdash \B'}
    \DP\\[2ex]
    \AIC{\A'\vdash A_1, \B'}
    \AIC{\A'\vdash A_2, \B'}
    \RL{\with_r}
    \BIC{\A'\vdash A_1\with A_2, \B'}
    \AIC{\C}
    \RL\rmcutpar
    \BIC{\A\vdash A_1\with A_2, \B}
    \DP
    \quad&\rightsquigarrow\\
    &\hspace{-5cm}
    \AIC{\A'\vdash A_1, \B'}
    \AIC{\C}
    \RL{\rmcutpar}
    \BIC{\A\vdash A_1, \B}
    \AIC{\A'\vdash A_2, \B'}
    \AIC{\C}
    \RL{\rmcutpar}
    \BIC{\A\vdash A_2, \B}
    \RL{\with_r}
    \BIC{\A\vdash A_1\with A_2, \B}
    \DP\\[2ex]
    \AIC{\A',A_1\vdash\B'}
    \AIC{\A',A_2\vdash\B'}
    \RL{\oplus_l}
    \BIC{\A', A_1\oplus A_2\vdash\B'}
    \AIC{\C}
    \RL\rmcutpar
    \BIC{\A, A_1\oplus A_2\vdash\B}
    \DP
    \quad&\rightsquigarrow\\
    &\hspace{-5cm}
    \AIC{\A', A_1\vdash\B'}
    \AIC{\C}
    \RL{\rmcutpar}
    \BIC{\A, A_1\vdash\B}
    \AIC{\A', A_2\vdash \B'}
    \AIC{\C}
    \RL{\rmcutpar}
    \BIC{\A, A_2\vdash \B}
    \RL{\oplus_l}
    \BIC{\A, A_1\oplus A_2\vdash\B}
    \DP
\end{align*}
\caption{Commutative (\mcut)-step of the additive fragment of $\muMALLinf{}$}\label{fig:mcutALLComm}
\end{figure*}

\begin{figure*}[htpb] 
    \begin{align*}
    \AIC{\pi}
    \noLine
    \UIC{\oc\A_1, \Box\A_2\vdash A, \wn\B_1, \lozenge\B_2}
    \RL{\ocpromloz}
    \UIC{\oc\A_1, \Box\A_2\vdash \oc A, \wn\B_1, \lozenge\B_2}
    \AIC{\ocboxProofs}
    \RL{\rmcutpar}
    \BIC{\oc\A', \Box\A''\vdash \oc A, \wn\B, \lozenge\B'}
    \DP\qquad&\rightsquigarrow\qquad
    \AIC{\pi}
    \noLine
    \UIC{\oc\A_1, \Box\A_2\vdash A, \wn\B_1, \lozenge\B_2}
    \AIC{\ocboxProofs}
    \RL{\rmcutpar}
    \BIC{\oc\A', \Box\A''\vdash A, \wn\B, \lozenge\B'}
    \RL{\ocpromloz}
    \UIC{\oc\A', \Box\A''\vdash \oc A, \wn\B, \lozenge\B'}
    \DP
    \\
    \AIC{\pi}
    \noLine
    \UIC{\A\vdash A, \B}
    \RL{\boxprom}
    \UIC{\Box\A\vdash \Box A, \lozenge\B}
    \AIC{\boxProofs}
    \RL{\rmcutpar}
    \BIC{\Box\A'\vdash \Box A, \lozenge\B'}
    \DP\qquad&\rightsquigarrow\qquad\nolinebreak
    \AIC{\pi}
    \noLine
    \UIC{\A\vdash A, \B}
    \AIC{\C}
    \RL{\rmcutpar}
    \BIC{\A'\vdash A, \B'}
    \RL{\boxprom}
    \UIC{\Box\A'\vdash \Box A, \lozenge\B'}
    \DP
    \\
    \AIC{\pi}
    \noLine
    \UIC{\A\vdash \B}
    \RL{\delta_\wk}
    \UIC{\A\vdash \delta A, \B}
    \AIC{\mathcal{C}}
    \RL{\rmcutpar}
    \BIC{\A'\vdash \delta A, \B'}
    \DP\qquad&\rightsquigarrow\qquad
    \AIC{\pi}
    \noLine
    \UIC{\A\vdash \B}
    \AIC{\mathcal{C}}
    \RL{\mathsf{mcut(\iota', \perp\!\!\!\perp')}}
    \BIC{\A'\vdash \B'}
    \RL{\delta_\wk}
    \UIC{\A'\vdash \delta A, \B'}
    \DP
    \\
    \AIC{\pi}
    \noLine
    \UIC{\A\vdash \delta A, \delta A, \B}
    \RL{\delta_\cntr{}}
    \UIC{\A\vdash \delta A, \B}
    \AIC{\mathcal{C}}
    \RL{\rmcutpar}
    \BIC{\A'\vdash \delta A, \B'}
    \DP\qquad&\rightsquigarrow\qquad
    \AIC{\pi}
    \noLine
    \UIC{\A\vdash \delta A, \delta A, \B}
    \AIC{\mathcal{C}}
    \RL{\mathsf{mcut(\iota', \perp\!\!\!\perp')}}
    \BIC{\A'\vdash \delta A, \delta A, \B'}
    \RL{\delta_\cntr{}}
    \UIC{\A'\vdash \delta A, \B'}
    \DP 
    \\
    \AIC{\pi}
    \noLine
    \UIC{\A\vdash  A, \B}
    \RL{\wnde}
    \UIC{\A\vdash \wn A, \B}
    \AIC{\mathcal{C}}
    \RL{\rmcutpar}
    \BIC{\A'\vdash \wn A, \B'}
    \DP\qquad&\rightsquigarrow\qquad
    \AIC{\pi}
    \noLine
    \UIC{\A\vdash A, \B}
    \AIC{\mathcal{C}}
    \RL{\mathsf{mcut(\iota', \perp\!\!\!\perp')}}
    \BIC{\A'\vdash A, \B'}
    \RL{\wnde}
    \UIC{\A'\vdash \wn A, \B'}
    \DP
    \end{align*}
    \centering
    $\delta\in\{\wn, \lozenge\}$
    \caption{First side of commutative cut-elimination steps of \muLLmodinf{}}\label{fig:app:mullmodinfexpcommcutstep1}
    \end{figure*}

    \begin{figure*}[htpb] 
        \begin{align*}
        \AIC{\pi}
        \noLine
        \UIC{\oc\A_1, \Box\A_2, A\vdash \wn\B_1, \lozenge\B_2}
        \RL{\wnpromloz}
        \UIC{\oc\A_1, \Box\A_2, \wn A\vdash \wn\B_1, \lozenge\B_2}
        \AIC{\ocboxProofs}
        \RL{\rmcutpar}
        \BIC{\oc\A', \Box\A'', \wn A\vdash\wn\B, \lozenge\B'}
        \DP\qquad&\rightsquigarrow\qquad
        \AIC{\pi}
        \noLine
        \UIC{\oc\A_1, \Box\A_2, A\vdash\wn\B_1, \lozenge\B_2}
        \AIC{\ocboxProofs}
        \RL{\rmcutpar}
        \BIC{\oc\A', \Box\A'', A\vdash\wn\B, \lozenge\B'}
        \RL{\wnpromloz}
        \UIC{\oc\A', \Box\A'', \wn A\vdash \wn\B, \lozenge\B'}
        \DP
        \\
        \AIC{\pi}
        \noLine
        \UIC{\A, a\vdash\B}
        \RL{\diaprom}
        \UIC{\Box\A, \lozenge A\vdash\lozenge\B}
        \AIC{\boxProofs}
        \RL{\rmcutpar}
        \BIC{\Box\A', \lozenge A\vdash\lozenge\B'}
        \DP\qquad&\rightsquigarrow\qquad\nolinebreak
        \AIC{\pi}
        \noLine
        \UIC{\A, A\vdash\B}
        \AIC{\C}
        \RL{\rmcutpar}
        \BIC{\A', A\vdash\B'}
        \RL{\diaprom}
        \UIC{\Box\A', \lozenge A\vdash\lozenge\B'}
        \DP
        \\
        \AIC{\pi}
        \noLine
        \UIC{\A\vdash \B}
        \RL{\delta_\wk}
        \UIC{\A,\delta A\vdash\B}
        \AIC{\mathcal{C}}
        \RL{\rmcutpar}
        \BIC{\A', \delta A\vdash \B'}
        \DP\qquad&\rightsquigarrow\qquad
        \AIC{\pi}
        \noLine
        \UIC{\A\vdash \B}
        \AIC{\mathcal{C}}
        \RL{\mathsf{mcut(\iota', \perp\!\!\!\perp')}}
        \BIC{\A'\vdash \B'}
        \RL{\delta_\wk}
        \UIC{\A', \delta A\vdash\B'}
        \DP
        \\
        \AIC{\pi}
        \noLine
        \UIC{\A,\delta A,\delta A\vdash\B}
        \RL{\delta_\cntr{}}
        \UIC{\A, \delta A\vdash\B}
        \AIC{\mathcal{C}}
        \RL{\rmcutpar}
        \BIC{\A', \delta A\vdash \B'}
        \DP\qquad&\rightsquigarrow\qquad
        \AIC{\pi}
        \noLine
        \UIC{\A, \delta A, \delta A\vdash\B}
        \AIC{\mathcal{C}}
        \RL{\mathsf{mcut(\iota', \perp\!\!\!\perp')}}
        \BIC{\A', \delta A, \delta A\vdash\B'}
        \RL{\delta_\cntr{}}
        \UIC{\A', \delta A\vdash\B'}
        \DP 
        \\
        \AIC{\pi}
        \noLine
        \UIC{\A\vdash  A, \B}
        \RL{\ocde}
        \UIC{\A, \oc A\vdash\B}
        \AIC{\mathcal{C}}
        \RL{\rmcutpar}
        \BIC{\A', \oc A\vdash\B'}
        \DP\qquad&\rightsquigarrow\qquad
        \AIC{\pi}
        \noLine
        \UIC{\A, A\vdash\B}
        \AIC{\mathcal{C}}
        \RL{\mathsf{mcut(\iota', \perp\!\!\!\perp')}}
        \BIC{\A', A\vdash\B'}
        \RL{\ocde}
        \UIC{\A', \oc A\vdash\B'}
        \DP
        \end{align*}
        \centering
        $\delta\in\{\oc, \Box\}$
        \caption{Second side of commutative cut-elimination steps of \muLLmodinf{}}\label{fig:app:mullmodinfexpcommcutstep2}
        \end{figure*}

\begin{figure*}
            \begin{align*}
            \hspace{-3.2cm} 
            {\small\AIC{\tikzmark{appcontrPrincipwnRed12}\C_{\A,\B}}
            \AIC{\pi}
            \noLine
            \UIC{\A\vdash \delta A, \delta A, \B}
            \RL{\delta_\cntr}
            \UIC{\tikzmark{appcontrPrincipwnRed12pl}\A\vdash \delta A, \B\tikzmark{appcontrPrincipwnRed12p}}
            \AIC{\tikzmark{appcontrPrincipwnRed22}\ocboxProofs_{\delta A}}
            \RL{\rmcutpar}
            \TIC{\tikzmark{appcontrPrincipwnRed21l}\oc\A_1, \tikzmark{appcontrPrincipwnRed31l}\Box\A_2, \tikzmark{appcontrPrincipwnRed11l} \A_3\vdash \tikzmark{appcontrPrincipwnRed11}\B_1, \tikzmark{appcontrPrincipwnRed21}\wn\B_2, \tikzmark{appcontrPrincipwnRed31}\lozenge\B_3}
            \DP}
            \quad\rightsquigarrow\quad &
            {\small\AIC{\tikzmark{appcontrPrincipwnRedpostred12}\C_{\A,\B}}
            \AIC{\pi}
            \noLine
            \UIC{\A\vdash \delta A, \delta A, \B\tikzmark{appcontrPrincipwnRedpostred12p}}
            \AIC{\tikzmark{appcontrPrincipwnRedpostred22}\ocboxProofs_{\delta A}}
            \AIC{\ocboxProofs_{\delta A}\tikzmark{appcontrPrincipwnRedpostred32}}
            \RL{\mathsf{mcut(\iota', \perp\!\!\!\perp')}}
            \QIC{\oc\A_1, \Box\A_2, \oc\A_1, \Box\A_2, \A_3\vdash\tikzmark{appcontrPrincipwnRedpostred11}\B_3, \quad \tikzmark{appcontrPrincipwnRedpostred211}\wn\B_1,\tikzmark{appcontrPrincipwnRedpostred21} \lozenge\B_2, \quad \tikzmark{appcontrPrincipwnRedpostred311}\wn\B_1, \tikzmark{appcontrPrincipwnRedpostred31}\lozenge\B_2}
            \doubleLine
            \RL{\wncontr, \occontr{}}
            \UIC{\oc\A_1, \Box\A_2, \Box\A_2, \A_3\vdash \B_3, \wn\A_1, \lozenge\A_2, \lozenge\A_2}
            \doubleLine
            \RL{\diacontr, \boxcontr}
            \UIC{\oc\A_1, \Box\A_2, \A_3\vdash \B_3, \wn\A_1, \lozenge\A_2}
            \DP} 
            \\
            \begin{tikzpicture}[overlay,remember picture,-,line cap=round,line width=0.1cm]
            \draw[rounded corners, smooth=2,blue, opacity=.4] ([xshift=1mm, yshift=1mm] pic cs:appcontrPrincipwnRed21) to ([xshift=2mm, yshift=1mm] pic cs:appcontrPrincipwnRed22);
            \draw[rounded corners, smooth=2,blue, opacity=.4] ([xshift=1mm, yshift=1mm] pic cs:appcontrPrincipwnRed21l) to ([xshift=2mm, yshift=1mm] pic cs:appcontrPrincipwnRed22);
            \draw[rounded corners, smooth=2,blue, opacity=.4] ([xshift=1mm, yshift=1mm] pic cs:appcontrPrincipwnRed31) to ([xshift=2mm, yshift=1mm] pic cs:appcontrPrincipwnRed22);
            \draw[rounded corners, smooth=2,blue, opacity=.4] ([xshift=1mm, yshift=1mm] pic cs:appcontrPrincipwnRed31l) to ([xshift=2mm, yshift=1mm] pic cs:appcontrPrincipwnRed22);
            \draw[rounded corners, smooth=2,red, opacity=.4] ([xshift=1mm, yshift=1mm] pic cs:appcontrPrincipwnRed11) to ([xshift=2mm, yshift=1mm] pic cs:appcontrPrincipwnRed12);
            \draw[rounded corners, smooth=2,red, opacity=.4] ([xshift=1mm, yshift=1mm] pic cs:appcontrPrincipwnRed11l) to ([xshift=2mm, yshift=1mm] pic cs:appcontrPrincipwnRed12);
            \draw[rounded corners, smooth=2,red, opacity=.4] ([xshift=1mm, yshift=1mm] pic cs:appcontrPrincipwnRed11) to ([xshift=-1mm, yshift=1mm] pic cs:appcontrPrincipwnRed12p);
            \draw[rounded corners, smooth=2,red, opacity=.4] ([xshift=1mm, yshift=1mm] pic cs:appcontrPrincipwnRed11l) to ([xshift=-1mm, yshift=1mm] pic cs:appcontrPrincipwnRed12p);
            \draw[rounded corners, smooth=2,red, opacity=.4] ([xshift=1mm, yshift=1mm] pic cs:appcontrPrincipwnRed11) to ([xshift=0mm, yshift=1mm] pic cs:appcontrPrincipwnRed12pl);
            \draw[rounded corners, smooth=2,red, opacity=.4] ([xshift=1mm, yshift=1mm] pic cs:appcontrPrincipwnRed11l) to ([xshift=0mm, yshift=1mm] pic cs:appcontrPrincipwnRed12pl);
         \end{tikzpicture}
            \hspace{-3.2cm}{\small\AIC{\C_{\A,\B}\tikzmark{appwkPrincipwnRedab}}
            \AIC{\A\vdash \B}
            \RL{\delta_\wk}
            \UIC{\tikzmark{appwkPrincipwnReda}\A\vdash \delta A, \B\tikzmark{appwkPrincipwnRedb}}
            \AIC{\ocboxProofs_{\delta A}\tikzmark{appwkPrincipwnRedocP}}
            \RL{\rmcutpar}
            \TIC{\tikzmark{appwkPrincipwnReda1}\oc\A_1, \tikzmark{appwkPrincipwnReda2}\Box\A_2, \tikzmark{appwkPrincipwnReda3}\A_3\vdash \B_1\tikzmark{appwkPrincipwnRedb1}, \wn\B_2\tikzmark{appwkPrincipwnRedb2}, \lozenge\B_3\tikzmark{appwkPrincipwnRedb3}}
            \DP}
            \quad&\rightsquigarrow\quad
            {\small\AIC{\C_{\A, \B}}
            \AIC{\A\vdash \B}
            \RL{\mathsf{mcut(\iota', \perp\!\!\!\perp')}}
            \BIC{\A_3\vdash\B_3}
            \RL{\wnwk, \ocwk}
            \doubleLine
            \UIC{\oc\A_1, \A_3\vdash \B_1, \wn\B_2}
            \doubleLine
            \RL{\diawk, \boxwk}
            \UIC{\oc\A_1, \Box\A_2, \A_3\vdash \B_1, \wn\B_2, \lozenge\B_3}
            \DP}
            \\
            \begin{tikzpicture}[overlay,remember picture,-,line cap=round,line width=0.1cm]
               \draw[rounded corners, smooth=2,blue, opacity=.4] ([xshift=1mm, yshift=1mm] pic cs:appwkPrincipwnReda1) to ([xshift=-1mm, yshift=1mm] pic cs:appwkPrincipwnRedocP);
               \draw[rounded corners, smooth=2,blue, opacity=.4] ([xshift=1mm, yshift=1mm] pic cs:appwkPrincipwnReda2) to ([xshift=-1mm, yshift=1mm] pic cs:appwkPrincipwnRedocP);
               \draw[rounded corners, smooth=2,blue, opacity=.4] ([xshift=-1mm, yshift=1mm] pic cs:appwkPrincipwnRedb2) to ([xshift=-1mm, yshift=1mm] pic cs:appwkPrincipwnRedocP);
               \draw[rounded corners, smooth=2,blue, opacity=.4] ([xshift=-1mm, yshift=1mm] pic cs:appwkPrincipwnRedb3) to ([xshift=-1mm, yshift=1mm] pic cs:appwkPrincipwnRedocP);
               \draw[rounded corners, smooth=2,red, opacity=.4] ([xshift=1mm, yshift=1mm] pic cs:appwkPrincipwnReda3) to ([xshift=-2mm, yshift=1mm] pic cs:appwkPrincipwnRedab);
               \draw[rounded corners, smooth=2,red, opacity=.4] ([xshift=-1mm, yshift=1mm] pic cs:appwkPrincipwnRedb1) to ([xshift=-2mm, yshift=1mm] pic cs:appwkPrincipwnRedab);
               \draw[rounded corners, smooth=2,red, opacity=.4] ([xshift=1mm, yshift=1mm] pic cs:appwkPrincipwnReda3) to ([xshift=1mm, yshift=1mm] pic cs:appwkPrincipwnReda);
               \draw[rounded corners, smooth=2,red, opacity=.4] ([xshift=-1mm, yshift=1mm] pic cs:appwkPrincipwnRedb1) to ([xshift=1mm, yshift=1mm] pic cs:appwkPrincipwnReda);
               \draw[rounded corners, smooth=2,red, opacity=.4] ([xshift=1mm, yshift=1mm] pic cs:appwkPrincipwnReda3) to ([xshift=-1mm, yshift=1mm] pic cs:appwkPrincipwnRedb);
               \draw[rounded corners, smooth=2,red, opacity=.4] ([xshift=-1mm, yshift=1mm] pic cs:appwkPrincipwnRedb1) to ([xshift=-1mm, yshift=1mm] pic cs:appwkPrincipwnRedb);
            \end{tikzpicture}
            \hspace{-3.2cm}{\small\AIC{\A_1\vdash A, \B_1}
            \RL{\wnde}
            \UIC{\A_1\vdash \wn A, \B_1}
            \AIC{\oc\A_2, \Box\A_3, A\vdash \wn\B_2, \lozenge\B_3}
            \RL{\wnpromloz}
            \UIC{\oc\A_2, \Box\A_3, \wn A\vdash \wn\B_2, \lozenge\B_3}
            \AIC{\C}
            \RL{\rmcutpar}
            \TIC{\A\vdash \B}
            \DP}
            \quad&\rightsquigarrow\quad
            {\small\AIC{\A_1\vdash A, \B_1}
            \AIC{\oc\A_2, \Box\A_3, A\vdash \wn\B_2, \lozenge\B_3}
            \AIC{\C}
            \RL{\mathsf{mcut(\iota', \perp\!\!\!\perp')}}
            \TIC{\A\vdash\B}
            \DP}
            \end{align*}
            \centering
            in all these proofs, $\delta\in\{\wn, \lozenge\}$
            \caption{First side of the principal cut-elimination steps of \muLLmodinf}\label{fig:appmullmodinfexpprincipcutstep1}
\end{figure*}

\begin{figure*}
    \begin{align*}
    \hspace{-3.2cm} 
    {\small\AIC{\tikzmark{appcontrPrincipocRed12}\C_{\A,\B}}
    \AIC{\pi}
    \noLine
    \UIC{\A, \delta A, \delta A\vdash \B}
    \RL{\delta_\cntr}
    \UIC{\tikzmark{appcontrPrincipocRed12pl}\A, \delta A\vdash \B\tikzmark{appcontrPrincipocRed12p}}
    \AIC{\tikzmark{appcontrPrincipocRed22}\ocboxProofs_{\delta A}}
    \RL{\rmcutpar}
    \TIC{\tikzmark{appcontrPrincipocRed21l}\oc\A_1, \tikzmark{appcontrPrincipocRed31l}\Box\A_2, \tikzmark{appcontrPrincipocRed11l} \A_3\vdash \tikzmark{appcontrPrincipocRed11}\B_1, \tikzmark{appcontrPrincipocRed21}\wn\B_2, \tikzmark{appcontrPrincipocRed31}\lozenge\B_3}
    \DP}
    \quad\rightsquigarrow\quad &
    {\small\AIC{\tikzmark{appcontrPrincipocRedpostred12}\C_{\A,\B}}
    \AIC{\pi}
    \noLine
    \UIC{\A, \delta A, \delta A\vdash\B\tikzmark{appcontrPrincipocRedpostred12p}}
    \AIC{\tikzmark{appcontrPrincipocRedpostred22}\ocboxProofs_{\delta A}}
    \AIC{\ocboxProofs_{\delta A}\tikzmark{appcontrPrincipocRedpostred32}}
    \RL{\mathsf{mcut(\iota', \perp\!\!\!\perp')}}
    \QIC{\oc\A_1, \Box\A_2, \oc\A_1, \Box\A_2, \A_3\vdash\tikzmark{appcontrPrincipocRedpostred11}\B_3, \quad \tikzmark{appcontrPrincipocRedpostred211}\wn\B_1,\tikzmark{appcontrPrincipocRedpostred21} \lozenge\B_2, \quad \tikzmark{appcontrPrincipocRedpostred311}\wn\B_1, \tikzmark{appcontrPrincipocRedpostred31}\lozenge\B_2}
    \doubleLine
    \RL{\wncontr, \occontr{}}
    \UIC{\oc\A_1, \Box\A_2, \Box\A_2, \A_3\vdash \B_3, \wn\A_1, \lozenge\A_2, \lozenge\A_2}
    \doubleLine
    \RL{\diacontr, \boxcontr}
    \UIC{\oc\A_1, \Box\A_2, \A_3\vdash \B_3, \wn\A_1, \lozenge\A_2}
    \DP} 
    \\
    \begin{tikzpicture}[overlay,remember picture,-,line cap=round,line width=0.1cm]
    \draw[rounded corners, smooth=2,blue, opacity=.4] ([xshift=1mm, yshift=1mm] pic cs:appcontrPrincipocRed21) to ([xshift=2mm, yshift=1mm] pic cs:appcontrPrincipocRed22);
    \draw[rounded corners, smooth=2,blue, opacity=.4] ([xshift=1mm, yshift=1mm] pic cs:appcontrPrincipocRed21l) to ([xshift=2mm, yshift=1mm] pic cs:appcontrPrincipocRed22);
    \draw[rounded corners, smooth=2,blue, opacity=.4] ([xshift=1mm, yshift=1mm] pic cs:appcontrPrincipocRed31) to ([xshift=2mm, yshift=1mm] pic cs:appcontrPrincipocRed22);
    \draw[rounded corners, smooth=2,blue, opacity=.4] ([xshift=1mm, yshift=1mm] pic cs:appcontrPrincipocRed31l) to ([xshift=2mm, yshift=1mm] pic cs:appcontrPrincipocRed22);
    \draw[rounded corners, smooth=2,red, opacity=.4] ([xshift=1mm, yshift=1mm] pic cs:appcontrPrincipocRed11) to ([xshift=2mm, yshift=1mm] pic cs:appcontrPrincipocRed12);
    \draw[rounded corners, smooth=2,red, opacity=.4] ([xshift=1mm, yshift=1mm] pic cs:appcontrPrincipocRed11l) to ([xshift=2mm, yshift=1mm] pic cs:appcontrPrincipocRed12);
    \draw[rounded corners, smooth=2,red, opacity=.4] ([xshift=1mm, yshift=1mm] pic cs:appcontrPrincipocRed11) to ([xshift=-1mm, yshift=1mm] pic cs:appcontrPrincipocRed12p);
    \draw[rounded corners, smooth=2,red, opacity=.4] ([xshift=1mm, yshift=1mm] pic cs:appcontrPrincipocRed11l) to ([xshift=-1mm, yshift=1mm] pic cs:appcontrPrincipocRed12p);
    \draw[rounded corners, smooth=2,red, opacity=.4] ([xshift=1mm, yshift=1mm] pic cs:appcontrPrincipocRed11) to ([xshift=0mm, yshift=1mm] pic cs:appcontrPrincipocRed12pl);
    \draw[rounded corners, smooth=2,red, opacity=.4] ([xshift=1mm, yshift=1mm] pic cs:appcontrPrincipocRed11l) to ([xshift=0mm, yshift=1mm] pic cs:appcontrPrincipocRed12pl);
 \end{tikzpicture}
    \hspace{-3.2cm}{\small\AIC{\C_{\A,\B}\tikzmark{appwkPrincipocRedab}}
    \AIC{\A\vdash \B}
    \RL{\delta_\wk}
    \UIC{\tikzmark{appwkPrincipocReda}\A, \delta A\vdash\B\tikzmark{appwkPrincipocRedb}}
    \AIC{\ocboxProofs_{\delta A}\tikzmark{appwkPrincipocRedocP}}
    \RL{\rmcutpar}
    \TIC{\tikzmark{appwkPrincipocReda1}\oc\A_1, \tikzmark{appwkPrincipocReda2}\Box\A_2, \tikzmark{appwkPrincipocReda3}\A_3\vdash \B_1\tikzmark{appwkPrincipocRedb1}, \wn\B_2\tikzmark{appwkPrincipocRedb2}, \lozenge\B_3\tikzmark{appwkPrincipocRedb3}}
    \DP}
    \quad&\rightsquigarrow\quad
    {\small\AIC{\C_{\A, \B}}
    \AIC{\A\vdash \B}
    \RL{\mathsf{mcut(\iota', \perp\!\!\!\perp')}}
    \BIC{\A_3\vdash\B_3}
    \RL{\wnwk, \ocwk}
    \doubleLine
    \UIC{\oc\A_1, \A_3\vdash \B_1, \wn\B_2}
    \doubleLine
    \RL{\diawk, \boxwk}
    \UIC{\oc\A_1, \Box\A_2, \A_3\vdash \B_1, \wn\B_2, \lozenge\B_3}
    \DP}
    \\
    \begin{tikzpicture}[overlay,remember picture,-,line cap=round,line width=0.1cm]
       \draw[rounded corners, smooth=2,blue, opacity=.4] ([xshift=1mm, yshift=1mm] pic cs:appwkPrincipocReda1) to ([xshift=-1mm, yshift=1mm] pic cs:appwkPrincipocRedocP);
       \draw[rounded corners, smooth=2,blue, opacity=.4] ([xshift=1mm, yshift=1mm] pic cs:appwkPrincipocReda2) to ([xshift=-1mm, yshift=1mm] pic cs:appwkPrincipocRedocP);
       \draw[rounded corners, smooth=2,blue, opacity=.4] ([xshift=-1mm, yshift=1mm] pic cs:appwkPrincipocRedb2) to ([xshift=-1mm, yshift=1mm] pic cs:appwkPrincipocRedocP);
       \draw[rounded corners, smooth=2,blue, opacity=.4] ([xshift=-1mm, yshift=1mm] pic cs:appwkPrincipocRedb3) to ([xshift=-1mm, yshift=1mm] pic cs:appwkPrincipocRedocP);
       \draw[rounded corners, smooth=2,red, opacity=.4] ([xshift=1mm, yshift=1mm] pic cs:appwkPrincipocReda3) to ([xshift=-2mm, yshift=1mm] pic cs:appwkPrincipocRedab);
       \draw[rounded corners, smooth=2,red, opacity=.4] ([xshift=-1mm, yshift=1mm] pic cs:appwkPrincipocRedb1) to ([xshift=-2mm, yshift=1mm] pic cs:appwkPrincipocRedab);
       \draw[rounded corners, smooth=2,red, opacity=.4] ([xshift=1mm, yshift=1mm] pic cs:appwkPrincipocReda3) to ([xshift=1mm, yshift=1mm] pic cs:appwkPrincipocReda);
       \draw[rounded corners, smooth=2,red, opacity=.4] ([xshift=-1mm, yshift=1mm] pic cs:appwkPrincipocRedb1) to ([xshift=1mm, yshift=1mm] pic cs:appwkPrincipocReda);
       \draw[rounded corners, smooth=2,red, opacity=.4] ([xshift=1mm, yshift=1mm] pic cs:appwkPrincipocReda3) to ([xshift=-1mm, yshift=1mm] pic cs:appwkPrincipocRedb);
       \draw[rounded corners, smooth=2,red, opacity=.4] ([xshift=-1mm, yshift=1mm] pic cs:appwkPrincipocRedb1) to ([xshift=-1mm, yshift=1mm] pic cs:appwkPrincipocRedb);
    \end{tikzpicture}
    \hspace{-3.2cm}{\small\AIC{\A_1, A\vdash\B_1}
    \RL{\ocde}
    \UIC{\A_1, \oc A\vdash \B_1}
    \AIC{\oc\A_2, \Box\A_3\vdash A, \wn\B_2, \lozenge\B_3}
    \RL{\ocpromloz}
    \UIC{\oc\A_2, \Box\A_3\vdash \oc A, \wn\B_2, \lozenge\B_3}
    \AIC{\C}
    \RL{\rmcutpar}
    \TIC{\A\vdash \B}
    \DP}
    \quad&\rightsquigarrow\quad
    {\small\AIC{\A_1\vdash A, \B_1}
    \AIC{\oc\A_2, \Box\A_3, A\vdash \wn\B_2, \lozenge\B_3}
    \AIC{\C}
    \RL{\mathsf{mcut(\iota', \perp\!\!\!\perp')}}
    \TIC{\A\vdash\B}
    \DP}
    \end{align*}
    \centering
    in all these proofs, $\delta\in\{\oc, \Box\}$
    \caption{Second side of the principal cut-elimination steps of \muLLmodinf}\label{fig:appmullmodinfexpprincipcutstep2}
\end{figure*}

\subsection{Details on $(-)^\circ$-translation}
\label{app:circTranslationDef}
\begin{defi}[Translation of \muLLmodinf{} into \muLLinf{}]
    Translation of formula is defined inductively on the formula:
    \begin{itemize}
        \item Translations of $\lozenge$ and $\Box$-formulas:
        \hfill $(\lozenge A)^{\circ} := \wn A^{\circ}\quad \text{and} \quad(\Box A)^{\circ} := \oc A^{\circ}. $
        \item Translations of atomic and unit formulas and variables $f$:
        \hfill $ f^\circ := f. $
        \item Translations of other non-fixed-point connectives: \hfill
        $ c(A_1, \dots, A_n)^\circ := c(A_1^\circ, \dots, A_n^\circ).$
        \item Translations of fixed-point connectives are given by:
        \hfill $ (\delta X. F)^\circ := \delta X. F^\circ$
        (with $\delta\in\{\mu,\nu\}$).
    \end{itemize}
    
    Translation of structural rules for modalities, $(\diacontr)$, $(\diawk)$, $(\boxcontr)$ and $(\boxwk)$ are respectively $(\wncontr)$, $(\wnwk)$, $(\occontr{})$ and $(\ocwk)$:
    $$
    \AIC{\A\vdash \B}
    \RL{\diawk}
    \UIC{\A\vdash \lozenge A, \B}
    \DP\quad\rightsquigarrow^\circ\quad
    \AIC{\A^\circ\vdash \B^\circ}
    \RL{\wnwk}
    \UIC{\A^\circ\vdash \wn A^\circ, \B^\circ}
    \DP
    $$
    $$
    \AIC{\A\vdash \lozenge A, \lozenge A, \B}
    \RL{\diacontr}
    \UIC{\A\vdash \lozenge A, \B}
    \DP\quad\rightsquigarrow^\circ\quad
    \AIC{\A^\circ\vdash \wn A^\circ, \wn A^\circ, \B^\circ}
    \RL{\wncontr}
    \UIC{\A^\circ\vdash \wn A^\circ, \B^\circ}
    \DP
    $$
    $$
    \AIC{\A\vdash \B}
    \RL{\boxwk}
    \UIC{\A, \Box A\vdash \B}
    \DP\quad\rightsquigarrow^\circ\quad
    \AIC{\A^\circ\vdash \B^\circ}
    \RL{\ocwk}
    \UIC{\A^\circ, \oc A^\circ\vdash \B^\circ}
    \DP
    $$
    $$
    \AIC{\A, \Box A, \Box A\vdash \B}
    \RL{\boxcontr}
    \UIC{\A, \Box A\vdash\B}
    \DP\quad\rightsquigarrow^\circ\quad
    \AIC{\A^\circ,\oc A^\circ, \oc A^\circ\vdash \B^\circ}
    \RL{\occontr{}}
    \UIC{\A^\circ,\oc A^\circ\vdash \B^\circ}
    \DP
    $$
    Translation of the modal rules and promotion rules are given by:
    $$
    \hspace{-2cm}
    \AIC{\A\vdash A, \B}
    \RL{\boxprom}
    \UIC{\Box\A\vdash \Box A, \lozenge\B}
    \DP\quad\rightsquigarrow\quad
    \AIC{\A^\circ\vdash A^{\circ}, \B^\circ}
    \doubleLine
    \RL{\ocde, \wnde}
    \UIC{\oc \A^\circ\vdash A^{\circ}, \wn\B^{\circ}}
    \RL{\ocprom}
    \UIC{\oc \A^\circ\vdash \oc A^{\circ}, \wn\B^{\circ}}
    \DP
    $$
    $$
    \AIC{\oc\A_1, \Box\A_2 \vdash A, \wn\B_1, \lozenge\B_2}
    \RL{\ocpromloz}
    \UIC{\oc\A_1, \Box\A_2\vdash \oc A, \wn\B_1, \lozenge\B_2}
    \DP\quad\rightsquigarrow\quad
    \AIC{\oc\A_1^\circ, \oc\A_2^\circ\vdash A^{\circ}, \wn \B_1^{\circ}, \wn\B_2^\circ}
    \RL{\ocprom}
    \UIC{\oc\A_1^\circ, \oc\A_2^\circ\vdash \oc A^{\circ}, \wn \B_1^{\circ}, \wn\B_2^\circ}
    \DP
    $$
    $$
    \hspace{-2cm}
    \AIC{\A, A\vdash\B}
    \RL{\diaprom}
    \UIC{\Box\A, \lozenge A\vdash\lozenge\B}
    \DP\quad\rightsquigarrow\quad
    \AIC{\A^\circ, A^{\circ}\vdash\B^\circ}
    \doubleLine
    \RL{\ocde, \wnde}
    \UIC{\oc \A^\circ, A^{\circ}\vdash\wn\B^{\circ}}
    \RL{\wnprom}
    \UIC{\oc \A^\circ, \wn A^{\circ}\vdash\wn\B^{\circ}}
    \DP
    $$
    $$
    \AIC{\oc\A_1, \Box\A_2, A \vdash\wn\B_1, \lozenge\B_2}
    \RL{\wnpromloz}
    \UIC{\oc\A_1, \Box\A_2,\wn A\vdash \wn\B_1, \lozenge\B_2}
    \DP\quad\rightsquigarrow\quad
    \AIC{\oc\A_1^\circ, \oc\A_2^\circ, A^{\circ}\vdash\wn \B_1^{\circ}, \wn\B_2^\circ}
    \RL{\wnprom}
    \UIC{\oc\A_1^\circ, \oc\A_2^\circ, \wn A^{\circ}\vdash\wn \B_1^{\circ}, \wn\B_2^\circ}
    \DP
    $$
    Translation of other inference rules $(r)$ are $(r)$ themselves.
    
    Translation of pre-proofs are defined co-inductively using translations of rules.
\end{defi}

\subsection{Proof of \Cref{mumodredSeqTranslationFiniteness}}
\label{app:mumodredSeqTranslationFiniteness}

\begin{lem}
    Consider a \muLLmodinf{} reduction step $\pi_0\rightsquigarrow\pi_1$, there exist a finite number of \muLLinf{} proofs $\theta_0, \dots, \theta_n$ such that:
    $$\pi_0^\circ = \theta_0 \redseq \theta_1\redseq\dots\redseq\theta_{n-1} \redseq \theta_n = \pi_1^{\circ}.$$
    \end{lem}
    \begin{proof}
        Reductions from the non-exponential part of $\muLLmodinf$ translates easily to one step of reduction in \muLLinf.
        The same is true for the exponential part except for the commutation of the modal rule. The translation of the left proof of it is of the form (we only do the case of $\diaprom$, $\boxprom$ is similar):
            $$
            \hspace{-2.5cm}
    \scalebox{.9}        {\small\AIC{\pi_1^\circ}
            \noLine
            \UIC{\oc\A_1^\circ  \vdash A_1^\circ, \B_1^\circ}
            \RL{\ocde, \wnde}
            \doubleLine
            \UIC{\oc\A_1^\circ\vdash A_1^\circ, \wn\B_1^\circ}
            \RL{\ocprom}
            \UIC{\oc\A_1^\circ\vdash \oc A_1^\circ, \wn\B_1^\circ}
            \AIC{\dots}
            \noLine
            \UIC{}
            \noLine
            \UIC{}
            \noLine
            \UIC{}
            \AIC{\pi_n^\circ}
            \noLine
            \UIC{\A_n^\circ\vdash A_n^\circ, \B_n^\circ}
            \RL{\ocde, \wnde}
            \doubleLine
            \UIC{\oc\A_n^\circ\vdash A_n^\circ, \wn\B_n^\circ}
            \RL{\ocprom}
            \UIC{\oc\A_n^\circ\vdash \oc A_n^\circ, \wn\B_n^\circ}
            \AIC{\pi_{n+1}^\circ}
            \noLine
            \UIC{\oc\A_{n+1}^\circ, A_{n+1}^\circ\vdash\B_{n+1}^\circ}
            \RL{\ocde, \wnde}
            \doubleLine
            \UIC{\oc\A_{n+1}^\circ, A_{n+1}^\circ\vdash\wn\B_{n+1}^\circ}
            \RL{\wnprom}
            \UIC{\oc\A_{n+1}^\circ, \wn A_{n+1}^\circ\vdash \wn\B_{n+1}^\circ}
            \AIC{\dots}
            \noLine
            \UIC{\hspace{-3cm}}
            \AIC{\pi_{n+m}^\circ}
            \noLine
            \UIC{\oc\A_{n+m}^\circ, A_{n+m}^\circ\vdash\B_{n+m}^\circ}
            \RL{\ocde, \wnde}
            \doubleLine
            \UIC{\oc\A_{n+m}^\circ, A_{n+m}^\circ\vdash\wn\B_{n+m}^\circ}
            \RL{\wnprom}
            \BIC{\oc\A_{n+m}^\circ, \wn A_{n+m}^\circ\vdash \wn\B_{n+m}^\circ}
            \RL{\rmcutpar}
            \QuinaryInfC{$\vdash \oc A^\circ, \wn\A^\circ$}
            \DP}
            $$
            We use a more general lemma:
            \begin{lem}
                Let $n\in\mathbb{N}$, let $d_1, \dots, d_n\in\mathbb{N}$ and let $p_1, \dots, p_n\in\{0, 1\}$. Let $\pi$ be a \muLLinf{}-proof concluded by an (\mcut)-rule, on top of which there is a list of $n$ proofs $\pi_1, \dots, \pi_n$. We ask for each $\pi_i$ to be of one of the following forms depending on $p_i$:
                \begin{itemize}
                \item If $p_i=1$, the $d_i+1$ last rules of $\pi_i$ are $d_i$ derelictions ($(\wnde)$ or $(\ocde)$) and then a promotion rule ($(\ocprom)$ or $(\wnprom)$). We ask for the principal formula of this promotion to be either a formula of the conclusion, or to be cut with a formula being principal in a proof $\pi_j$ on one of the last $d_j+p_j$ rules.
                
                \item If $p_i=0$, the $d_i$ last rules of $\pi_i$  are $d_i$ derelictions.
                \end{itemize}
                In each of these two cases, we ask for $\pi_i$ that each principal formulas of the $d_i$ derelictions to be either a formula of the conclusion of the multicut, either a cut-formula being cut with a formula appearing in $\pi_j$ such that $p_j=1$.
                We prove that $\pi$ reduces through a finite number of \mcut{}-reductions to a proof where each of the last $d_i+p_i$ rules either were eliminated by a $(\ocprom/\ocde)$-principal case, a $(\wnprom/\wnde)$-principal case or were commuted below the cut.
                \end{lem}
                \begin{proof}
                    We prove the property by induction on the sum of all the $d_i$ and of all the $p_i$:
                    \begin{itemize}
                    \item (Initialization). As the sum of the $d_i$ and $p_i$  is $0$, all $d_i$ and $p_i$ are equal to $0$, meaning that our statement is vacuously true.
                    
                    \item (Heredity). We have several cases:
                    \begin{itemize}
                    \item If the last rule of a proof $\pi_i$ is a promotion or a dereliction for which the principal formula is in the conclusion of the (\mcut), we do a commutation step on this rule obtaining $\pi'$. We apply our induction hypothesis on the proof ending with the (\mcut); and with parameters $d'_1, \dots, d'_n$ as well as $p'_1, \dots, p'_n$ and proofs $\pi'_1, \dots, \pi'_n$. To describe these parameters we have two cases:
                    \begin{itemize}
                    \item If the rule is a promotion. We take for each $j\in\llbracket 1, n\rrbracket$, $d'_j=d_j$; $p'_j = p_j$ if $j\neq i$, $p'_i=0$.
                    \item If the rule is a dereliction. We take for each $j\in\llbracket 1, n\rrbracket$, $d'_j=d_j$ if $j\neq i$, $d'_i=d_i-1$; $p'_j=p_j$.
                    \end{itemize}
                    In each cases, we take ; $\pi'_j=\pi_j$ if $j\neq i$ and $\pi'_i$ being the premise of the rule with conclusion $\pi_i$.
                    Note that $\sum d'_j +\sum p'_j =\sum d_j+ \sum p_j -1 < \sum d_j+ \sum p_j$ meaning that we can apply our induction hypothesis.
                    Combining our reduction step with the reduction steps of the induction hypothesis, we obtain the desired result.
                    
                    \item If there are no rules from the conclusion but that for an $i$ we have $d_i>0$ and $p_i=0$, meaning that the proof ends by a dereliction on a formula $F$. This means that there is proof $\pi_j$ such that $p_j=1$ and such that $F$ is cut with one of the formula of $\pi_j$. As $p_j=1$, $F$ is the principal formula of the last rule applied on $\pi_j$. We therefore can perform an promotion/dereliction principal case on the last rules from $\pi_i$ and $\pi_j$, leaving us with a proof $\pi'$ with an (\mcut) as conclusion. We apply the induction hypothesis on this proof with parameters $d'_1=d_1, \dots d'_i=d'_i-1 \dots, d'_n=d'_n$, $p'_1=p_1, \dots, p'_j=p'_j-1, \dots, p'_n=p_n$ and with the proofs being the hypotheses of the multicut.
                    Combining our steps with the steps from the induction hypotheses, we obtain the desired result.
                    
                    \item We will show that the case where there are no rules from the conclusion and that no $i$ is such that $d_i>0$ and $p_i=0$, is impossible by contradiction. We will construct an infinite sequence of proofs $(\theta_i)_{i\in\mathbb{N}}$ all different and all being hypotheses of the multi-cut.
                    As the sum of all $d_i$ and $p_i$ is not $0$, we have at least one dereliction or one promotion on a cut-formula right on top of the multicut. In both cases, we have a proof $\theta_0:=\pi_j$ ending with a promotion on a cut-formula $F$. This proof is in relation by the $\cutrel$-relation to another proof $\theta_1:=\pi_{j'}$. We know that this proof cannot be $\pi_j$ because the $\cutrel$-relation restricted to sequents is acyclic. As $d_{j'}>0$, we cannot have $p_{j'}=0$ by hypothesis. Therefore, this proof also ends with a promotion on a principal formula which is not from the conclusion. By repeating this process, we obtain the desired sequence $(\theta_i)_{i\in\mathbb{N}}$, giving us a contradiction.
                \end{itemize}
            \end{itemize}
            The statement is therefore true by induction
            \end{proof}   
            We apply this result on this proof with all the $p_i$ being equal to $1$ and with $d_i=\#(\B_i)+\#(\A_i)$.
            We can easily check that the conditions of the lemma are satisfied.
            Moreover, we notice that there will be only one promotion rule commuting under the cut and that it commutes before any dereliction, giving us the translation of the functorial promotion.
    \end{proof}

\section{Appendix on the \Cref{section:muLKmodinfcutElim}}
\label{app:muLKmodinfcutElim}

\subsection{Linear translation of rules}
\label{app:LinTranslation}

\begin{defi}[$\trans{(-)}$-translation of rules]
    The translation into \muLLmodinf{} of rules of \muLKmodinf{} are depicted in \Cref{fig:LinearRuleTranslationfixmod,fig:LinearRuleTranslationMult,fig:LinearRuleTranslationStruct}.
\end{defi}

\begin{figure*}
\begin{align*}
\AIC{\B\vdash F, \A}
\RL{\boxprom}
\UIC{\Box\B\vdash\Box F,\lozenge\A}
\DP\quad&\rightsquigarrow\quad
\AIC{\trans{\B}\vdash\wn\trans{F}, \wn\trans{\A}}
\doubleLine
\RL{\ocde, \wnprom}
\UIC{\oc\wn\trans{\B}\vdash\wn\trans{F}, \wn\trans{\A}}
\RL{\ocprom}
\UIC{\oc\wn\trans{\B}\vdash\oc\wn\trans{F}, \wn\trans{\A}}
\RL{\boxprom}
\UIC{\Box\oc\wn\trans{\B}\vdash\Box\oc\wn\trans{F}, \lozenge\wn\trans{\A}}
\doubleLine
\RL{\ocde}
\UIC{\oc\Box\oc\wn\trans{\B}\vdash\Box\oc\wn\trans{F}, \lozenge\wn\trans{\A}}
\doubleLine
\RL{\wnde, \ocpromloz}
\UIC{\oc\Box\oc\wn\trans{\B}\vdash\wn\oc\Box\oc\wn\trans{F}, \wn\oc\lozenge\wn\trans{\A}}
\DP\\
\AIC{\B, F\vdash \A}
\RL{\diaprom}
\UIC{\Box\B, \lozenge F\vdash\lozenge\A}
\DP\quad&\rightsquigarrow\quad
\AIC{\trans{\B}, \trans{F}\vdash \wn\trans{\A}}
\doubleLine
\RL{\ocde, \wnprom}
\UIC{\oc\wn\trans{\B}, \trans{F}\vdash \wn\trans{\A}}
\RL{\wnprom}
\UIC{\oc\wn\trans{\B}, \wn\trans{F}\vdash \wn\trans{\A}}
\RL{\diaprom}
\UIC{\Box\oc\wn\trans{\B}, \lozenge\wn\trans{F}\vdash \lozenge\wn\trans{\A}}
\doubleLine
\RL{\ocde}
\UIC{\oc\Box\oc\wn\trans{\B},\oc\lozenge\wn\trans{F}\vdash \lozenge\wn\trans{\A}}
\doubleLine
\RL{\wnde, \ocpromloz}
\UIC{\oc\Box\oc\wn\trans{\B}, \oc\lozenge\wn\trans{F}\vdash \wn\oc\lozenge\wn\trans{\A}}
\DP\\
\AIC{\pi}
\noLine
\UIC{\A\vdash F[X:=\delta X. F], \B}
\RL{\delta_r}
\UIC{\A\vdash \delta X. F, \B}
\DP
\quad&\rightsquigarrow\quad
\AIC{\trans{\pi}}
\noLine
\UIC{\trans{\A}\vdash \wn\trans{F}[X:=\delta X. \wn\trans{F}], \wn\trans{\B}}
\RL{\delta_r}
\UIC{\trans{\A}\vdash \delta X.\wn\trans{F}, \wn\trans{\B}}
\RL{\wnde, \ocprom}
\doubleLine
\UIC{\trans{\A}\vdash \wn\oc(\delta X.\wn\trans{F}), \wn\trans{\B}}
\DP\\
\AIC{\pi}
\noLine
\UIC{\A, F[X:=\delta X. F]\vdash \B}
\RL{\delta_l}
\UIC{\A, \delta X. F\vdash\B}
\DP
\quad&\rightsquigarrow\quad
\AIC{\trans{\pi}}
\noLine
\UIC{\trans{\A}, \trans{F}[X:=\delta X. \wn\trans{F}]\vdash \wn\trans{\B}}
\RL{\wnprom}
\UIC{\trans{\A}, \wn\trans{F}[X:=\delta X. \wn\trans{F}]\vdash \wn\trans{\B}}
\RL{\delta_l}
\UIC{\trans{\A}, \delta X.\wn\trans{F}\vdash \wn\trans{\B}}
\RL{\ocde}
\UIC{\trans{\A}, \oc(\delta X.\wn\trans{F})\vdash \wn\trans{\B}}
\DP
\end{align*}
\centering
with $\delta\in\{\mu, \nu\}$
\caption{Translation of the fixed-point and modal fragment of rules of \muLKmodinf{} into \muLLmodinf}\label{fig:LinearRuleTranslationfixmod}
\end{figure*}

\begin{figure*}
\begin{align*}
\AIC{\pi}
\noLine
\UIC{\A, F_1\vdash F_2, \B}
\RL{\rightarrow_r}
\UIC{\A\vdash F_1\rightarrow F_2, \B}
\DP\quad&\rightsquigarrow\quad
\AIC{\trans{\pi}}
\noLine
\UIC{\trans{\A}, \trans{F_1}\vdash \wn\trans{F_2}, \wn\trans{\B}}
\RL{\wnprom}
\UIC{\trans{\A}, \wn\trans{F_1}\vdash \wn\trans{F_2}, \wn\trans{\B}}
\RL{\multimap_r}
\UIC{\trans{\A}\vdash \wn\trans{F_1}\multimap \wn\trans{F_2}, \wn\trans{\B}}
\RL{\wnde, \ocprom}
\doubleLine
\UIC{\trans{\A}\vdash \wn\oc(\wn\trans{F_1}\multimap \wn\trans{F_2}), \wn\trans{\B}}
\DP\\
\AIC{\pi_1}
\noLine
\UIC{\A_1\vdash F_1, \B_1}
\AIC{\pi_2}
\noLine
\UIC{\A_2, F_2\vdash \B_2}
\RL{\rightarrow_r}
\BIC{\A_1, \A_2, F_1\rightarrow F_2\vdash\B_1, \B_2}
\DP\quad&\rightsquigarrow\quad
\AIC{\trans{\pi_1}}
\noLine
\UIC{\trans{\A_1}\vdash\wn\trans{F_1}, \wn\trans{\B_2}}
\AIC{\trans{\pi_2}}
\noLine
\UIC{\trans{\A_2}, \trans{F_2}\vdash \wn\trans{\B_2}}
\RL{\wnprom}
\UIC{\trans{\A_2}, \wn\trans{F_2}\vdash \wn\trans{\B_2}}
\RL{\multimap_l}
\BIC{\trans{\A_1}, \trans{\A_2}, \wn\trans{F_1}\multimap \wn\trans{F_2}\vdash\wn\trans{\B_1}, \wn\trans{\B_2}}
\RL{\ocde}
\UIC{\trans{\A_1}, \trans{\A_2}, \oc(\wn\trans{F_1}\multimap \wn\trans{F_2})\vdash\wn\trans{\B_1}, \wn\trans{\B_2}}
\DP\\
\AIC{\pi_1}
\noLine
\UIC{\A\vdash F_1, \B}
\AIC{\pi_2}
\noLine
\UIC{\A\vdash F_2, \B}
\RL{\land_r}
\BIC{\A\vdash F_1\land F_2, \B}
\DP
\quad&\rightsquigarrow\quad
\AIC{\trans{\pi_1}}
\noLine
\UIC{\trans{\A}\vdash \wn\trans{F_1}, \wn\trans{\B}}
\AIC{\trans{\pi_2}}
\noLine
\UIC{\trans{\A}\vdash \wn\trans{F_2}, \wn\trans{\B}}
\RL{\with_r}
\BIC{\trans{\A}\vdash \wn \trans{F_1}\with \wn\trans{F_2}, \wn\trans{\B}}
\RL{\wnde, \ocprom}
\doubleLine
\UIC{\trans{\A}\vdash \wn\oc(\wn \trans{F_1}\with \wn\trans{F_2}), \wn\trans{\B}}
\DP\\
\AIC{\pi}
\noLine
\UIC{\A, F_i\vdash\B}
\RL{\land_l^i}
\UIC{\A, F_1\land F_2\vdash \B}
\DP
\quad&\rightsquigarrow\quad
\AIC{\trans{\pi}}
\noLine
\UIC{\trans{\A}, \trans{F_i}\vdash \wn\trans{\B}}
\RL{\wnprom}
\UIC{\trans{\A}, \wn\trans{F_i}\vdash \wn\trans{\B}}
\RL{\with_l^i}
\UIC{\trans{\A}, \wn\trans{F_1}\with\wn\trans{F_2}\vdash \wn\trans{\B}}
\RL{\ocde}
\UIC{\trans{\A}, \oc(\wn\trans{F_1}\with\wn\trans{F_2})\vdash \wn\trans{\B}}
\DP\\
\AIC{\pi_1}
\noLine
\UIC{\A, F_1\vdash\B}
\AIC{\pi_2}
\noLine
\UIC{\A, F_2\vdash \B}
\RL{\lor_l}
\BIC{\A, F_1\lor F_2\vdash\B}
\DP
\quad&\rightsquigarrow\quad
\AIC{\trans{\pi_1}}
\noLine
\UIC{\trans{\A},\trans{F_1}\vdash\wn\trans{\B}}
\RL{\wnprom}
\UIC{\trans{\A},\wn\trans{F_1}\vdash\wn\trans{\B}}
\AIC{\trans{\pi_2}}
\noLine
\UIC{\trans{\A}, \trans{F_2}\vdash \wn\trans{\B}}
\RL{\wnprom}
\UIC{\trans{\A}, \wn\trans{F_2}\vdash \wn\trans{\B}}
\RL{\oplus_l}
\BIC{\trans{\A}, \wn\trans{F_1}\oplus \wn\trans{F_2}\vdash\wn\trans{\B}}
\RL{\ocde}
\UIC{\trans{\A}, \oc(\wn \trans{F_1}\oplus \wn\trans{F_2})\vdash \wn\trans{\B}}
\DP\\
\AIC{\pi}
\noLine
\UIC{\A\vdash F_i, \B}
\RL{\lor_r^i}
\UIC{\A\vdash F_1\lor F_2, \B}
\DP
\quad&\rightsquigarrow\quad
\AIC{\trans{\pi}}
\noLine
\UIC{\trans{\A}\vdash  \wn\trans{F_i}, \wn\trans{\B}}
\RL{\oplus_r^i}
\UIC{\trans{\A} \vdash \wn\trans{F_1}\oplus\wn\trans{F_2}, \wn\trans{\B}}
\RL{\wnde, \ocprom}
\doubleLine
\UIC{\trans{\A} \vdash \wn\oc(\wn\trans{F_1}\oplus\wn\trans{F_2}), \wn\trans{\B}}
\DP
\end{align*}
\caption{Translation of \LK{} connective-rules of \muLKmodinf{} into \muLLmodinf}\label{fig:LinearRuleTranslationMult}
\end{figure*}

\begin{figure*}
\begin{align*}
\AIC{}
\RL{\ax}
\UIC{F\vdash F}
\DP
\quad&\rightsquigarrow\quad
\AIC{}
\RL{\ax}
\UIC{\trans{F}\vdash \trans{F}}
\RL{\wnde}
\UIC{\trans{F}\vdash \wn\trans{F}}
\DP\\
\AIC{\pi_1}
\noLine
\UIC{\A_1\vdash F, \B_1}
\AIC{\pi_2}
\noLine
\UIC{\A_2, F\vdash\B_2}
\RL{\cut}
\BIC{\A_1, \A_2\vdash \B_1, \B_2}
\DP
\quad&\rightsquigarrow\quad
\AIC{\trans{\pi_1}}
\noLine
\UIC{\trans{\A_1}\vdash \wn \trans{F}, \wn\trans{\B_1}}
\AIC{\trans{\pi_2}}
\noLine
\UIC{\trans{\A_2}, \trans{F}\vdash \wn\trans{\B_2}}
\RL{\wnprom}
\UIC{\trans{\A_2}, \wn\trans{F}\vdash \wn\trans{\B_2}}
\RL{\cut}
\BIC{\trans{\A_1}, \trans{\A_2}\vdash \trans{\B_1}, \trans{\B_2}}
\DP\\
\AIC{}
\RL{\true_r}
\UIC{\A\vdash \true, \B}
\DP
\quad&\rightsquigarrow\quad
\AIC{}
\RL{\top_r}
\UIC{\trans{\A}\vdash \top, \wn\trans{\B}}
\RL{\wnde,\ocprom}
\doubleLine
\UIC{\trans{\A}\vdash \wn\oc\top, \wn\trans{\B}}
\DP\\
\AIC{}
\RL{\false_l}
\UIC{\A, \false\vdash \B}
\DP
\quad&\rightsquigarrow\quad
\AIC{}
\RL{0_l}
\UIC{\trans{\A}, 0\vdash\wn\trans{\B}}
\RL{\ocde}
\UIC{\trans{\A}, \oc 0\vdash\wn\trans{\B}}
\DP\\
\AIC{\pi}
\noLine
\UIC{\A_1, G, F, \A_2\vdash \B}
\RL{\exch_l}
\UIC{\A_1, F, G, \A_2\vdash\B}
\DP\quad&\rightsquigarrow\quad
\AIC{\trans{\pi}}
\noLine
\UIC{\trans{\A_1}, \trans{G}, \trans{F}, \trans{\A_2}\vdash \wn\trans{\B}}
\RL{\exch_l}
\UIC{\trans{\A_1}, \trans{F}, \trans{G}, \trans{\A_2} \vdash\wn\trans{\B}}
\DP\\
\AIC{\pi}
\noLine
\UIC{\A\vdash \B_1, G, F, \B_2}
\RL{\exch_r}
\UIC{\A\vdash\B_1, F, G, \B_2}
\DP\quad&\rightsquigarrow\quad
\AIC{\trans{\pi}}
\noLine
\UIC{\trans{\A}\vdash \wn\trans{\B_1}, \wn\trans{G}, \wn\trans{F}, \wn\trans{\B_2}}
\RL{\exch_r}
\UIC{\trans{\A}\vdash\wn\trans{\B_1}, \wn\trans{F}, \wn\trans{G}, \wn\trans{\B_2}}
\DP\\
\AIC{\pi}
\noLine
\UIC{\A\vdash\B}
\RL{\lkwk_l}
\UIC{\A, F\vdash \B}
\DP
\quad&\rightsquigarrow\quad
\AIC{\trans{\pi}}
\noLine
\UIC{\trans{\A}\vdash\wn\trans{\B}}
\RL{\ocwk}
\UIC{\trans{\A}, \trans{F}\vdash \wn\trans{\B}}
\DP\\
\AIC{\pi}
\noLine
\UIC{\A\vdash\B}
\RL{\lkwk_r}
\UIC{\A\vdash F, \B}
\DP
\quad&\rightsquigarrow\quad
\AIC{\trans{\pi}}
\noLine
\UIC{\trans{\A}\vdash\wn\trans{\B}}
\RL{\wnwk}
\UIC{\trans{\A} \vdash \wn\trans{F}, \wn\trans{\B}}
\DP\\
\AIC{\pi}
\noLine
\UIC{\A, F, F\vdash\B}
\RL{\lkcontr_l}
\UIC{\A, F\vdash \B}
\DP
\quad&\rightsquigarrow\quad
\AIC{\trans{\pi}}
\noLine
\UIC{\trans{\A}, \trans{F}, \trans{F}\vdash\wn\trans{\B}}
\RL{\occontr{}}
\UIC{\trans{\A}, \trans{F}\vdash \wn\trans{\B}}
\DP\\
\AIC{\pi}
\noLine
\UIC{\A\vdash F, F, \B}
\RL{\lkcontr_r}
\UIC{\A\vdash F, \B}
\DP
\quad&\rightsquigarrow\quad
\AIC{\trans{\pi}}
\noLine
\UIC{\trans{\A}\vdash \wn\trans{F}, \wn\trans{F}, \wn\trans{\B}}
\RL{\wncontr}
\UIC{\trans{\A} \vdash \wn\trans{F}, \wn\trans{\B}}
\DP\\
\AIC{\pi}
\noLine
\UIC{\A\vdash F, \B}
\RL{(-)^\perp_l}
\UIC{\A, F^\perp\vdash\B}
\DP
\quad&\rightsquigarrow\quad
\AIC{\trans{\pi}}
\noLine
\UIC{\trans{\A}\vdash \wn\trans{F}, \wn\trans{\B}}
\RL{(-)^\perp_l}
\UIC{\trans{\A}, (\wn\trans{F})^\perp \vdash \wn\trans{\B}}
\RL{\ocde}
\UIC{\trans{\A}, \oc((\wn\trans{F})^\perp)\vdash \wn\trans{\B}}
\DP\\
\AIC{\pi}
\noLine
\UIC{\A, F\vdash\B}
\RL{(-)^\perp_r}
\UIC{\A\vdash F^\perp, \B}
\DP
\quad&\rightsquigarrow\quad
\AIC{\trans{\pi}}
\noLine
\UIC{\trans{\A}, \trans{F} \vdash \wn\trans{\B}}
\RL{\wnprom}
\UIC{\trans{\A}, \wn\trans{F} \vdash \wn\trans{\B}}
\RL{(-)^\perp_r}
\UIC{\trans{\A} \vdash (\wn\trans{F})^\perp, \wn\trans{\B}}
\RL{\wnde, \ocprom}
\doubleLine
\UIC{\trans{\A}\vdash \wn\oc((\wn\trans{F})^\perp), \wn\trans{\B}}
\DP
\end{align*}
\caption{Translation of structural \& unit fragment of \muLKmodinf{} into \muLLmodinf}\label{fig:LinearRuleTranslationStruct}
\end{figure*}

\subsection{Details on the proof of \Cref{linearSkComp}}
\label{sec:app:linearSkComp}

\begin{lem}[Composition of $\sk{-}$ and of $\trans{(-)}$]
    \label{app:linearSkComp}
    Let $\pi$ be a \muLKmodinf{} pre-proof. We have that $\sk{\trans{\pi}}$ is equal to $\pi$.
    \end{lem}
    \begin{proof}
    This comes from the fact that $\trans{(-)}$-translation translates each rules $(r)$ of \muLKmodinf{} to a derivation containing the pre-image of $(r)$ by the translation $\sk{}$, adding only exponential rules. As exponential rules disappears from the proof by $\sk{}$, we get that $\sk{\trans{r}}$ is equal to $(r)$. We coinductively apply this result on pre-proofs.
    \end{proof}



\end{document}